\documentclass[journal]{IEEEtran}

\usepackage{placeins}
\usepackage{cuted}
\usepackage{latexsym}
\usepackage{amssymb}
\usepackage{bm}
\usepackage{stmaryrd}
\usepackage[dvips]{graphics}
\usepackage{graphicx}
\usepackage{psfrag}
\usepackage{amsfonts,amsmath,amssymb,hyperref, amsthm}
\usepackage{booktabs}
\usepackage{multicol}
\usepackage{algorithm}
\usepackage{algorithmic}
\usepackage{dsfont,color,subfigure,setspace, cite}
\usepackage{mathabx}
\usepackage{setspace}
\usepackage{textcomp}
\usepackage{authblk}
\usepackage{enumerate}

\newtheorem{definition}{Definition}

\newtheorem{theorem}{Theorem}
\newtheorem{lemma}{Lemma}

\newtheorem{remark}{Remark}
\newtheorem{proposition}{Proposition}
\newtheorem{example}{Example}

\newtheorem{finding}{Finding}

\begin{document}
\long\def\/*#1*/{}
\title{From Denoising to Compressed Sensing}

\author{Christopher~A.~Metzler,
~Arian~Maleki,~and~Richard~G.~Baraniuk
\thanks{C.\ Metzler and R.\ Baraniuk are with the Department of Electrical and Computer Engineering, Rice University, Houston, TX 77023 USA (e-mail: chris.metzler@rice.edu and richb@rice.edu).}%
\thanks{A.\ Maleki is with the Department of Statistics, Columbia University, New York, NY 10023 USA (e-mail: arian@stat.columbia.edu).}%
\thanks{The work of C.\ Metzler supported by the NSF GRF Program and the DoD NDSEG Program. The work of C.\ Metzler and R.\ Baraniuk was supported in part  and by the grants
NSF CCF1527501,
AFOSR FA9550-14-1-0088,
ARO W911NF-15-1-0316,
ONR N00014-12-1-0579, and
DoD HM04761510007. The work of A.\ Maleki was supported by the grant NSF CCF-1420328.}%
}

\date{} 

\maketitle
\begin{abstract}
A denoising algorithm seeks to remove noise, errors, or perturbations from a signal.
Extensive research has been devoted to this arena over the last several decades, and as a result, today’s denoisers can effectively remove large amounts of additive white Gaussian noise.
A compressed sensing (CS) reconstruction algorithm seeks to
recover a structured signal acquired using a small number of
randomized measurements. Typical CS reconstruction algorithms
can be cast as iteratively estimating a signal from a perturbed
observation. This paper answers a natural question: How can
one effectively employ a generic denoiser in a CS reconstruction
algorithm? In response, we develop an extension
of the approximate message passing (AMP) framework, called
Denoising-based AMP (D-AMP), that can integrate a wide class of denoisers within
its iterations. We demonstrate that, when used with a
high performance denoiser for natural images, D-AMP offers state-of-the-art CS
recovery performance while operating tens
of times faster than competing methods. We explain
the exceptional performance of D-AMP by analyzing some of its
theoretical features. A key element in D-AMP is the
use of an appropriate Onsager correction term in its iterations,
which coerces the signal perturbation at each iteration
to be very close to the white Gaussian noise that denoisers are
typically designed to remove.
\end{abstract}
\section{Introduction}\label{sec:Intro}

\subsection{Compressed sensing}\label{sec:CSoverview}

The fundamental challenge faced by a compressed sensing (CS) reconstruction algorithm is to reconstruct a high-dimensional signal from a small number of measurements. The process of taking compressive measurements can be thought of as a linear mapping of a length $n$ signal vector $x_o$ to a length $m$, $m\ll n$, measurement vector $y$.  Because this process is linear, it can be modeled by a measurement matrix $\mathbf{\Phi} \in \mathbb{C}^{m \times n}$.  The matrix $\mathbf{\Phi}$ can take on a variety of physical interpretations:  In a compressively sampled MRI, $\mathbf{\Phi}$ might be sampled rows of an $n \times n$ Fourier matrix \cite{CandesRobust, LustigMRI1}.  In a single pixel camera, $\mathbf{\Phi}$ might be a sequence of 1s and 0s representing the modulation of a micromirror array \cite{singlepixelcam}.

Oftentimes a signal $x_o$ is sparse (or approximately sparse) in some transform domain, i.e., $x_o=\mathbf{\Psi} u$ with sparse $u$, where $\mathbf{\Psi}$ represents the inverse transform matrix.  In this case we lump the measurement and transformation into a single measurement matrix $\mathbf{A}=\mathbf{\Phi} \mathbf{\Psi}$.  When a sparsifying basis is not used $\mathbf{A}=\mathbf{\Phi}$.  Future references to the measurement matrix refer to $\mathbf{A}$.  

The compressed sensing reconstruction problem is to determine which signal $x_o$ produced $y$ when sampled according to $y=\mathbf{A}x_o+w$ where $w$ represents measurement noise.   Because $\mathbf{A} \in \mathbb{R}^{m \times n}$ and $m \ll n$, the problem is severely under-determined.\footnote{Note that for notational simplicity in our current derivations and algorithms we restrict $\mathbf{A}$ to be in $\mathbb{R}^{m \times n}$. However, an extension to $\mathbb{C}^{m \times n}$ is also possible.}  Therefore to recover $x_o$ one first assumes that $x_o$ possesses a certain structure and then searches, among all the vectors $x$ that satisfy $y \approx \mathbf{A}x$, for one that also exhibits the given structure. In case of sparse $x_o$, one recovery method is to solve the convex problem
\begin{equation} \label{eqn:Basis_Pursuit_Problem}
 \underset{x}{\text{minimize}} \ 
\|x\|_1 \ \ \ \ 
 \text{subject to} \ \  \ 
\|y-\mathbf{A}x\|_2^2 \le\lambda,
\end{equation}
which is known formally as basis pursuit denoising (BPDN).  It was first shown in \cite{CaTa05, Donoho1} that if $x_o$ is sufficiently sparse and  $\mathbf{A}$ satisfies certain properties, then \eqref{eqn:Basis_Pursuit_Problem} can accurately recover $x_o$.

The initial work in CS solved \eqref{eqn:Basis_Pursuit_Problem} using convex programming methods.  However, when dealing with large signals, such as images, these convex programs are extremely computationally demanding. Therefore, lower cost iterative algorithms were developed; including matching pursuit\cite{MallatMP}, orthogonal matching pursuit \cite{TroppOMP}, iterative hard-thresholding \cite{BlumensathIHT}, compressive sampling matching pursuit\cite{TroppNeedellCoSaMP}, approximate message passing\cite{DoMaMo09}, and iterative soft-thresholding \cite{FiNoWr07, BeFr08, CoWa05, ElMaShZi07, YiOsGoDa08, DaDeDe04}, to name just a few. See \cite{MaDo09sp, yang2010fast} for a complete set of references.
 
Iterative thresholding (IT) algorithms generally take the form
\begin{equation}
\label{eqn:IT_alg}
\begin{array}{lcl}
x^{t+1}&=&\eta_\tau (\mathbf{A}^* z^t + x^t),\\z^t &=& y - \mathbf{A} x^t,
\end{array}
\end{equation}
where $\eta_{\tau} (y)$ is a shrinkage/thresholding non-linearity, $x^t$ is the estimate of $x_o$ at iteration $t$, and $z^t$ denotes the estimate of the residual $y-Ax_o$ at iteration $t$. When $\eta_{\tau} (y) = (|y|- \tau)_+ {\rm sign}(y)$ the algorithm is known as iterative soft-thresholding (IST). 

AMP extends iterative soft-thresholding by adding an extra term to the residual known as the {\em Onsager correction term}:
\begin{equation}
\label{eqn:AMP_alg}
\begin{array}{lcl}
x^{t+1}&=&\eta_\tau (\mathbf{A}^* z^t + x^t),\\z^t &=& y - \mathbf{A} x^t + \frac{1}{\delta}z^{t-1}\langle\eta'_\tau ( \mathbf{A}^* z^{t-1} + x^{t-1})\rangle.
\end{array}
\end{equation}
Here, $\delta=m/n$ is a measure of the under-determinacy of the problem, $\langle \cdot \rangle$ denotes the average of a vector, and $\frac{1}{\delta}\langle\eta'_\tau ( \mathbf{A}^* z^{t-1} + x^{t-1})\rangle$, where $\eta'_\tau$ represents the derivative of $\eta_\tau$, is the Onsager correction term. The role of this term is illustrated in Figure \ref{fig:QQplotAMPIST}. This figure compares the QQplot\footnote{A QQplot is a visual inspection tool for checking the Gaussianity of the data. In a QQplot, deviation from a straight line is an evidence of non-Gaussianity. } of $x^t+\mathbf{A}^*z^t -x_o$ for IST and AMP. We call this quantity the {\em effective noise} of the algorithm at iteration $t$. As is clear from the figure, the QQplot of the effective noise in AMP is a straight line. This means that  the noise  is approximately Gaussian.  This important feature enables the accurate analysis of the algorithm \cite{DoMaMo09, DoMaMoNSPT}, the optimal tuning of the parameters \cite{MousaviMB13a}, and leads to the linear convergence of $x^t$ to the final solution \cite{MalekiThesis}. We will employ this important feature of AMP in our work as well.

\begin{figure}[t]
\begin{center}
\hspace*{-.75cm}
  \includegraphics[width=.55\textwidth]{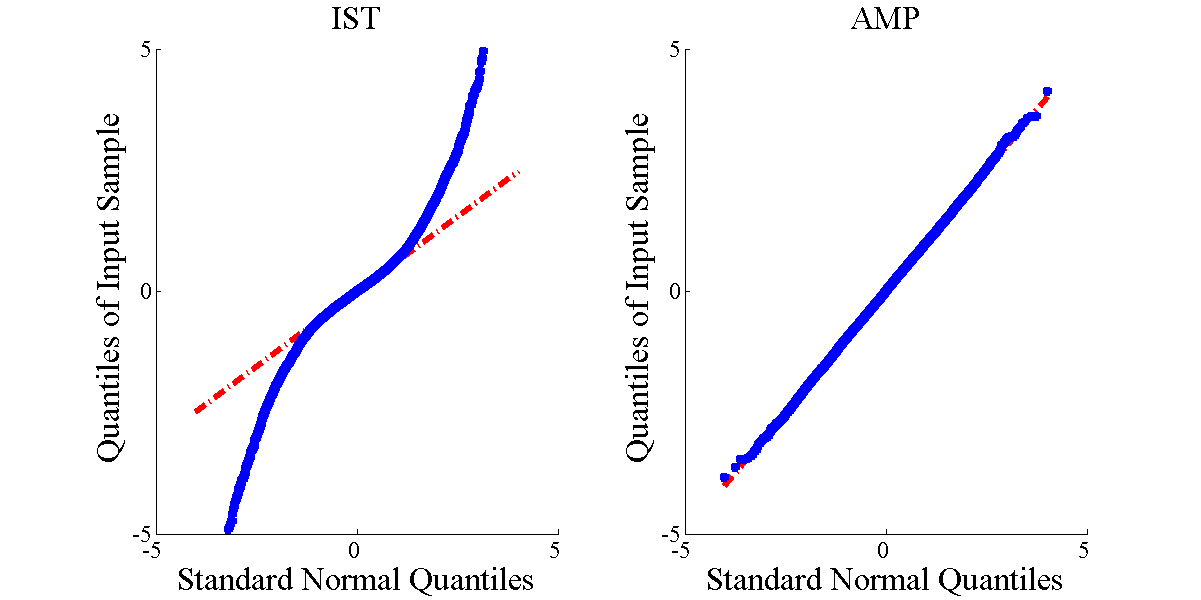} 
  \caption{QQplot comparing the distributions of the effective noise of the IST and AMP algorithms at iteration 5 while reconstructing a 50\% sampled Barbara test image.  Notice the heavy tailed distribution of IST.  AMP remains Gaussian because of the Onsager correction term.}
   \label{fig:QQplotAMPIST}
  \end{center}
\end{figure}

\subsection{Main contributions}
A sparsity model is accurate for many signals and has been the focus of the majority of CS research. Unfortunately, sparsity-based methods are less appropriate for many imaging applications. The reason for this failure is that natural images do not have an exactly sparse representation in any known basis (DCT, wavelet, curvelet, etc.).  Figure \ref{fig:BarbaraCoefs} shows the wavelet coefficients of the classic signal processing image Barbara.  The majority of the coefficients are non-zero and many are far from zero.  As a result, algorithms that seek only wavelet-sparsity fail to recover the signal.

\begin{figure}[t]
\begin{center}
  \includegraphics[width=.5\textwidth, height= .4 \textwidth]{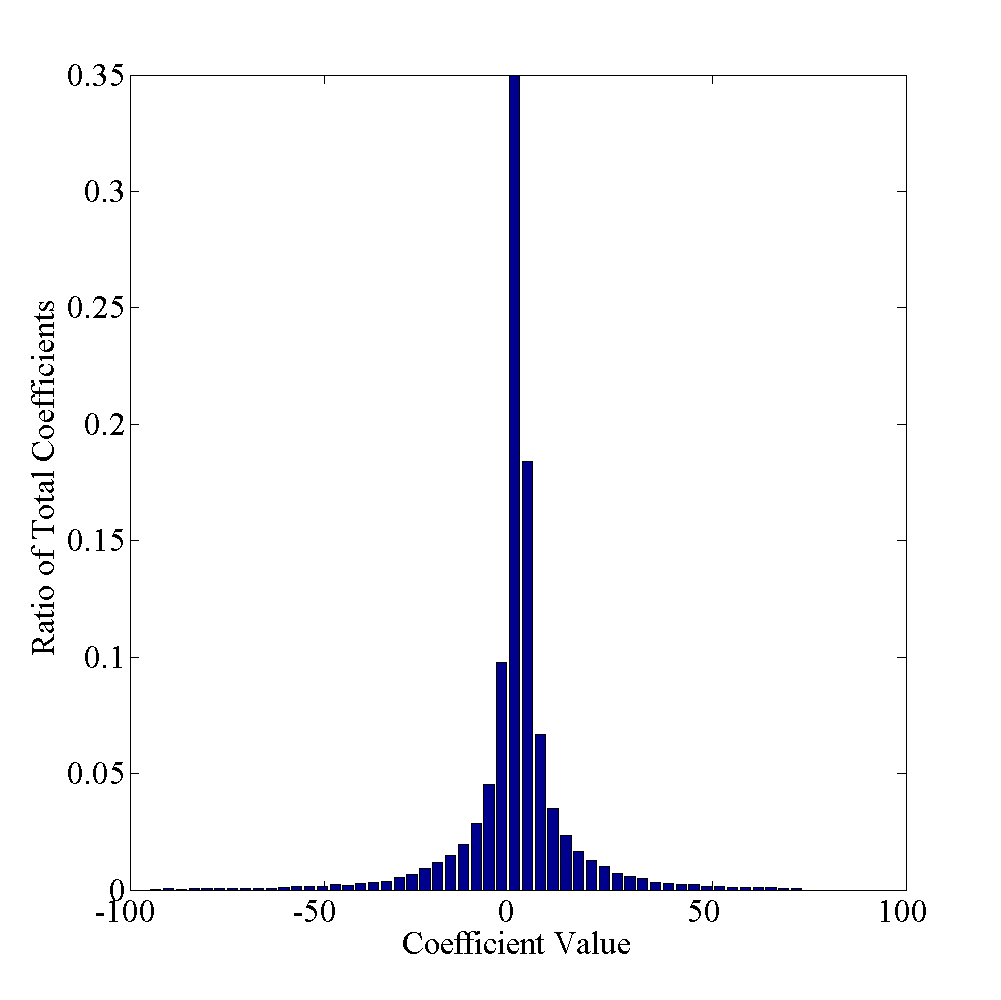} 
  \vspace{-.75cm}
  \caption{Histogram of the Daubechies 4 wavelet coefficients of the Barbara test image.  Notice the non-sparse distribution of the coefficients.  Sparsity-based compressed sensing algorithms fail because of this distribution.}
   \label{fig:BarbaraCoefs}
   \end{center}
\end{figure}

In response to this failure, researchers have considered more elaborate structures for CS recovery.  These include minimal total variation \cite{CandesRobust, BlockTV}, block sparsity \cite{eldar2010block}, wavelet tree sparsity \cite{RichModelBasedCS, HegdeAproxTree}, hidden Markov mixture models \cite{DuarteMarkovCS, YookyungKim1, som2012compressive}, non-local self-similarity \cite{PeyreNonLocal1, JianZhangImprovedTV, dongcompressive}, 
and simple representations in adaptive bases \cite{FowlerMH, JianZhangALSB}. Many of these approaches have led to significant improvements in imaging tasks.

In this paper, we take a complementary approach to enhancing the performance of CS recovery of non-sparse signals \cite{DAMP599}. Rather than focusing on developing new signal models, we demonstrate how the existing rich literature on {\em signal denoising} can be leveraged for enhanced CS recovery.\footnote{In this paper, denoising refers to any algorithm that receives $x_o+ \sigma z $, where $\sigma z \sim N(0,  \sigma^2 I)$ denotes the noise, as its input and returns an estimate of $x_o$ as its output. Refer to Sections \ref{sec:denoiseprop} and \ref{sec:reviewdenoise} for more information on denoisers. } The idea is simple: Signal denoising algorithms (whether based on an explicit or implicit model) have been developed and optimized for decades. Hence, any CS recovery scheme that employs such denoising algorithms should be able to capture complicated structures that have heretofore not been captured by existing CS recovery schemes. 

The approximate message passing algorithm (AMP) \cite{MalekiThesis, DonohoJM13} presents a natural way to employ denoising algorithms for CS recovery. We call the AMP that employs denoiser D D-AMP. D-AMP assumes that $x_o$ belongs to a class of signals $C \subset \mathbb{R}^n$, such as the class of natural images of a certain size, for which a family of denoisers $\{D_{\sigma} \ : \ \sigma>0\}$ exists. Each denoiser $D_{\sigma}$ can be applied to $x_o + \sigma z$ with $z \sim N(0,I)$ and will return an estimate of $x_o$ that is hopefully closer to $x_o$ than $x_o+ \sigma z$. These denoisers may employ simple structures such as sparsity or much more complicated structures, which we will discuss in Section \ref{sec:reviewdenoise}. In this paper we treat each denoiser as a black box; it receives a signal plus Gaussian noise and returns an estimate of $x_o$. Hence, we do not assume any knowledge of the signal structure/information the denoising algorithm is employing to achieve its goal. 
This makes our derivations applicable to a wide variety of signal classes and a wide variety of denoisers. 

D-AMP has several advantages over existing CS recovery algorithms:  (i) It can be easily applied to many different signal classes. (ii) It outperforms existing algorithms and is extremely robust to measurement noise (our simulation results are summarized in Section \ref{sec:DAMPinPractice}). (iii) It comes with an analysis framework that not only characterizes its fundamental limits, but also suggests how we can best use the framework in practice.

\begin{figure*}[t]
	\begin{center}
		\includegraphics[width=\textwidth]{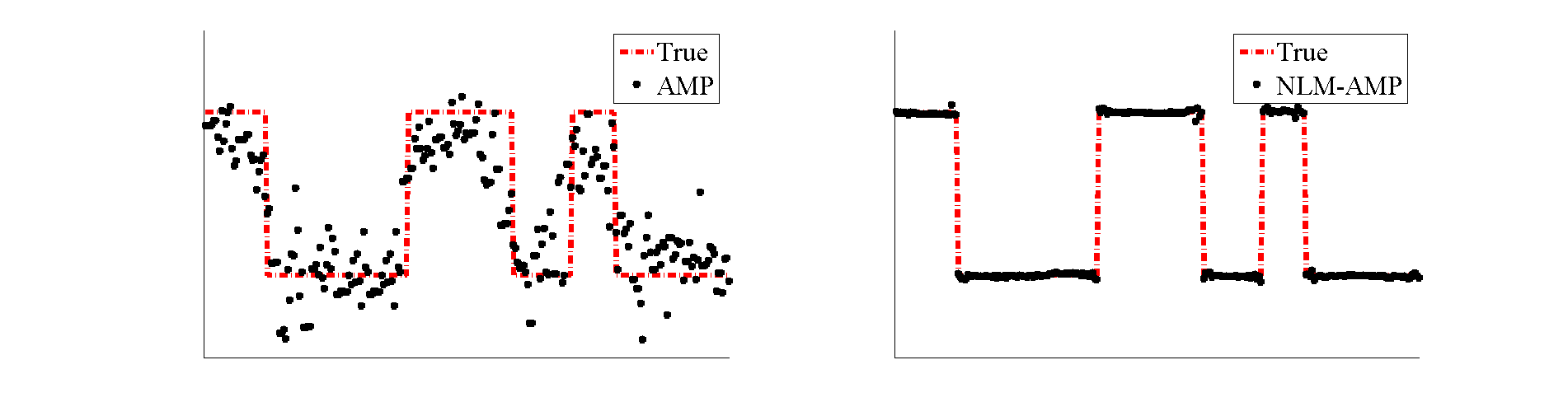} 
		\vspace{-0.75cm}
		\caption{Reconstructions of a piecewise constant signal that was sampled at a rate of $\delta=1/3$.  Notice that NLM-AMP successfully reconstructs the piecewise constant signal whereas AMP, which is based on wavelet thresholding, does not.}
		\label{fig:1D_Constant} 
	\end{center}
\end{figure*}

D-AMP employs a denoiser in the following iteration:
\begin{eqnarray}\label{eqn:DAMP}
x^{t+1} &=& D_{\hat{\sigma}^t} (x^t+\mathbf{A}^*z^t), \nonumber \\
z^t&=& y-\mathbf{A}x^t+ z^{t-1} {\rm div} D_{\hat{\sigma}^{t-1}}(x^{t-1}+\mathbf{A}^*z^{t-1}) /m, \nonumber\\
(\hat{\sigma}^t)^2&=& \frac{\|z^t\|_2^2}{m}. 
\end{eqnarray}
Here, $x^t$ is the estimate of $x_o$ at iteration $t$ and $z^t$ is an estimate of the residual. As we will show later, $x^t+\mathbf{A}^*z^t$ can be written as $x_o+ v^t$, where $v^t$ can be considered as i.i.d. Gaussian noise.\footnote{This conjecture has been validated empirically elsewhere \cite{MalekiThesis, MaAnYaBa11, DonohoJM13} for simpler denoisers. For known $\sigma_t$ the conjecture has been proven for scalar denoisers in \cite{bayati2011dynamics}. By combining the proof of \cite{bayati2011dynamics} with the proof technique developed in \cite{kamilov2012approximate} we can prove the above conjecture for the scalar denoisers. Since scalar denoisers are not of our main concern in this paper, we do not include a proof here. We will present empirical evidence that it holds for many of the state-of-the-art image denoising algorithms.\label{foot:SEProofdisccusion}}  $\hat{\sigma}^t$ is an estimate of the standard deviation of that noise. ${\rm div} D_{\hat{\sigma}^{t-1}}$ denotes the divergence of the denoiser.\footnote{In the context of this work the divergence ${\rm div} D(x)$ is simply the sum of the partial derivatives with respect to each element of $x$, i.e., ${\rm div} D(x)=\sum\limits_{i=1}^{n} \frac{\partial D(x)}{\partial x_i}$, where $x_i$ is the $i^{\rm th}$ element of $x$.} The term  $z^{t-1}{\rm div} D_{\hat{\sigma}^{t-1}}(x^{t-1}+\mathbf{A}^*z^{t-1}) /m$ is the Onsager correction term. We will show later that this term has a major impact on the performance of the algorithm. The explicit calculation of this term is not always straightforward: many popular denoisers do not have explicit formulations. However, we will show that it may be approximately calculated without requiring the explicit form of the denoiser.

D-AMP applies an existing denoising algorithm to vectors that are generated from compressive measurements. The intuition is that at every iteration D-AMP obtains a better estimate of $x_o$ and that this sequence of estimates eventually converges to $x_o$. 

To predict the performance of D-AMP, we will employ a novel state evolution framework to theoretically track the standard deviation of the noise, $\hat{\sigma}_t$ at each iteration of D-AMP. Our framework extends and validates the state evolution framework proposed in \cite{DoMaMo09, MalekiThesis}. Through extensive simulations we show that in high-dimensional settings (for the subset of denoisers that we consider in this paper) our state evolution predicts the mean square error (MSE) of D-AMP accurately. Based on the state evolution we characterize the performance of  D-AMP and connect the number of measurements D-AMP requires to the performance of the denoiser. We also employ the state evolution to address practical concerns such as the tuning of the parameters of denoisers and the sensitivity of the algorithm to measurement noise. Furthermore, we use the state evolution to explore the optimality of D-AMP. We postpone a detailed discussion to Section \ref{sec:Theory}. 

Figure \ref{fig:1D_Constant} compares the performance of the original AMP (which employs sparsity in the wavelet domain) with that of D-AMP based on the non-local means denoising algorithms \cite{Buades05areview}, called NLM-AMP here. Since NLM is a better denoiser than wavelet thresholding for piecewise constant functions, NLM-AMP dramatically outperforms the original AMP. The details of our simulations are given in Section \ref{sec:one dim}.

\subsection{Related work}

\subsubsection{Approximate message passing and extensions}\label{sec:AMPextensions}
In the last five years, message passing and approximate message passing algorithms have been the subject of extensive research in the field of compressed sensing \cite{MalekiThesis, DoMaMo09, DoMaMoNSPT, MaAnYaBa11, DoMaMo2011, Schniter2010, som2012compressive, kudekar2010effect, rangan2011generalized,parker2013bilinear, krzakala2012statistical, donoho2012information, MalekiAmpAnalysis, bayati2011dynamics, barbier2012compressed, reeves2013minimax, kamilov2012approximate, bayati2013estimating, borgerding2013generalized , cakmak2014s, kabashima2014signal }. Most previously published papers consider a Bayesian framework in which a signal prior $p_{x}$ is defined on the class of signals $C$ to which $x_o$ belongs.  
Message passing algorithms are then considered as heuristic approaches of calculating the posterior mean, $\mathbb{E} (x_o \ | \ y, A)$. Message passing has been simplified to approximate message passing (AMP) by employing the high dimensionality of the data \cite{DoMaMo09}. The state evolution framework has been proposed as a way of analyzing the AMP algorithm \cite{DoMaMo09}.  The main difference between this line of work and our work is that we do not assume any signal prior on the signal space $C$. This distinction introduces a difference between the state evolution framework we develop in this paper and the ones that have been developed elsewhere. We will highlight the connection between these two different state evolutions in Section \ref{sec:connectionSE}.


Note that in the development of D-AMP we are not concerned about whether the algorithm is approximating a posterior distribution for a certain prior or not.  Nor are we concerned about whether or not the denoisers used within D-AMP's iterations are tied to any prior.
Instead, we rely on only one important feature of AMP---that $x^t+\mathbf{A}^*z^t-x_o$ behaves similar to i.i.d.\ Gaussian noise. Our analysis is based on this assumption.  We validate this assumption with extensive simulations that are presented in Section \ref{sec:seandgaussianitycheck}.


Donoho et al.  \cite{DonohoJM13} also extended the AMP framework based upon the fact that $x^t+\mathbf{A}^*z^t-x_o$ behaves similar to i.i.d.\ Gaussian noise. In their framework the denoiser D can be any scale-invariant function. There are several major differences between our work and theirs: (i) We do not impose scale-invariance on the denoiser, because this assumption does not hold for many practical denoisers.  (ii) We present a far broader validation of our method and state evolution: The empirical validation  \cite{DonohoJM13} presented is concerned with very specific simple denoisers and has remained at the level of maximin phase transition \cite{MalekiThesis}. Likewise, the state evolution they employed in their empirical validation is based on the Bayesian framework described above and the validations are restricted to simple distributions. In this paper we consider a deterministic version of the state evolution and for the first time present evidence that such a state evolution can in fact predict the performance of D-AMP. Note that the evidence we present goes far beyond the match in the maximin phase transition.  
This development is important because the maximin framework employed in \cite{DonohoJM13} is not useful in most practical applications that deal with naturally occurring signals. (iii) We present a signal-dependent parameter tuning strategy for AMP and show that our deterministic state evolution can cope with those situations as well. (iv) We show how practical denoisers whose explicit functional form is not given can be employed in AMP. (v) We investigate the optimality of D-AMP as a means to employ different denoisers in the AMP algorithm.

While writing this paper, we became aware of another relevant paper about extensions to the AMP algorithm\cite{DrorDenoiserAMP}. In this work, the authors employ AMP with scalar denoisers that are better adapted to the statistics of  natural images. By doing so, they have obtained a major improvement over existing algorithms. In this paper, we consider a much broader class of denoisers. We not only show how the AMP algorithm can be adapted to such denoisers; we also explore the theoretical properties of our recovery algorithms.\\

\subsubsection{Model-based CS imaging}
Many researchers have noticed the weakness of sparsity-based methods for imaging applications and have therefore explored the use of more complicated signal models. These models can be enforced explicitly, by constraining the solution space, or implicitly, by using penalty functionals to encourage solutions of a certain form.  

Initially these model-based methods were restricted to simple concepts like minimal total variation \cite{BlockTV} and block sparsity \cite{eldar2010block}, but they have since been extended to structures such as wavelet trees \cite{RichModelBasedCS, HegdeAproxTree} and mixture models \cite{DuarteMarkovCS, YookyungKim1, som2012compressive}. Furthermore, some researchers have employed more complicated signal models through non-local regularization \cite{PeyreNonLocal1, JianZhangImprovedTV, dongcompressive} and the use of adaptive over-complete dictionaries \cite{FowlerMH, JianZhangALSB}. A non-local regularization method, NLR-CS \cite{dongcompressive}, represents the current state-of-the-art in CS recovery. 
Through the use of denoisers, rather than explicit models or penalty functionals, our algorithm outperforms these methods on standard test images. 

An additional reconstruction algorithm does not fit into any of the above categories but in many ways relates closely to our own.  Egiazarian et al.  developed a denoising-based CS recovery algorithm \cite{EgiazarianCSRecon} that uses the same research group's BM3D denoising algorithm \cite{DabovBM3D} to impose a non-parametric model on the reconstructed signal.  This method solves the CS problem when the measurement matrix is a subsampled DFT matrix.  The method iteratively adds noise to the missing part of the spectra and then applies BM3D to the result. In \cite{dongcompressive} it was shown that the BM3D-based algorithm performed considerably worse than NLR-CS.  Therefore it is not tested here.

Finally we should emphasize another major difference between our work and other approaches designed for imaging applications. D-AMP comes with an accurate 
analysis that explains the behavior of the algorithm, its optimality properties, and its limitations. Such an accurate analysis does not exist for other methods.

\subsection{Structure of the paper}
 The remainder of this paper is structured as follows: Section \ref{sec:DAMP} introduces our D-AMP algorithm and some of its main features. Section \ref{sec:Theory} is devoted to the theoretical analysis of D-AMP and its optimality properties. Section \ref{sec:connectionSE}  
 establishes a connection between our state evolution and existing state evolutions. Section \ref{sec:onsager} explains two different approaches to calculating the Onsager correction term.  
Section \ref{sec:Smooth} explains how to smooth poorly behaved denoisers so that they can be used within our framework.
 Section \ref{sec:DAMPinPractice} summarizes our main simulation results: it provides evidence on the validity of our state evolution framework; it provides a detailed guideline on setting and tuning of different parameters of the algorithms; it compares the performance of our D-AMP algorithm with the state-of-the-art algorithms in compressive imaging.


\section{Denoising-based approximate message passing}\label{sec:DAMP}

Consider a family of denoising algorithms $D_{\sigma}$ for a class of signals $\mathcal{C}$. Our goal is to employ these denoisers to obtain a good estimate of $x_o \in \mathcal{C}$ from $y= \mathbf{A}x_o+ w$, where $w \sim N(0, \sigma_w^2 I)$. We start with the following approach that is inspired by the iterative hard-thresholding algorithm \cite{BlumensathIHT} and its extensions for block-based compressive imaging \cite{mun2009block, gan2007block, fowler2010block}. To better understand this approach, consider the noiseless setting in which $y = \mathbf{A}x_o$, and assume that the denoiser is a projection onto $C$. The affine subspace defined by $y=\mathbf{A}x$ and the set $C$ are illustrated in Figure \ref{fig:DIT}.  We assume that the point $x_o$ is the unique point in the intersection of $y=\mathbf{A}x$ and $C$. 
  
\begin{figure}[t]
\begin{center}
  \includegraphics[width=.4\textwidth]{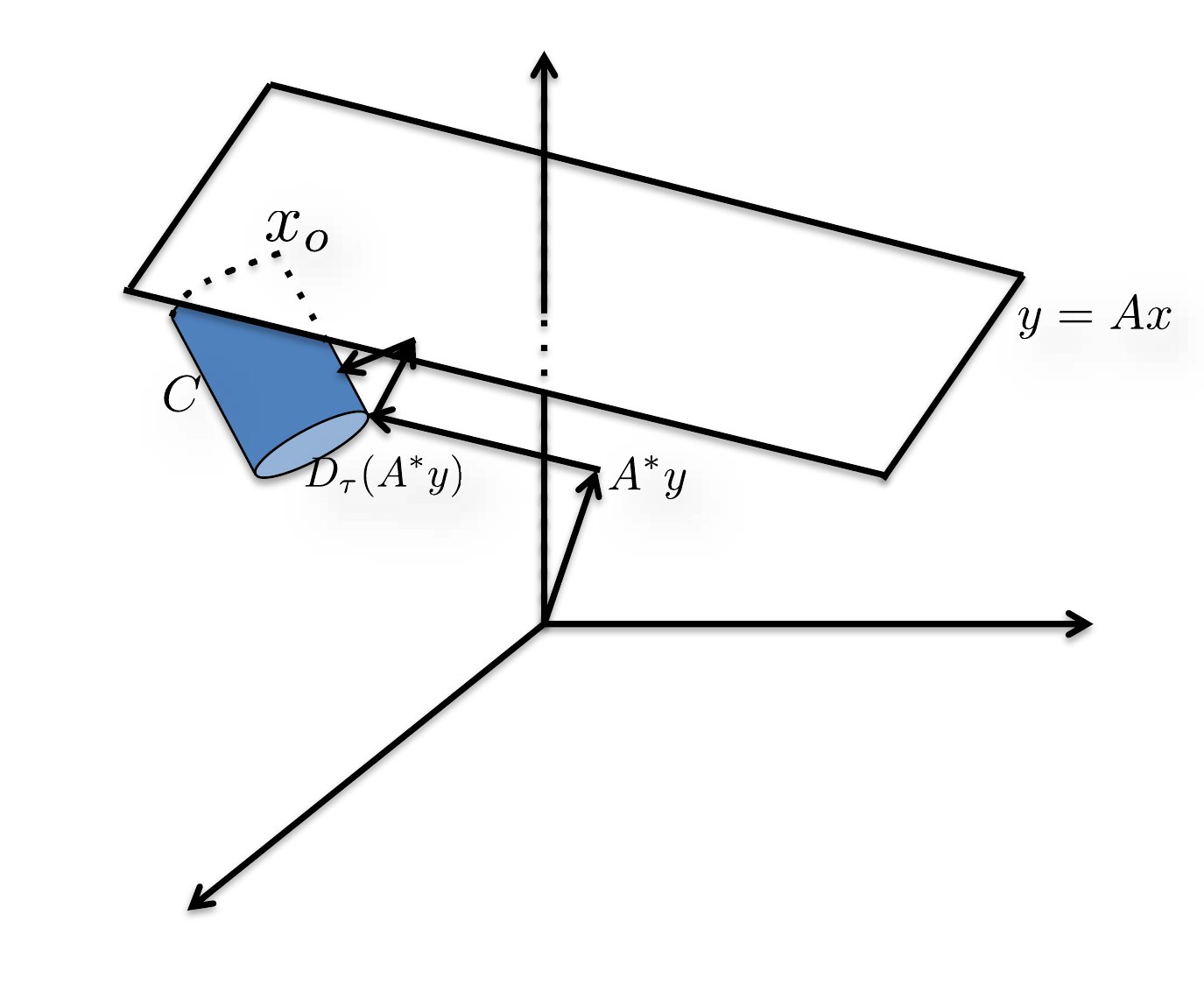} 
  \caption{Reconstruction behavior of denoising-based iterative thresholding algorithm. }
   \label{fig:DIT}
  \end{center}
\end{figure}

We know that the solution lies in the affine subspace $\{x|y= \mathbf{A}x\}$.  Therefore, starting from $x^0=0$, we move in the direction that is orthogonal to the subspace, i.e., $\mathbf{A}^*y$. $\mathbf{A}^*y$ is closer to the subspace however it is not necessarily close to $C$. Hence, we employ denoising (or projection in the figure) to obtain an estimate that satisfies the structure of our signal class $C$. After these two steps we obtain $D(\mathbf{A}^*y)$. As is also clear in the figure, by repeating these two steps, i.e., moving in the direction of the gradient and then projecting onto $\mathcal{C}$, our estimate may eventually converge to the correct solution $x_o$. This leads us to the following iterative algorithm:\footnote{Note that if $D$ is a projection operator onto $C$ and $C$ is a convex set, then this algorithm is known as projected gradient descent and is known to converge to the correct answer $x_o$. }
\begin{equation}\label{eqn:DIT}
\begin{array}{lcl}
x^{t+1}&=& D_{\hat{\sigma}} (\mathbf{A}^* z^t + x^t),\\z^t &=& y - \mathbf{A} x^t.
\end{array}
\end{equation}
For ease of notation, we have introduced the vector of estimated residual as $z^t$. We call this algorithm denoising-based iterative thresholding (D-IT). Note that if we replace $D$ (that was assumed to be projection onto set $C$ in Figure \ref{fig:DIT}) with a denoising algorithm we implicitly assume that $x^t+\mathbf{A}^*z^t$ can be modeled as $x_o + v^t$, where $v^t \sim N(0, (\sigma^t)^2 I)$ and is independent of $x_o$. Hence, by applying a denoiser we obtain a signal that is closer to $x_o$. Unfortunately, as is shown in Figure \ref{fig:QQplot_BM3DAMP_BM3DIT_AMP}(a) (we will show stronger evidence in Section \ref{sec:StateEvolution}), this assumption is not true for D-IT. This is the same phenomenon that we observed in Section \ref{sec:CSoverview} for iterative soft-thresholding.

\begin{figure}[t]
 	\begin{center}
 		\hspace*{-.75cm}
 		\includegraphics[width=.55\textwidth]{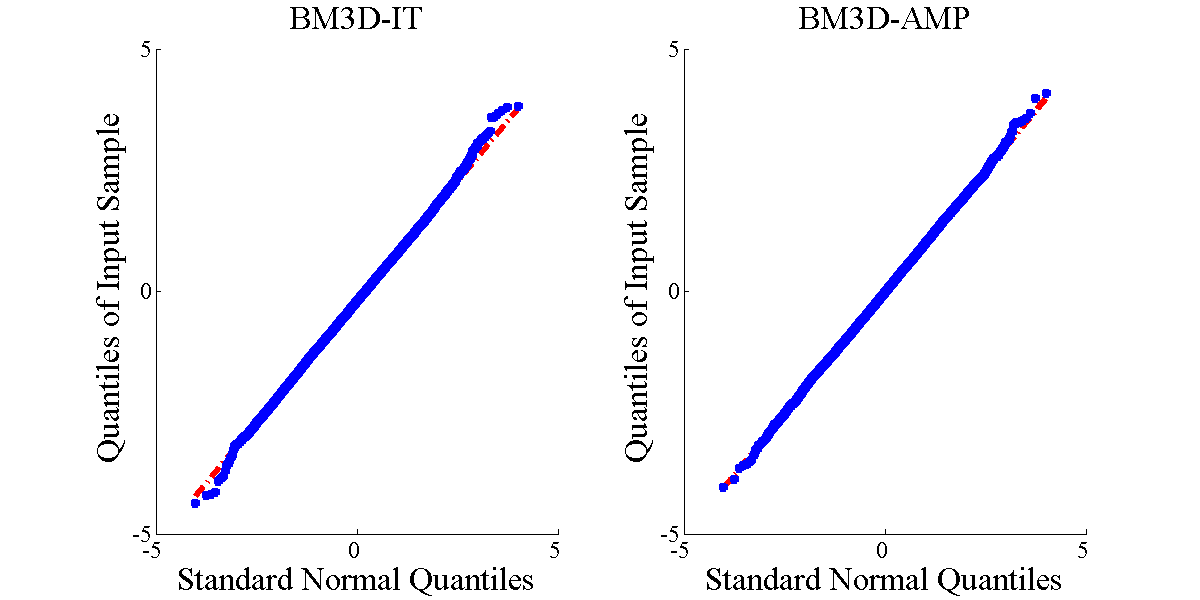}
 		\vspace{-0.5cm}
 		\caption{QQplot comparing the distribution of the effective noise of D-IT and D-AMP at iteration 5 while reconstructing a 50\% sampled Barbara test image.  Notice the highly Gaussian distribution of D-AMP and the slight deviation from Gaussianity at the ends of the D-IT QQplot. This difference is due to D-AMP's use of an Onsager correction term. The denoiser that is employed in these simulations is BM3D.  BM3D will be reviewed in Section \ref{sec:reviewdenoise}.}
 		\label{fig:QQplot_BM3DAMP_BM3DIT_AMP}
 	\end{center}
 \end{figure}

Our proposed solution to avoid the non-Gaussianity of the noise in the case of the iterative thresholding algorithms was to employ message passing/approximate message passing. Following the same path, we propose the following message passing algorithm: 
\begin{eqnarray}
x_{\cdot \rightarrow a}^t &=& D_{\hat{\sigma}^t}\left(\begin{bmatrix} \sum_{b \neq a} {\mathbf{A}_{b1}z^t_{b \rightarrow 1}} \\ \sum_{b \neq a} {\mathbf{A}_{b2}z^t_{b \rightarrow 2}}\\ \vdots \\\sum_{b \neq a} {\mathbf{A}_{bn}z^t_{b \rightarrow n}} \end{bmatrix} \right), \nonumber \\
z_{a \rightarrow i}^t&=& y_a - \sum_{j \neq i} \mathbf{A}_{aj} x^t_{j \rightarrow a}. 
\end{eqnarray}
Here, $x_{\cdot \rightarrow a}^t =[ x^t_{1 \rightarrow a}, x^t_{2 \rightarrow a}, \ldots,  x^t_{n \rightarrow a} ]^T$
 provides an estimate of $x_o$. $\hat{\sigma}^t$ denotes the standard deviation of the vector
 \[
 v_{\cdot \rightarrow a}^t= \left(\begin{bmatrix} \sum_{b \neq a} {\mathbf{A}_{b1}z^t_{b \rightarrow 1}} \\ \sum_{b \neq a} {\mathbf{A}_{b2}z^t_{b \rightarrow 2}}\\ \vdots \\\sum_{b \neq a} {\mathbf{A}_{bn}z^t_{b \rightarrow n}} \end{bmatrix} \right) - x_o.
 \]
Our empirical findings, summarized in Section \ref{sec:seandgaussianitycheck}, show that $v_{\cdot \rightarrow a}^t$ closely resembles i.i.d. Gaussian noise in high-dimensional settings (both $m$ and $n$ are large). This result has been rigorously proved for a class of scalar denoisers and can also be proved for a class of block-wise denoisers \cite{MaAnYaBa11, DonohoJM13}.\footnote{There are some subtle differences between our claim regarding the Gaussianity of $ v_{\cdot \rightarrow a}^t$ and the claims presented in other works. Our claim is made in a deterministic setting, while in existing works the Gaussianity claim is made in regards to stochastic settings. This point will be clarified in Section \ref{sec:connectionSE}.} 
 
  Despite their advantage in avoiding the non-Gaussianity of the effective noise vector $v$, message passing algorithms have $m$ (number of measurements) different estimates of $x_o$; each $x^t_{\cdot \rightarrow a}$ is an estimate of $x_o$. Similarly, they have $n$ different estimates of the residual $y-\mathbf{A}x_o$. The update of all these messages is computationally demanding. Fortunately, if the problem is high dimensional, we can approximate a message passing algorithm's iterations and obtain the denoising-based approximate message passing algorithm (D-AMP):
\begin{eqnarray}\label{eqn:DAMP2}
x^{t+1} &=& D_{\hat{\sigma}^t} (x^t+\mathbf{A}^*z^t), \nonumber\\
z^t&=& y-\mathbf{A}x^t+z^{t-1} \frac{{\rm div} D_{\hat{\sigma}^{t-1}}(x^{t-1}+\mathbf{A}^*z^{t-1})}{m}. \nonumber \\
\end{eqnarray}
 The only difference between D-AMP and D-IT is again in the Onsager correction term   $ z^{t-1}{\rm div} D_{\hat{\sigma}^{t-1}}(x^{t-1}+\mathbf{A}^*z^{t-1}) /m$. The derivation of D-AMP from the Denoising-based Message Passing (D-MP) algorithm is similar to the derivation of AMP from Message Passing (MP) which can be found in Chapter 5 of \cite{MalekiThesis}.  Similar to D-IT, D-AMP relies on the assumption that the effective noise $v^t = x^t+\mathbf{A}^*z^t-x_o$ resembles i.i.d. Gaussian noise (independent of the signal $x_o$) at every iteration. Our empirical findings confirm this assumption: Figure \ref{fig:QQplot_BM3DAMP_BM3DIT_AMP}(b) displays the effective noise at iteration 5 of D-AMP with BM3D denoising, which will be briefly explained in Section \ref{sec:reviewdenoise} (we call this algorithm BM3D-AMP). Notice the clearly Gaussian distribution. Based on this observation, and stronger evidence that we will provide in Section \ref{sec:seandgaussianitycheck}, we conjecture that $v^t$ indeed behaves as additive white Gaussian noise for high dimensional problems. The proof of this property is left for future work. In this paper we not only provide strong empirical evidence to support our conjecture, but also explore its theoretical implications.

\section{Theoretical analysis of D-AMP}\label{sec:Theory}
The main objective of this section is to characterize the theoretical properties of the D-AMP framework. In this section (and also in our simulations) we start the algorithm with $x^0=0$ and $z^0= y$. 
Our analysis is under the high dimensional setting: $n, m$ are very large while $m/n  = \delta<1$ is a fixed number. $\delta$ is called the under-determinacy of the system of equations.

\subsection{Notation}
We use boldfaced capital letters such as $\mathbf{A}$ for matrices. For a matrix $\mathbf{A}$; $\mathbf{A}_{i}$, $\mathbf{A}_{i,j}$, and $\mathbf{A}^*$ denote its $i^{\rm th}$ column, $ij^{\rm th}$ element, and its transpose, respectively. Small letters such as $x$ are reserved for vectors and scalars. For a vector $x$, $x_i$ and $\|x\|_p$ denote the $i^{\rm th}$ element of the vector and its $p$-norm, respectively. The notations $\mathbb{E}$ and $\mathbb{P}$ denote the expected value of a random variable (or a random vector) and probability of an event, respectively. If the expected value is with respect to two random variables $X$ and $Z$, then $\mathbb{E}_X$ (or $\mathbb{E}_Z$) denotes the expectation with respect to $X$ (or $Z$) and $\mathbb{E}$ denotes the expected value with respect to both $X$ and $Z$. 

\subsection{Denoiser properties}\label{sec:denoiseprop}
The role of a denoiser is to estimate a  signal $x_o$ belonging to a class of signals $C \subset \mathbb{R}^n$ from noisy observations, $x_o+ \sigma \epsilon$, where $\epsilon \sim N(0, I)$, and $\sigma >0$ denotes the standard deviation of the noise. We let $D_{\sigma}$ denote a family of denoisers indexed with the standard deviation of the noise. At every value of $\sigma$, $D_{\sigma}$ takes $x_o+ \sigma \epsilon$ as the input and returns an estimate of $x_o$.

To analyze D-AMP, we require the denoiser family to be {\em (near) proper}, {\em monotone}, and {\em Lipschitz continuous} (proper and monotone are defined below). Because most denoisers easily satisfy these first two properties, and can be modified to satisfy the third (see Section \ref{sec:Smooth}), the requirements do not overly restrict our analysis.

\begin{definition}\label{def:proper}
$D_{\sigma}$ is called a proper family of denoisers of level $\kappa$ ($\kappa \in (0,1)$) for the class of signals $C$ if 
\begin{equation}\label{eq:minimaxperformance_1}
\sup_{x_o \in C} \frac{\mathbb{E} \| D_{\sigma} (x_o + \sigma \epsilon) - x_o\|^2_2}{n} \leq \kappa \sigma^2,
\end{equation}
for every $\sigma>0$. Note that the expectation is with respect to $\epsilon \sim N(0, I)$. 
\end{definition}

To clarify the above definition, we consider the following examples:

\begin{example}\label{ex:denoiseperformk}
Let $C$ denote a $k$-dimensional subspace of $\mathbb{R}^n$ ($k< n$). Also, let $D_{\sigma}(y)$ be the projection of $y$ onto subspace $C$ denoted by $P_C(y)$. Then,
\[
 \frac{\mathbb{E} \| D_{\sigma} (x_o + \sigma \epsilon) - x_o\|^2_2}{n} = \frac{k}{n} \sigma^2,
\]
for every $x_o \in C$ and every $\sigma^2$. Hence, this family of denoisers is proper of level $k/n$.

\begin{proof}
 First note that since the projection onto a subspace is a linear operator and since $P_C(x_o) = x_o$ we have
\[
\mathbb{E} \| P_C(x_o+ \sigma \epsilon) - x_o \|_2^2= \mathbb{E} \| x_o + \sigma P_C( \epsilon) - x_o \|_2^2= \sigma^2 \mathbb{E} \|{P}_C(\epsilon)\|_2^2.
\]
Also note that since $P_C^2 = P_C$, all the eigenvalues of $P_C$ are either zero or one. Furthermore, since the null space of $P_C$ is $n-k$ dimensional, the rank of $P_C$ is $k$. Hence, $P_C$ has $k$ eigenvalues equal to $1$ and the rest are zero. Hence  $\|{P}_C(\epsilon)\|_2^2$ follows a $\chi^2$ distribution with $k$ degrees of freedom and $\mathbb{E} \| P_C(x_o+ \sigma \epsilon) - x_o \|_2^2 = k \sigma^2$.
\end{proof}
\end{example}
Next we consider a slightly more complicated example that has been popular in signal processing for the last twenty-five years. Let $\Gamma_k$ denote the set of $k$-sparse vectors. 

\begin{example}\label{ex:softfamily}
Let $\eta(y;\tau \sigma)= (|y|- \tau \sigma)_+ {\rm sign} (y)$ denote the family of soft-thresholding denoisers. Then
\begin{eqnarray}
\lefteqn{ \sup_{x_o \in \Gamma_k } \frac{\mathbb{E} \| \eta (x_o + \sigma \epsilon; \tau \sigma) - x_o\|^2_2}{n}} \nonumber \\
 &=& \Big[\frac{(1+ \tau^2)k}{n}+\frac{n-k}{n} \mathbb{E} (\eta(\epsilon_1; \tau))^2\Big]\sigma^2.\nonumber
\end{eqnarray}

Similar results can be found in other papers including \cite{DoMaMoNSPT}. But since the proof is short and the result is slightly different from similar existing results, we mention the proof here. 
\begin{proof}
For notational simplicity we assume that the first $k$ coordinates of $x_o$ are non-zero and the rest are equal to zero. 
\begin{eqnarray}\label{eq:softthreshrisk1}
	\lefteqn{\frac{\mathbb{E} \| \eta (x_o + \sigma \epsilon; \tau \sigma) - x_o\|^2_2}{n\sigma^2}} \nonumber \\
	&=& \frac{\sum_{i=1}^k \mathbb{E} (\eta(x_{o,i}+ \sigma \epsilon_i ; \tau \sigma)- x_{o,i})^2  }{n \sigma^2} + \frac{n-k}{n\sigma^2} \mathbb{E} (\eta( \sigma \epsilon_n; \tau \sigma))^2 \nonumber \\
	&=& \frac{\sum_{i=1}^k \mathbb{E} \left(\eta\left(\frac{x_{o,i}}{\sigma}+ \epsilon_i ; \tau \right)- \frac{x_{o,i}}{\sigma}\right)^2  }{n} + \frac{n-k}{n} \mathbb{E} (\eta(  \epsilon_n; \tau))^2. \nonumber \\	
\end{eqnarray}
Note that $\mathbb{E} \left(\eta\left(\frac{x_{o,i}}{\sigma}+ \epsilon_i ; \tau \right)- \frac{x_{o,i}}{\sigma}\right)^2$ is an increasing function of $\frac{x_{o,i}}{\sigma}$ \cite{mousavi2013asymptotic}. Therefore, it is straightforward to see that
\begin{eqnarray}\label{eq:softthreshrisk2}
\lefteqn{\mathbb{E} \left(\eta\left(\frac{x_{o,i}}{\sigma}+ \epsilon_i ; \tau \right)- \frac{x_{o,i}}{\sigma}\right)^2} \nonumber \\
&\leq& \lim_{x_{o,i} \rightarrow \infty} \mathbb{E} \left(\eta\left(\frac{x_{o,i}}{\sigma}+ \epsilon_i ; \tau \right)- \frac{x_{o,i}}{\sigma}\right)^2 = 1+ \tau^2,
\end{eqnarray}
where the last step swaps the $\lim$ and $\mathbb{E}$ (by the dominated convergence theorem). We obtain the desired result by combining \eqref{eq:softthreshrisk1} and \eqref{eq:softthreshrisk2}.
\end{proof}
\end{example}
Note that the optimal threshold $\tau$ to use within soft-thresholding depends on the sparsity $k/n$ of the signal being denoised. One can optimize the parameter $\tau$ for every value of $k/n$ and obtain an optimized  family of denoisers. Figure \ref{fig:levelsofthresh} displays the level $\kappa$ of the optimized soft-thresholding  in terms of $k/n$.  Note that for sparse signals ($k/n$ small) soft-thresholding is an effective denoiser and thus $\kappa$ is small.
\begin{figure}[t]
\begin{center}
  \includegraphics[width=.3\textwidth]{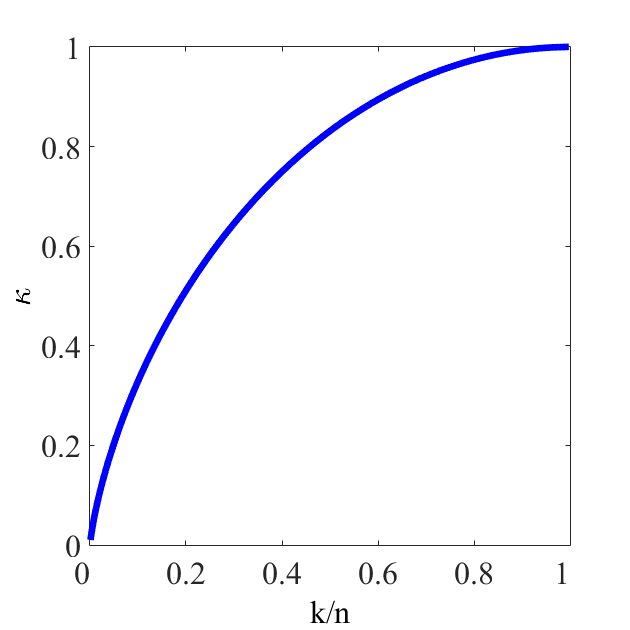} 
  \caption{The level $\kappa$ of optimal soft-thresholding method as a function of normalized sparsity $k/n$.  For sparse signals, soft-thresholding is a high performance denoiser.}
   \label{fig:levelsofthresh}
   \end{center}
\end{figure}

The previous denoisers both utilized prior knowledge about the structure of the signal (its dimensionality and its sparsity) in order to denoise $x_o$.  When nothing is known about $x_o$ a proper denoiser might be too much to ask for.  For instance, consider the maximum likelihood estimator.

\begin{example}
If $D_{\sigma} (x_o+\sigma \epsilon)$ is the maximum likelihood estimate of $x_o$ from $x_o+ \sigma \epsilon$, then 
\[
 \frac{\mathbb{E} \| D_{\sigma} (x_o + \sigma \epsilon) - x_o\|^2_2}{n\sigma^2} = 1.
\]
So, this family of denoisers are not proper of level $\kappa$ for any $\kappa<1$. The proof of this statement is straightforward and hence is skipped here. It has been shown that (Chapter 5 of \cite{lehmann1998theory}) for any denoiser $\tilde{D}_\sigma$ we have
\[
\sup_{x_o \in \mathbb{R}^n} \frac{\mathbb{E} \| \tilde{D}_{\sigma} (x_o + \sigma \epsilon) - x_o\|^2_2}{n\sigma^2}=1.
\]
\end{example}

In this example the class of signals we have considered is generic and hence the denoiser cannot employ any specific structure in $x_o$.

There are occasions when we want to deal with denoisers that are not proper because of an error/bias term that is independent of the noise level.  To deal with scenarios such as these, we introduce the definition {\em near proper}.
\begin{definition}\label{def:nearproper}
$D_{\sigma}$ is called a near proper family of denoisers of levels $\kappa$ ($\kappa \in (0,1)$) and $B$ ($B \in \mathbb{R}_+$) for the class of signals $C$ if 
\begin{equation}\label{eq:minimaxperformance_2}
\sup_{x_o \in C} \frac{\mathbb{E} \| D_{\sigma} (x_o + \sigma \epsilon) - x_o\|^2_2}{n} \leq \kappa \sigma^2+B,
\end{equation}
for every $\sigma>0$. Note that the expectation is with respect to $\epsilon \sim N(0, I)$. 
\end{definition}

As in Definition \ref{def:proper}, the constants $\kappa$ and $B$ determine the quality of the denoiser family.  Better denoisers have smaller constants.


\begin{example}
Let $\mathcal{C}_p=\{x \in \mathbb{R}^n \ : \ \|x\|_p \le 1\}$ for some $0<p\leq1$.\footnote{For every $0<p \leq 1$, $\|x\|_p^p = \sum_{i=1}^n |x_i|^p$. }  For a fixed $k$, let $D_\sigma$ denote a denoiser that, through oracle information, knows the indices of the k largest elements of $x$ and projects the noisy observation $x_o+\sigma \epsilon$ onto those coordinates.  Then
\[
\sup_{x_o \in \mathcal{C}_p}\frac{\mathbb{E} \| D_{\sigma} (x_o + \sigma \epsilon) - x_o\|^2_2}{n} \leq \frac{k}{n} \sigma^2+\frac{k^{1-2/p}}{n(2/p-1)},
\]
for every $x_o \in \mathcal{C}_p$ and every $\sigma^2$. Hence, this family of denoisers is near proper with $\kappa=\frac{k}{n}$ and $B=\frac{(k+1)^{1-2/p}}{n(2/p-1)}$.
\begin{proof}

Let $\Lambda$ denote the set of indices of the $k$-largest coefficients of $x_o$. For a vector $x$, define $x_{\Lambda}$ in the following way: $x_{\Lambda, i} = x_i$ if $i \in \Lambda$ and otherwise $x_{\Lambda,i}=0$. Note that $x_{o,\Lambda}$ is the best $k$-term approximation of $x_o$. 
We have
\begin{eqnarray}\label{eq:nearprop_proof1}
\mathbb{E} \| D_\sigma(x_o+ \sigma \epsilon) - x_o \|_2^2=\mathbb{E} \| x_{o,\Lambda} +\sigma\epsilon_\Lambda- x_o \|_2^2 \nonumber \\
=  \| x_{o,\Lambda}  - x_o \|_2^2+ \sigma^2\mathbb{E} \|\epsilon_\Lambda\|_2^2.
\end{eqnarray}
 Following the same logic as used in Example \ref{ex:denoiseperformk} we see that
\begin{equation}\label{eq:nearprop_proof2}
\mathbb{E} \|\epsilon_\Lambda\|_2^2=k.
\end{equation}

The term $ \| x_{o,\Lambda}  - x_o \|_2^2$ above is simply the squared $\ell_2$-norm of the smallest $n-k$ values of $x_o$. Below we obtain an upper bound for this quantity. Note that since $x_o \in \mathcal{C}_p$, we have
\begin{equation}\label{eq:ell_pball}
\sum\limits_{i=1}^{n} |x_{o,i}|^p\leq 1.
\end{equation}
 Let $x_{o,(j)}$ denote the $j^{\rm  th}$ largest element in absolute value of $x_o$. It is clear that $|x_{o,(1)}| \geq |x_{o,(2)}|$, \ldots, $|x_{o,(j-1)}| \geq |x_{o,(j)}|$. Combining this fact with \eqref{eq:ell_pball} we obtain $j|x_{o,(j)}|^p\leq 1$, which in turn implies $|x_{o,(j)}|\leq j^{-1/p}$. Returning to \eqref{eq:nearprop_proof1}, we see that
\begin{eqnarray}\label{eq:nearprop_proof3}
\| x_{o,\Lambda}  - x_o \|_2^2 \leq \sum\limits_{j=k+1}^{n}\|x_{o,(j)}\|^2 \leq \sum\limits_{j=k+1}^{n} j^{-2/p} \nonumber \\
\leq \int_{k}^{\infty}\gamma^{-2/p} d\gamma
=\frac{\gamma^{1-2/p}}{(1-2/p)}\Big|_{k}^{\infty}=\frac{(k)^{1-2/p}}{(2/p-1)}.
\end{eqnarray}
Substituting \eqref{eq:nearprop_proof2} and \eqref{eq:nearprop_proof3} into  \eqref{eq:nearprop_proof1} gives the desired result.
\end{proof}
\end{example}

In subsequent sections we assume our signal belongs to a class $C$ for which we have a proper or near proper family of denoisers $D_{\sigma}$.  The class and denoiser can be very general. For instance, we may assume $C$ to be the class of natural images and $D_{\sigma}$ to denote the BM3D algorithm\footnote{We will review this algorithm briefly in Section \ref{sec:reviewdenoise}} at different noise levels \cite{DabovBM3D}.

\begin{definition}
We call a denoiser monotone if for every $x_o$ its risk function 
\[
R(\sigma^2, x_o) = \frac{\mathbb{E} (\|D_{\sigma}(x_o+ \sigma \epsilon)-x_o \|_2^2)}{n}, 
\]
is a non-decreasing function of $\sigma^2$.  
\end{definition}

We make a few remarks regarding monotone denoisers.

\begin{remark}
Monotonicity is a natural property to expect from denoisers. Many standard denoisers such as soft-thresholding and group soft-thresholding are monotone if we optimize over the threshold parameter. See Lemma 4.4 in \cite{MousaviMB13a} for more information. 
\end{remark}

\begin{remark}
If a family of denoisers $D_{\sigma}$ is not monotone, then it is straightforward to construct a new denoiser that outperforms $D_{\sigma}$. Here is a simple proof. Suppose that for $\sigma_1 < \sigma_2$ we have
\[
 R(\sigma_1^2, x_o) > R(\sigma_2^2, x_o). 
 \]
 Then construct a new denoiser for noise level $\sigma_1$ in the following way:
 \[
 \tilde{D}_{\sigma_1}(y)= \mathbb{E}_{\tilde{\epsilon}} D_{\sigma_2}\left(y+ \sqrt{\sigma_2^2- \sigma_1^2}\tilde{\epsilon}\right),
 \]
 where $\tilde{\epsilon} \sim N(0,I)$ is independent of $y$ and $\mathbb{E}_{\tilde{\epsilon}} (y+ \sqrt{\sigma_2^2- \sigma_1^2}\tilde{\epsilon})$ denotes the expected value with respect to $\tilde{\epsilon}$. Let $\tilde{\sigma}_2 = \sqrt{\sigma_2^2- \sigma_1^2}$. A simple application of Jensen's inequality shows that
 \begin{eqnarray}
  \lefteqn{\frac{\mathbb{E}_\epsilon (\|\tilde{D}_{\sigma_1}(x_o+ \sigma_1 \epsilon)-x_o \|_2^2)}{n}} \nonumber \\
  &=&  \frac{\mathbb{E}_\epsilon( \| \mathbb{E}_{\tilde{\epsilon}} {D}_{{\sigma}_2}(x_o+ \sigma_1 \epsilon+ \tilde{\sigma}_2 \tilde{\epsilon})-x_o \|_2^2)}{n} \nonumber \\
  &\leq& \frac{\mathbb{E}_{\epsilon, \tilde{\epsilon}}( \| {D}_{\sigma_2}(x_o+ \sigma_1 \epsilon+ \tilde{\sigma}_2 \tilde{\epsilon})-x_o \|_2^2)}{n}. \nonumber
 \end{eqnarray}
Note that since $\tilde{\epsilon}$ and $\epsilon$ are independent $\frac{\mathbb{E}_{\tilde{\epsilon}, \epsilon}( \| {D}_{\sigma_2}(x_o+ \sigma_1 \epsilon+ \tilde{\sigma}_2 \tilde{\epsilon})-x_o \|_2^2)}{n} = R(\sigma_2^2, x_o)$. Therefore, $\tilde{D}$ improves $D$ and does not violate the monotone property. Therefore, as is clear from this statement, non-monotone denoisers are not desirable in general since we can  easily improve them.  
\end{remark}

In the rest of the paper we consider only monotone denoisers.
\subsection{State evolution}\label{sec:StateEvolution}
A key ingredient in our analysis of D-AMP is the state evolution; a series of equations that predict the intermediate MSE of AMP algorithms at each iteration.
Here we introduce a new ``deterministic'' state-evolution to predict the performance of D-AMP.
Starting from $\theta^0 = \frac{\|x_o\|_2^2}{n}$ the state evolution generates a sequence of numbers through the following iterations:
\begin{equation}\label{eq:stateevo1}
\theta^{t+1}(x_o, \delta, \sigma_w^2) = \frac{1}{n}  \mathbb{E} \| D_{\sigma^t} (x_o + \sigma^t \epsilon) -x_o \|_2^2, 
\end{equation}
where $(\sigma^t)^2 = \frac{\theta^t}{\delta}(x_o, \delta, \sigma_w^2) + \sigma_w^2$ 
 and the expectation is with respect to $\epsilon \sim N(0,I)$. Note that our notation $\theta^{t+1}(x_o, \delta, \sigma_w^2)$ is set to emphasize that $\theta^t$ may depend on the signal $x_o$, the under-determinacy $\delta$, and the measurement noise. Consider the iterations of D-AMP and let $x^t$ denote  its estimate at iteration $t$. Our empirical findings show that the MSE of D-AMP is predicted accurately by the state evolution. We formally state our finding.
 
 \begin{figure}[t]
 	\begin{center}
 		\includegraphics[width=.5\textwidth]{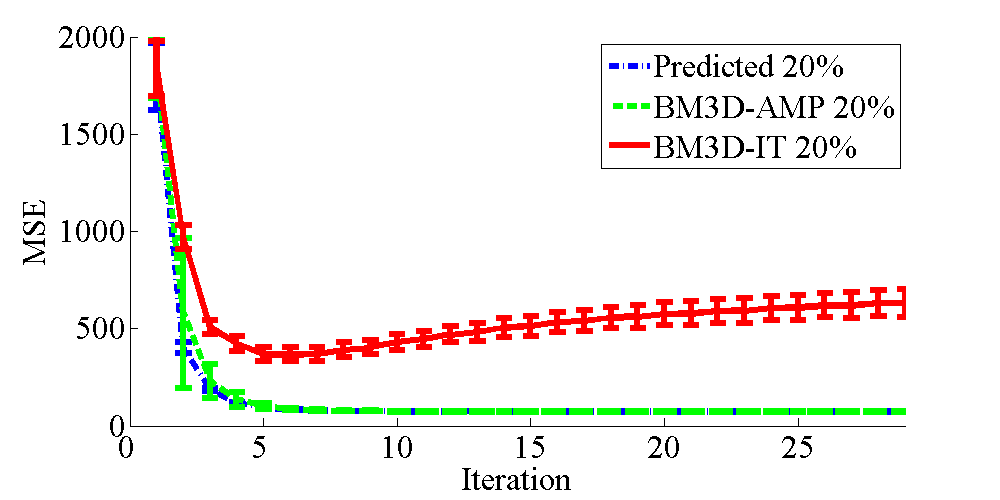} 
 		\caption{The MSE of the intermediate estimate versus the iteration count for BM3D-AMP and BM3D-IT alongside their predicted state evolution.  Notice that BM3D-AMP is well predicted by the state evolution whereas BM3D-IT is not.}
 		\label{fig:StateEvolution} 
 	\end{center}
 \end{figure}
 
\begin{finding}\label{find:mainone}
If the D-AMP algorithm starts from $x^0=0$, then for large values of $m$ and $n$, state evolution predicts the mean square error of D-AMP, i.e., 
\[
\theta^t (x_o, \delta, \sigma_w^2) \approx \frac{1}{n} \|x^t-x_o\|_2^2. 
\]
\end{finding}

Based on extensive simulations, we believe that this finding is true if the following properties are satisfied: (i) The elements of the matrix $\mathbf{A}$ are i.i.d.\ Gaussian (or subGaussian) with mean zero and standard deviation $1/m$. (ii) The noise $w$ is also i.i.d.\ Gaussian. (iii) The denoiser $D$ is Lipschitz continuous.\footnote{A denoiser is said to be $L$-Lipschitz continuous if for every $x_1,x_2 \in C$ we have $\|D(x_1)-D(x_2)\|_2^2\leq L \|x_1-x_2\|_2^2.$	Many advanced image denoisers have no closed form expression, thus it is very hard to verify whether or not they are Lipschitz continuous.  That said,  every advanced denoisers we tested was found to closely follow our state evolution equations (Finding \ref{find:mainone}), suggesting they are in fact Lipschitz. In Section \ref{sec:Smooth} we show examples in which Lipschitz continuity is violated and propose a simple approach for dealing with discontinuous denoisers.}	
In all our simulations the elements of $A$ are i.i.d.\ Gaussian. The same is true for the elements of $w$.\label{pg:finding}

Figure \ref{fig:StateEvolution} compares the state evolution predictions of D-AMP (based on the BM3D denoising algorithm \cite{DabovBM3D}) with the empirical performance of D-AMP and D-IT. As is clear from this figure, the state evolution is accurate for D-AMP but not for D-IT. We have checked the validity of the above finding for the following denoising algorithms: (i) BM3D\cite{DabovBM3D}, (ii) BLS-GSM\cite{PortillaBLSGSM}, (iii) Non-local means\cite{Buades05areview}, (iv) AMP with soft-wavelet-thresholding\cite{DoMaMo09, Donoho94idealspatial}. We report some of our simulations on this phenomenon in Section \ref{sec:seandgaussianitycheck}. We have posted our code online\footnote{\href{http://dsp.rice.edu/software/DAMP-toolbox}{http://dsp.rice.edu/software/DAMP-toolbox}}  to enable other researchers to verify our findings in more general settings and explore the validity of this conjecture on a wider range of denoisers. 

In the following sections we assume that the state evolution is accurate for D-AMP and derive some of the main features of D-AMP based on this assumption.

\subsection{Analysis of D-AMP in the absence of measurement noise}\label{sec:noiseless}

In this section we consider the noiseless setting $\sigma_w^2=0$ and characterize the number of measurements D-AMP requires (under the validity of the state evolution framework) to recover the signal $x_o$ exactly. We consider monotone denoisers, as defined in section \ref{sec:denoiseprop}. Consider the state evolution equation under the noiseless setting $\sigma_w^2=0$:
\[
\theta^{t+1}(x_o, \delta, 0) =  \frac{1}{n}  \mathbb{E} \|D_{\sigma^t} (x_o + \sigma^t \epsilon) -x_o \|_2^2,
\]
where $(\sigma^t)^2 = \frac{\theta^t(x_o, \delta, 0)}{\delta}$. Starting with $\theta^0(x_o, \delta, 0) = \frac{\|x_o\|_2^2}{n}$, depending on the value of $\delta$ there are two conceivable scenarios for the state evolution equation:
\begin{itemize}
\item[(i)] $\theta^t(x_o, \delta, 0) \rightarrow 0$ as $t \rightarrow \infty$.
\item[(ii)] $\theta^t(x_o, \delta, 0) \nrightarrow 0$ as $t \rightarrow \infty$.
\end{itemize}
$\theta^t(x_o, \delta, 0) \rightarrow 0$ implies the success of D-AMP algorithm, while $\theta^t(x_o, \delta, 0) \nrightarrow 0$ implies its failure in recovering $x_o$. The main goal of this section is to study the success and failure regions.

\begin{lemma}\label{lem:ptlemma1}
For monotone denoisers, if for $\delta_0$, $\theta^{t}(x_o, \delta_0, 0) \rightarrow 0$, then for any $\delta> \delta_0$, $\theta^t(x_o, {\delta},0) \rightarrow 0$ as well.  \end{lemma}
\begin{proof}
Define $(\sigma^t)^2 = \frac{\theta^t(x_o, \delta, \sigma_w^2)}{\delta}$. Clearly, since $\theta^t(x_o, \delta_0, \sigma_w^2)\rightarrow 0$ so does 
 $\sigma^t$. Our first claim is that for every $\sigma^2 < \frac{\|x_o\|_2^2}{n\delta_0} =(\sigma^0)^2$ (this is where D-AMP is initialized)  we have
 \[
 \frac{1}{n\delta_0}  \mathbb{E} \|D_{\sigma} (x_o + \sigma \epsilon) -x_o \|_2^2 < \sigma^2,  \ \ \ \forall \sigma^2>0. 
 \]
Suppose that this is not true and define
\[
\sigma_*^2 = \sup_{\sigma^2 \leq \frac{\|x_o\|_2^2}{n\delta_0}} \{\sigma^2 \ : \   \frac{1}{n\delta_0}  \mathbb{E} \|D_{\sigma} (x_o + \sigma \epsilon) -x_o \|_2^2 \geq \sigma^2\}. 
\]
We claim that if  $\frac{1}{n\delta_0}  \mathbb{E} \|D_{\sigma} (x_o + \sigma^0 \epsilon) -x_o \|_2^2 < (\sigma^0)^2$, then  $\sigma^t \rightarrow \sigma_*$ as $t \rightarrow \infty$. First, it is straightforward to see that $\frac{1}{n\delta_0}  \mathbb{E} \|D_{\sigma_*} (x_o + \sigma_* \epsilon) -x_o \|_2^2 = \sigma_*^2$. For $\sigma> \sigma_*$ we know that
\[
\frac{1}{n\delta_0}  \mathbb{E} \|D_{\sigma} (x_o + \sigma \epsilon) -x_o \|_2^2 < \sigma^2. 
\]
By using the monotonicity of the denoiser we have for every $\sigma \geq  \sigma_*$
\[
\frac{1}{n\delta_0}  \mathbb{E} \|D_{\sigma} (x_o + \sigma \epsilon) -x_o \|_2^2 \geq \frac{1}{n\delta_0}  \mathbb{E} \|D_{\sigma_*} (x_o + \sigma_* \epsilon) -x_o \|_2^2= \sigma_*^2. 
\]
This (through simple induction) implies that for every $t$,
\[
(\sigma^t)^2  \geq (\sigma_*)^2. 
\]
Furthermore according to the definition of $\sigma_*^2$ and the fact that $\sigma^t> \sigma_*$, we have
\[
(\sigma^{t+1})^2= \frac{1}{n\delta_0}  \mathbb{E} \|D_{\sigma^t} (x_o + \sigma^t \epsilon) -x_o \|_2^2 < (\sigma^t)^2.
\]
Therefore, $\sigma^{t+1}$ is a decreasing sequence with lower bound $\sigma_*$. Hence, $\sigma^t$ converges to $\sigma^\infty \geq \sigma_*$. The last step is to show that $\sigma^\infty= \sigma_* $. If this is not the case, then $\sigma^{\infty}> \sigma_*$. But according the definition of $\sigma_*$ and the supposition that $\sigma^{\infty} > \sigma_*$, we have
\[
\frac{1}{n\delta_0}  \mathbb{E} \|D_{\sigma^{\infty}} (x_o + \sigma^{\infty} \epsilon) -x_o \|_2^2 < (\sigma^{\infty})^2,
\] 
which is a contradiction to $\sigma^{\infty}$ being a fixed point. Hence $\sigma^{\infty} = \sigma_*$. Since $\sigma^{\infty} =0$, we conclude that $\sigma_*=0$ and we have
 \[
 \frac{1}{n\delta_0}  \mathbb{E} \|D_{\sigma} (x_o + \sigma \epsilon) -x_o \|_2^2 < \sigma^2,  \ \ \ \forall \sigma^2>0. 
 \]
Since, $\delta> \delta_0$ we can conclude that 
\[
 \frac{1}{n\delta}  \mathbb{E} \|D_{\sigma} (x_o + \sigma \epsilon) -x_o \|_2^2 < \sigma^2,  \ \ \ \forall \sigma^2>0. 
 \]
Hence the only fixed point of this equation is also at zero and hence $\theta^t(x_o, {\delta},0) \rightarrow 0$. Note that all the above argument is based on the assumption that $\frac{1}{n\delta_0}  \mathbb{E} \|D_{\sigma} (x_o + \sigma^0 \epsilon) -x_o \|_2^2 < (\sigma^0)^2$. What if this assumption is violated? Using similar argument it is straightforward to show that if the  $\frac{1}{n\delta_0}  \mathbb{E} \|D_{\sigma} (x_o + \sigma \epsilon) -x_o \|_2^2$ has a fixed point above $(\sigma^0)^2$, then the algorithm converges to the closest fixed point above $\sigma_0$, which is a contradiction again. Also, if the algorithm does not have any fixed point above $\sigma^0$, then it will diverge to infinity, which is again a contradiction. 
\end{proof}
Note that for very small values of $\delta$, it is straightforward to see that $\theta^t(x_o , \delta, 0) \nrightarrow 0$ as $t \rightarrow \infty$. If we combine this result with Lemma \ref{lem:ptlemma1} we conclude the following simple result: For small values of $\delta$ D-AMP fails in recovering $x_o$. As $\delta$ increases, after a certain value of $\delta$ D-AMP will successfully recover $x_o$ from its undersampled measurements. Define 
\[
\delta^*(x_o) = \inf_{\delta \in (0,1)} \{\delta \ : \ \theta^t(x_o,{\delta},0) \rightarrow 0 \  \ \text{ as }\  t \rightarrow \infty\}.  
\]
$\delta^*(x_o)$ denotes the minimum number of measurements required for the successful recovery of $x_o$. Our goal is to characterize $\delta^*(x_o)$ in terms of the performance (we will clarify what we mean by performance) of the denoising algorithm. However, since the number of measurements $\delta^*(x_o)$ depends on the signal $x_o$, a more natural question in the design of a system is the following: How many measurements does D-AMP require to recover every signal $x_o \in C$? The following result addresses this question.
\begin{proposition}\label{prop:phasetrans}
Suppose that for signal class $C$ the denoiser $D_{\sigma}$ is proper at level $\kappa$. Then 
\[
\sup_{x_o \in C} \delta^*(x_o) \leq \kappa.
\]
\end{proposition}  

\begin{proof}
The proof of this proposition is a simple application of the state evolution equation. Similar to the proof of Lemma \ref{lem:ptlemma1} define 
\[
(\sigma^t(x_o, \delta, \sigma_w^2))^2 = \frac{\theta^t(x_o, \delta, \sigma_w^2)}{\delta}.
\]
Also for notational simplicity we use the notation $\sigma^t$ instead of $\sigma^t(x_o, \delta, 0)$ in the equation below. According to state evolution we have 
\begin{eqnarray}\label{eq:noiselesspt}
(\sigma^{t+1})^2  &=& \frac{1}{n\delta}  \mathbb{E} \|D_{\sigma^t} (x_o + \sigma^t \epsilon) -x_o \|_2^2 \nonumber \\
&=& \frac{ (\sigma^t)^2}{n\delta (\sigma^t)^2}  \mathbb{E} \|D_{\sigma^t} (x_o + \sigma^t \epsilon) -x_o \|_2^2 \nonumber \\
&\leq& \frac{(\sigma^t)^2}{\delta}   \sup_{x_o \in C} \frac{\mathbb{E} \|D_{\sigma^t} (x_o + \sigma^t \epsilon) -x_o \|_2^2}{n (\sigma^t)^2}  \nonumber \\
&\leq& \frac{\kappa (\sigma^t)^2}{\delta}.
\end{eqnarray}
It is straightforward to see that
\[
(\sigma^{t}(x_o, \delta, 0))^2 \leq \left(\frac{\kappa}{\delta}\right)^t (\sigma^{0}(x_o, \delta, 0))^2 . 
\]
Hence, if $\delta > \kappa$, then $(\sigma^{t}(x_o, \delta, 0))^2 \rightarrow 0$ as $t \rightarrow \infty$. 
\end{proof}
We can apply Proposition \ref{prop:phasetrans} to the examples of Section \ref{sec:denoiseprop} and derive some well-known results, such as the phase transition of AMP with the soft-threshold denoiser \cite{DoMaMo09}.

If our denoiser is only nearly proper, perfect recovery may not be possible. However, we can use the same technique to bound the recovery error of D-AMP.

\begin{lemma}\label{lem:nearpropphasetrans}
Let $D_\sigma$ denote a near proper family of denoisers with levels $\kappa$ and $B$, as defined in Definition \ref{def:nearproper}. Then, if $\delta> \kappa$, the error of D-AMP is upper bounded by
\[
\lim_{t \rightarrow \infty} (\sigma^{t}(x_o, \delta, 0))^2 \leq \frac{B}{\delta-\kappa}.
\]
\end{lemma}

\begin{proof}\label{eq:nearproperbound}
The proof of this result is much like the one used for proper denoisers.  Again define $\sigma^t(x_o, \delta, \sigma_w^2) = \frac{\theta^t(x_o, \delta, \sigma_w^2)}{\delta}$.  Using the state evolution and the definition of near proper we have
\begin{eqnarray}
\lefteqn{(\sigma^{t+1}(x_o, \delta, 0))^2}\nonumber \\
&=& \frac{1}{n\delta}  \mathbb{E} \|D_{\sigma^t(x_o, \delta, 0)} (x_o + \sigma^t(x_o, \delta,0) \epsilon) -x_o \|_2^2 \nonumber \\
&\leq& \frac{\kappa (\sigma_t(x_o, \delta, 0))^2+B}{\delta}.\nonumber
\end{eqnarray}
Hence
\begin{eqnarray}
(\sigma^{t}(x_o, \delta, 0))^2 &\leq&(\frac{\kappa}{\delta})^t\frac{||x_o||^2_2}{n}+(\frac{1-(\kappa/\delta)^t}{1-\kappa/\delta})\frac{B}{\delta}.\nonumber
\end{eqnarray}
For $\delta>\kappa$, the limit of this sequence is as follows
\begin{eqnarray*}
\lim_{t \to \infty} (\sigma^{t}(x_o, \delta, 0))^2 \leq \frac{B}{\delta-\kappa}.
\end{eqnarray*}
\end{proof}

Note that the proof techniques employed above was first developed in \cite{DoMaMo09} and was later employed to establish the phase transition of AMP extensions \cite{DonohoJM13}. There are some minor differences between our derivation and the derivations presented in the other papers since we have not adopted the minimax setting.

\subsection{Noise sensitivity of D-AMP}\label{ssec:noisy}
In Section \ref{sec:noiseless} we considered the performance of D-AMP in the noiseless setting where $\sigma_w^2=0$. This section will be devoted to the analysis of D-AMP in the presence of the measurement noise. Here we assume that the denoiser is near proper at levels $\kappa$ and $B$, i.e.,
\begin{equation}\label{eq:minimaxperformance1}
 \sup_{\sigma^2}\sup_{x_o \in C} \frac{\mathbb{E} \| D_{\sigma} (x_o + \sigma \epsilon) - x_o\|^2_2}{n} \leq \kappa \sigma^2+B .
\end{equation}

The following result shows that D-AMP is robust to the measurement noise. Let $\theta^\infty(x_o, \sigma_w^2, \delta)$ denote the fixed point of the state evolution equation. Since there is measurement noise, $\theta^\infty(x_o, \sigma_w^2, \delta) \neq 0$, i.e., D-AMP will not recover $x_o$ exactly. We define the noise sensitivity of D-AMP as
\[
NS(\sigma_w^2, \delta) = \sup_{x_o \in C} \theta^\infty(x_o, \delta, \sigma_w^2).
\]
The following proposition provides an upper bound for the noise sensitivity as a function of the number of measurements and the variance of the measurement noise.

\begin{proposition}\label{prop:noisesens}
Let $D_\sigma$ denote a near proper family of denoisers at levels  $\kappa$ and $B$. Then, for $\delta> \kappa$, the noise sensitivity of D-AMP satisfies
\begin{equation}\label{eq:nearproperNS}
NS(\sigma_w^2, \delta) \leq \frac{\kappa  \sigma_w^2+B}{1-\frac{\kappa}{\delta}}.
\end{equation}
\end{proposition}

\begin{proof}
Note that $\theta^{\infty}(x_o, \delta, \sigma_w^2)$ is a fixed point of the state evolution equation and hence it satisfies
\[
\theta^{\infty} (x_o, \delta, \sigma_w^2)= \frac{1}{n} \mathbb{E} \|D_{\sigma^{\infty}}(x_o + \sigma^{\infty}(x_o, \delta, \sigma_w^2) \epsilon)-x_o\|_2^2, 
\]
where $\sigma^{\infty} (x_o, \delta, \sigma_w^2)= \sqrt{\theta^{\infty}(x_o, \delta, \sigma_w^2)/\delta+ \sigma_w^2 }$. 
Therefore,
\begin{eqnarray}
NS(\sigma_w^2, \delta) &=& \sup_{x_o \in C} \theta^\infty(x_o, \delta, \sigma_w^2) \nonumber \\
&=& \sup_{x_o \in C} \frac{1}{n} \mathbb{E} \|D_{\sigma^{\infty}}(x_o + \sigma^{\infty} (x_o, \delta, \sigma_w^2)\epsilon)-x_o\|_2^2 \nonumber \\
&\leq&  \sup_{x_o \in C} {\kappa} (\sigma^{\infty}(x_o, \tau, \sigma_w^2))^2+B \nonumber \\
&=& \sup_{x_o \in C}  \kappa \left(\frac{\theta^{\infty}(x_o, \delta, \sigma_w^2)}{\delta} + \sigma_w^2 \right)+B \nonumber \\
&=& \frac{\kappa}{ \delta} NS(\sigma_w^2, \delta) + \kappa \sigma_w^2+B. \nonumber
\end{eqnarray}
A simple calculation completes the proof.
\end{proof}

Substituting in $B=0$ into the above result gives the noise sensitivity for proper denoisers.

\begin{equation}\label{eq:properNS}
NS(\sigma_w^2, \delta) \leq \frac{\kappa  \sigma_w^2}{1-\frac{\kappa}{\delta}}.
\end{equation}

There are several interesting features of this proposition that we would like to emphasize. 

\begin{remark}
The bound we presented in Proposition \ref{prop:noisesens} is a worst case analysis. The bound may be achieved for certain signals in $C$ and certain noise variances. However, for most signals in $C$ and most noise variances D-AMP will perform better than what is predicted by the bound. Figure \ref{fig:MSEvsSigma} shows the performance of BM3D-AMP in terms of the standard deviation of the noise. 
\end{remark}

The technique we employed above was first developed in \cite{DoMaMoNSPT}. The result we derived in Proposition \ref{prop:noisesens} can be considered as a generalization of the result of \cite{DoMaMoNSPT} to much broader class of denoisers.

As an aside, upper and lower bounds were recently derived for the minimax noise sensitivity of {\em any} recovery algorithm when the measurement matrix is i.i.d. Gaussian and the compressively sampled signal is sparse \cite{reeves2013minimax}. Note that while our results can be applied to sparse signals, they have been derived under much more general setting. In this section we discussed upper bounds on the noise sensitivity. See Section \ref{sec:optimality} for some preliminary results on the lower bound. 
\begin{figure}[t]
\begin{center}
  \includegraphics[width=.5\textwidth]{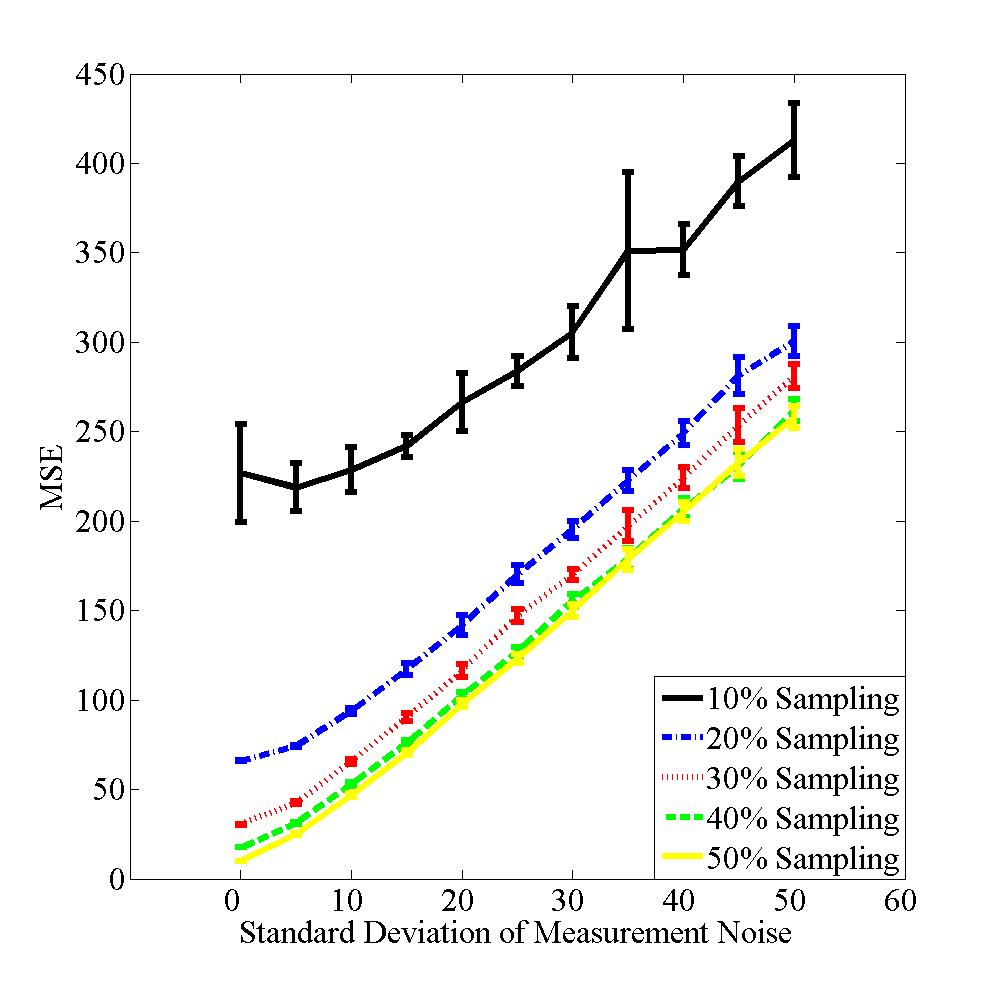} 
  \caption{The MSE of BM3D-AMP reconstructions of $128 \times 128$ Barbara test image with varying amounts of measurement noise at different sampling rates ($\delta$).}
   \label{fig:MSEvsSigma} 
  \end{center}
\end{figure}

\subsection{Tuning the parameters of D-AMP}\label{sec:DAMP_Params}
Practical denoisers typically have a few free parameters and the denoisers' performance relies on the effective tuning of these parameters. One of the simplest examples of a denoiser with parameters is soft-thresholding (introduced in Example \ref{ex:softfamily}), for which the threshold can be regarded as a parameter. There exists extensive literature on tuning the free parameters of denoisers \cite{zhu2010automatic, GalatsanosParams}. Diverse and powerful algorithms such as SURE (Stein's Unbiased Risk Estimation) have been proposed for this purpose \cite{RamaniMCSure}. 

D-AMP can employ any of these tuning schemes. However, once we use a denoising algorithm in the D-AMP framework the problem of tuning the free parameters  of the denoiser seems to become dramatically more difficult: to produce good performance from D-AMP the parameters must be tuned jointly across different iterations. 
 To state this challenge we overload our notation of a denoiser to $D_{\sigma, \tau}$, where $\tau$ denotes the denoiser's parameters. According to this notation the state evolution is given by
\[
o^{t+1}(\tau^0, \tau^1, \ldots, \tau^{t}) = \frac{1}{n}  \mathbb{E} \|D_{\sigma^t, \tau^t} (x_o + \sigma^t \epsilon) -x_o \|_2^2, 
\]
where $(\sigma^t)^2= \frac{o^{t}(\tau^0, \tau^1, \ldots, \tau^{t-1})}{\delta} + \sigma_w^2$. 
Note that we have changed our notation for the state evolution variables to emphasize the dependence of $o^{t+1}$ on the choice of the parameters we pick at at the previous iterations. The first question that we ask is the following: What does the optimality of $\tau^0, \tau^1, \ldots, \tau^t$ mean? Suppose that the sequence of parameters $\tau^t$ is bounded. 

\begin{definition}
A sequence of parameters $\tau_*^1, \ldots, \tau_*^t$ is called optimal at iteration $t+1$ if
\[
o^{t+1}(\tau_*^0, \ldots, \tau_*^t) = \min_{\tau^0, \tau^1, \ldots, \tau^t} o^{t+1}(\tau^0, \tau^1, \ldots, \tau^t).
\]
\end{definition}

Note that $\tau_*^0, \ldots, \tau_*^t$ is optimal in the sense that they produce the smallest mean square error D-AMP can achieve after $t$ iterations. This definition was first given in \cite{MousaviMB13a} for the AMP algorithm based on soft-thresholding. \\

It seems from our formulation that we should solve a joint optimization on $\tau^0, \ldots, \tau^t$ to obtain the optimal values of these parameters. However, it turns out that in D-AMP the optimal parameters can be found much more easily. Consider the following greedy algorithm for setting the parameters: 
\begin{enumerate}[(i)]
\item Tune $\tau^0$ such that $o^{1}(\tau^0)$ is minimized. Call the optimal value $\tau_*^0$. 
\item If $\tau^0, \ldots, \tau^{t-1}$ are set to $\tau^0_*, \ldots, \tau^{t-1}_{*}$, then set $\tau^{t}$ such that it minimizes $o^{t+1}(\tau_*^0, \ldots,\tau^{t-1}_{*}, \tau^t)$. 
\end{enumerate}

Note that the above strategy is a greedy parameter selection. 
The following result proves that in the context of D-AMP this greedy strategy is optimal:

\begin{lemma}\label{lem:tuning}
Suppose that the denoiser $D_{\sigma, \tau}$ is monotone in the sense that $\inf_\tau \mathbb{E} \|D_{\sigma, \tau} (x_o + \sigma \epsilon) -x_o \|_2^2$ is a non-decreasing function of $\sigma$. If $\tau_*^0, \ldots, \tau_*^t$ is generated according to the greedy tuning algorithm described above, then
\[
o^{t+1}(\tau_*^0, \ldots, \tau_*^t)  \leq o^{t+1}(\tau^0, \ldots, \tau^t), \ \ \forall \tau^0, \ldots, \tau^t, 
\]
for every $t$. 
\end{lemma}
\begin{proof}

Our proof is based on an induction. According to the first step of our procedure we know that 
 \[
 o^1(\tau_*^0) \leq o^1(\tau^0), \ \ \ \ \forall \tau^0.
 \]
 Now suppose that the claim of the theorem is true for every $t \leq T$. We would like to prove that  the result also holds for $t=T+1$, i.e.,
 \[
o^{T+1}(\tau_*^0, \ldots, \tau_*^{T})  \leq o^{T+1}(\tau^0, \ldots, \tau^{T}), \ \ \forall \tau^0, \ldots, \tau^{T}. 
\]
Suppose that it is not true and for $\tau_o^0, \ldots, \tau_o^{T}$ we have
 \begin{equation}\label{eq:contradiction1}
o^{T+1}(\tau_*^0, \ldots, \tau_*^{T})  > o^{T+1}(\tau_o^0, \ldots, \tau_o^{T}). 
\end{equation}
Clearly, 
\[
o^{T+1}(\tau_*^1, \tau_*^2, \ldots, \tau_*^{T}) = \frac{1}{n}  \mathbb{E} \|D_{\sigma^t, \tau^T} (x_o + \sigma_*^{T} \epsilon) -x_o \|_2^2, 
\]
where $ (\sigma_*^T)^2 = \frac{o^{T}(\tau_*^0, \ldots, \tau_*^{T-1})}{\delta} + \sigma_w^2$. If we define $(\sigma_o^T)^2 = \frac{o^{T}(\tau_o^0, \ldots, \tau_o^{T-1})}{\delta} + \sigma_W^2$, then according to the induction assumption $\sigma_*^{T} \leq \sigma_o^T$. Therefore, according to the monotonicity of the denoiser
\begin{eqnarray}
o^{T+1}(\tau_*^0, \tau_*^1, \ldots, \tau_*^{T}) &=& \inf_{\tau^T}\frac{1}{n}  \mathbb{E} \|D_{\sigma_*^T, \tau^T} (x_o + \sigma_*^{T} \epsilon) -x_o \|_2^2 \nonumber\\
&\leq& \inf_{\tau^T}  \frac{1}{n}  \mathbb{E} \|D_{\sigma_o^t, \tau^T} (x_o + \sigma_o^{T} \epsilon) -x_o \|_2^2 \nonumber \\
&\leq& \frac{1}{n}  \mathbb{E} \|D_{\sigma_o^t, \tau_o^T} (x_o + \sigma_o^{T} \epsilon) -x_o \|_2^2 \nonumber \\
&= &o^{T+1}(\tau_o^0, \tau_o^1, \ldots, \tau_o^{T}).  \nonumber
\end{eqnarray}
This is in contradiction with \eqref{eq:contradiction1}. Hence, 
 \[
o^{T+1}(\tau_*^0, \ldots, \tau_*^{T})  \leq o^{T+1}(\tau^0, \ldots, \tau^{T}), \ \ \forall \tau^0, \ldots, \tau^{T}. 
\]
\end{proof}

To summarize the above discussion, greedy parameter tuning is optimal for D-AMP, thus the tuning of D-AMP is as simple (or as difficult) as the tuning of the denoising algorithm that is employed in D-AMP. Many researchers in the area of signal denoising have optimized the parameters of state-of-the-art denoisers, such as BM3D.  
Lemma \ref{lem:tuning} implies that optimally tuned denoisers will induce the best possible performance from D-AMP. Therefore, the tuning of D-AMP has already been thoroughly addressed in the denoising literature\cite{zhu2010automatic, GalatsanosParams, RamaniMCSure}.

\subsection{Optimality of D-AMP}\label{sec:optimality}
\subsubsection{Problem definition}
D-AMP is a framework by which to employ denoisers to solve linear inverse problems. But is D-AMP optimal? In other words, given a family of denoisers, $D_{\sigma}$, for a set $C$, can we come up with an algorithm  for recovering $x_o$ from $y=\mathbf{A}x_o+w$ that outperforms D-AMP? Note that this problem is ill-posed in the following sense: the denoising algorithm might not capture all the structures that are present in the signal class $C$. Hence, a recovery algorithm employs extra structures not used by the denoiser (and thus not used by D-AMP) clearly might outperform D-AMP. In the following sections we use two different approaches to analyze the optimality of D-AMP.

\subsubsection{Uniform optimality} 
Let $\mathcal{{E}}_{\kappa}$ denote the set of all classes of signals $C$ for which there exists a family of denoisers $D_{\sigma}^C$ that satisfies
\begin{equation}\label{eq:minimaxperformance}
 \sup_{\sigma^2}\sup_{x_o \in C} \frac{\mathbb{E} \| D^C_{\sigma} (x_o + \sigma \epsilon) - x_o\|^2_2}{n\sigma^2} \leq \kappa .
\end{equation}
We know from Proposition \ref{prop:phasetrans} that for any $C \in \mathcal{E}_k$, D-AMP recovers all the signals in $C$ from $ \delta > \kappa$ measurements.

We now ask our uniform optimality question: {\em Does there exist any other signal recovery algorithm that can recover all the signals in all these classes with fewer measurements than {\rm D-AMP}?} If the answer is affirmative, then D-AMP is sub-optimal in the uniform sense, meaning there exists an approach that outperforms D-AMP uniformly over all classes in $\mathcal{E}_\kappa$. The following proposition shows that any recovery algorithm requires at least $m= \kappa n$ measurements for accurate recovery, i.e., D-AMP is optimal in this sense. 
\begin{proposition}\label{prop:optimalDAMP}  If $m^*$ denotes the minimum number of measurements required (by any recovery algorithm) for a set $C \in \mathcal{E}_\kappa$, then
\[
\sup_{C \in \mathcal{E}_{\kappa}} \frac{m^*(C)}{n} \geq {\kappa}. 
\]
\end{proposition}
\begin{proof}
First note that according to Example \ref{ex:denoiseperformk} any $\kappa n$ dimensional subspace of $\mathbb{R}^n$ belongs to $\mathcal{E}_{\kappa}$ (assume that $\kappa n$ is an integer). From the fundamental theorem of linear algebra we know that to recover the vectors in a $k$ dimensional subspace we require at least $k$ measurements. Hence 
\[
\sup_{C \in \mathcal{E}_{\kappa}} \frac{m^*(C)}{n} \geq \frac{\kappa n}{n}= \kappa. 
\]
\end{proof}

According to this simple result, D-AMP is optimal for at least certain classes of signals and certain denoisers. Hence, it cannot be uniformly improved.  \\

\subsubsection{Single class optimality}
The uniform optimality framework we introduced above considers a set of signal classes and measures the performance of an algorithm on every class in this set. However, in many applications such as imaging we are interested in the performance of D-AMP on a specific class of signals, such as images. Unfortunately, the uniform optimality framework does not provide any conclusion in such cases. Therefore, in this section we introduce another framework for evaluating the optimality of D-AMP that we call {\em single class optimality}.  \\

Let $C$ denote a class of signals. Instead of assuming that we are given a family of denoisers for the signals in  class $C$, we assume that we can find the denoiser that brings about the best performance from D-AMP. This ensures that D-AMP employs as much information as it can about $C$. Let $\theta_D^{\infty}(x_o, \delta, \sigma_w^2)$ denote the fixed point of the state evolution equation given in \eqref{eq:stateevo1}. Note that we have added a subscript $D$ to our notation for $\theta$ to indicate the dependence of this quantity on the choice of the denoiser. The best denoiser for D-AMP is a denoiser that minimizes $\theta_D^{\infty}(x_o, \delta, \sigma_w^2)$. Note that according to Finding \ref{find:mainone}, $\theta_D^{\infty}(x_o, \delta, \sigma_w^2)$ corresponds to the mean square error of the final estimate  that D-AMP returns.

\begin{definition}\label{def:minimaxdenoi}
A family of denoisers $D^*_{\sigma}$ is called minimax optimal for D-AMP at noise level $\sigma_w^2$, if it achieves
\[
\inf_{D_{\sigma}} \sup_{x_o \in C} \theta_D^{\infty}(x_o, \delta, \sigma_w^2). 
\]
\end{definition}
Note that according to our definition, the optimal denoiser may depend on both $\sigma_w^2$ and $\delta$ and it is not necessarily unique. We call the version of D-AMP that employs $D^*_{\sigma}$, ${\rm D}^*$-AMP. 

Armed with this definition, we formally ask the single class optimality question
: {\em Can we provide a new algorithm that can recover signals in class $C$ with fewer measurements than ${\rm D}^*$-{\rm AMP}?} If negative, it means that if we employ the optimal denoiser for D-AMP algorithm no other algorithm can outperform D-AMP. Unfortunately, we will show that there are signal classes for which D-AMP is not optimal in this sense. Our proof requires the following standard definition from statistics text books \cite{lehmann1998theory}: 

\begin{definition}\label{def:d-ampminimaxoptimal}
 The minimax risk of a set of signals $C$ at the noise level $\sigma^2$ is defined as
\[
R_{MM}(C, \sigma^2) = \inf_{D} \sup_{x_o \in C} \mathbb{E} \| D(x_o+ \sigma \epsilon) - x_o\|_2^2,
\]
where the expected value is with respect to $\epsilon \sim N(0, I)$. If $D^{M}_{\sigma}$ achieves $R_{MM}(C, \sigma^2)$, then it will be called the  family of minimax denoisers for the set $C$ under the square loss. 
\end{definition}

\begin{proposition}\label{prop:minimaxconnec}
The family of minimax denoisers for $C$ is a family of optimal denoisers for D-AMP. Furthermore, in order to recover every $x_o \in C$, ${\rm D}^*$-AMP requires at least $n \kappa_{MM}$ measurements:
\[
\kappa_{MM} = \sup_{\sigma^2>0}  \frac{ R_{MM}(\sigma^2)}{n \sigma^2}. 
\]
\end{proposition}
\begin{proof}
Since the proof of this result is slightly more involved, we postpone it to Appendix \ref{app:proofprop1}. 
\end{proof}

Based on this result, we can simplify the single class optimality question: {\em Does there exist any recovery algorithm that can recover every $x_o \in C$ from fewer observations than $n \kappa_{MM}$?} Unfortunately, the answer is affirmative. 

Consider the following extreme example. Let $B^n_k$ denote the class of signals that consist of $k$ ones and $n-k$ zeros. 
Define $\rho = k/n$ and let $\phi(z)$ denote the density function of a standard normal random variable.  

\begin{proposition}\label{prop:onemeasurementrecovery}
For very high dimensional problems, there are recovery algorithms that can recover signals in $B_k$ accurately from $1$ measurement. On the other hand, $D^*$-AMP requires at least $n (\kappa_{MM} -o(1))$ measurement to recover signals from this class, where
\begin{eqnarray}
 \kappa_{MM}= \sup_{\sigma^2>0} \frac{1}{\sigma^2}  \mathbb{E}_{z_1 \sim \phi} \left( \frac{\rho \phi_{\sigma}(z_1)}{\rho \phi_\sigma(z_1) +(1-\rho) \phi_{\sigma}(z_1+1) }-1 \right)^2\rho \nonumber \\
 + \mathbb{E}_{z_1 \sim \phi}  \left( \frac{\rho \phi_{\sigma}(z_1-1)}{\rho \phi_\sigma(z_1-1) +(1-\rho) \phi_{\sigma}(z_1) } \right)^2(1-\rho),\nonumber
\end{eqnarray}
where $\phi_{\sigma} (z) = \phi(z/\sigma).$
\end{proposition}
The proof of this result is slightly more involved and hence is postponed to Appendix \ref{app:proof theorem1}. According to this proposition, since $\kappa_{MM}$ is non-zero, the number of measurements $D^*$-AMP requires is proportional to the ambient dimension $n$, while the actual number of measurements that is required for recovery is equal to $1$. Hence, in such cases ${\rm D}^*$-AMP is sub-optimal. 

However, it is also important to note that while D-AMP is sub-optimal for this class, according to Proposition \ref{prop:optimalDAMP} D-AMP is optimal for other classes. Characterizing the classes of signals for which D-AMP is optimal is left as an open direction for future research. Despite this sub-optimality result, we will show in Section \ref{sec:DAMPinPractice} that D-AMP provides impressive results for the class of natural images and outperforms state-of-the-art recovery algorithms. 

\subsection{Additional miscellaneous properties of D-AMP}\label{sec:genericprop}

\subsubsection{Better denoisers lead to better recovery}
This intuitive result is a key feature of D-AMP.  We formalize it below.

\begin{theorem}
Let a family of denoisers $D_{\sigma}^1$ be a better denoiser than a family $D^2_{\sigma}$ for signal $x_o$ in the following sense:
\begin{equation}\label{eq:minimaxperformance}
  \frac{\mathbb{E} \| D^1_{\sigma} (x_o + \sigma \epsilon) - x_o\|^2_2}{n\sigma^2} \leq   \frac{\mathbb{E} \| D^2_{\sigma} (x_o + \sigma \epsilon) - x_o\|^2_2}{n\sigma^2}, \ \ \ \forall \sigma^2>0 .
\end{equation}
Also, let $\theta^\infty_{D_i}(x_o, \delta, \sigma_w^2)$ denote the fixed point of state evolution for denoiser $D_i$. Then,
\begin{equation*}
\theta^\infty_{D_1}(x_o, \delta, \sigma_w^2) \leq \theta^\infty_{D_2}(x_o, \delta, \sigma_w^2). 
\end{equation*}
\end{theorem}
\begin{proof}
The proof of this result is straightforward. Since, the state evolution of $D^1$ is uniformly lower than $D^2$, its fixed point is lower as well.   
\end{proof}

\subsubsection{D-AMP as a regularization technique}
Explicit regularization is a popular technique to recover signals from an undersampled set of linear  measurements \cite{Donoho1,BlockTV,PeyreNonLocal1,JianZhangImprovedTV,FowlerMH,JianZhangALSB,koseEntropyCS, dongcompressive}. In these approaches a cost function, $J(x)$, also known as a regularizer, is considered on $\mathbb{R}^n$. This function returns large values for $x \notin C$ and returns small values for $x \in C$. Regularized techniques recover $x_o$ from measurements $y$ by setting up and solving the following optimization problem:
\begin{equation}\label{eq:regularformul}
\hat{x}=\operatornamewithlimits{argmin}\limits_{x} \frac{1}{2} || y- \mathbf{A} x ||_2^2+\lambda J(x).
\end{equation}
 Since in many cases $J(x)$ is non-convex and non-differentiable, iterative heuristic methods have been proposed for solving the above optimization problem.\footnote{Many of these methods solve \eqref{eq:regularformul} accurately when $J$ is convex.}  D-AMP provides another heuristic approach for solving \eqref{eq:regularformul}. It has two main advantages over the other heuristics: (i) D-AMP can be analyzed by the state evolution theoretically. Hence, we can theoretically predict the number of measurements required and the noise sensitivity of D-AMP. (ii) The performance of most heuristic methods depend on their free parameters. As discussed in Section \ref{sec:DAMP_Params} there are efficient approaches for tuning the parameters of D-AMP optimally. Below we briefly review the application of D-AMP for solving \eqref{eq:regularformul}. 
 
 Assume that there exists a computationally efficient scheme for solving the optimization problem $\chi_J(u; \lambda) = \arg\min \frac{1}{2} \|u-x\|_2^2+ \lambda J(x)$. $\chi_J(u, \lambda)$ is called the proximal operator for the function $J$. The D-AMP algorithm for solving \eqref{eq:regularformul} is given by
 \begin{eqnarray}
 x^{t+1} &=& \chi_J(x^t+\mathbf{A}^*z^t; \lambda^t),  \\
z^t&=& y-\mathbf{A}x^t+ z^{t-1}{\rm div} \chi_J(x^{t-1}+\mathbf{A}^*z^{t-1}; \lambda^{t-1})/m. \nonumber 
 \end{eqnarray}
 Considering $\chi_J$ as a denoiser, this algorithm has exactly the same interpretation as our generic D-AMP algorithm. Furthermore, if the explicit calculation of the Onsager correction term is challenging we can employ the Monte Carlo technique that will be discussed in Section \ref{ssec:onsagerapprox}.

\section{Connection with other state evolutions}\label{sec:connectionSE}
In Section \ref{sec:Theory} we introduced a new, ``deterministic'' state evolution (SE) and used it to analyze D-AMP.
Here we review this SE and compare it with AMP's Bayesian SE, which was first introduced in \cite{MaDo09sp, DoMaMo09}.
\subsection{Deterministic state evolution}
The deterministic SE assumes that $x_o$ is an arbitrary but fixed vector in $C$.
Starting from $\theta^0 = \frac{\|x_o\|_2^2}{n}$, the deterministic SE generates a sequence of numbers through the following iterations:
\begin{equation}
\theta^{t+1}(x_o, \delta, \sigma_w^2) = \frac{1}{n}  \mathbb{E}_\epsilon \| D_{\sigma^t} (x_o + \sigma^t \epsilon) -x_o \|_2^2, \tag{\ref{eq:stateevo1}}
\end{equation}
where $(\sigma^t)^2 = \frac{1}{\delta}\theta^t(x_o, \delta, \sigma_w^2) + \sigma_w^2$ 
 and $\epsilon \sim N(0,I)$.

\subsection{Bayesian state evolution}
The Bayesian SE assumes that $x_o$ is a vector drawn from a probability density function (pdf) $p_x$, where the support of $p_x$ is a subset of $C$.
Starting from $\bar{\theta}^0 = \frac{\|x_o\|_2^2}{n}$, the Bayesian SE generates a sequence of numbers through the following iterations:
\begin{equation}\label{eq:stateevo2}
{\bar{\theta}}^{t+1}(p_x, \delta, \sigma_w^2) = \frac{1}{n} \mathbb{E}_{x_o, \epsilon} \|D_{\bar{\sigma}^t} (x_o + \bar{\sigma}^t \epsilon)-x_o\|_2^2 , 
\end{equation}
 where $(\bar{\sigma}^t)^2 = \frac{1}{\delta}\bar{\theta}^t (p_x, \delta, \sigma_w^2)+ \sigma_w^2$. 
  We have used the notation $\bar{\theta}$ to distinguish the Bayesian SE from its deterministic counterpart. In Definition \ref{def:d-ampminimaxoptimal} we presented a definition of the optimal denoiser  under the deterministic framework. One can do the same for the Bayesian framework.  
\begin{definition}\label{def:BayesOptDenosier}
A family of denoisers $\tilde{D}_{\sigma}$ is called Bayes-optimal for D-AMP at noise level $\sigma_w^2$, if it achieves
\[
\inf_{D_{\sigma}}  \bar{\theta}_D^{\infty}(p_x, \delta, \sigma_w^2). 
\]
\end{definition}
\noindent It is straightforward to see that the family $\tilde{D}_\sigma (x_o + \sigma \epsilon) = \mathbb{E} (x_o \ | \ x_o + \sigma \epsilon )$ is Bayes-optimal for D-AMP. 

 While the deterministic and Bayesian SEs are different, we can establish a connection between them by employing standard results in theoretical statistics regarding the connection between the minimax risk and the Bayesian risk. Next section briefly discusses this connection. 
 
 \subsection{Connection between the two state evolutions}
 In this section we would like to establish a connection between the fixed points of the Bayes-optimal denoisers and the minimax-optimal denoisers  for D-AMP. 
Let $\bar{\theta}^{\infty}(p_x, \delta, \sigma_w^2)$ denote the fixed point of the Bayesian SE \eqref{eq:stateevo2} associated with the family of Bayes-optimal denoisers from Definition \ref{def:BayesOptDenosier}. Also, let $\theta^{\infty}(x_o, \delta, \sigma_w^2)$ denote the fixed point of the deterministic SE \eqref{eq:stateevo1} for the family of minimax denoisers from Definition \ref{def:minimaxdenoi}. 
 
\begin{theorem}\label{thm:minimaxbayes}
Let $\mathcal{P}$ denote the set of all distributions whose support is a subset of $C$.  Then,
\[
 \sup_{p_x \in \mathcal{P}}  \bar{\theta}^{\infty}(p_x, \delta, \sigma_w^2)
\leq
\sup_{x_o \in C} \theta^{\infty}(x_o, \delta, \sigma_w^2) .
\]
\end{theorem}   
\begin{proof}
For an arbitrary family of denoisers $D_\sigma$ we have
\begin{equation}\label{eq:SEcomparison_one}
\mathbb{E}_{x_o, \epsilon}  \|D_{\sigma} (x_o + \sigma \epsilon) -x_o\|_2^2 \leq \sup_{x_o \in C} \mathbb{E}_\epsilon \| D_{\sigma} (x_o+ \sigma \epsilon) -x_o \|_2^2. 
\end{equation}
If we take the minimum with respect to $D_{\sigma}$ on both sides of \eqref{eq:SEcomparison_one}, we obtain the following inequality
\[
\mathbb{E}_{x_o, \epsilon}  \|\tilde{D}_{\sigma} (x_o + \sigma \epsilon) -x_o\|_2^2 \leq \sup_{x_o \in C} \mathbb{E}_\epsilon \| D^{MM} (x_o+ \sigma \epsilon) -x_o \|_2^2,
\]
where $D^{MM}$ denotes the minimax denoiser and $\tilde{D}_{\sigma}( \ x_o + \sigma \epsilon)$ denotes $\mathbb{E} (x_o  \ | \ x_o + \sigma \epsilon )$. 
Let $(\bar{\sigma}^{\infty})^2 = \frac{\bar{\theta}^{\infty} (x_o, \delta, \sigma_w^2)}{\delta}+ \sigma_w^2$  and $(\sigma_{mm}^{\infty})^2 = \frac{ \theta^{\infty}(x_o, \delta, \sigma_w^2)}{\delta}+ \sigma_w^2$. Also, for notational simplicity assume that $\sup_{x_o \in C} \mathbb{E}_\epsilon \| D^{MM} (x_o+ \bar{\sigma}^{\infty} \epsilon) -x_o \|_2^2$ is achieved at a certain value $x_{mm}$.  We then have
\begin{eqnarray}
\bar{\theta}^{\infty}(p_x, \delta, \sigma_w^2) &=& \frac{\mathbb{E}_{x_o, \epsilon}  \|\tilde{D}_{\bar{\sigma}^{\infty}} (x_o + \bar{\sigma}^{\infty} \epsilon) -x_o\|_2^2}{n} \nonumber \\ &\leq& \frac{\mathbb{E}_\epsilon \| D^{MM}_{\bar{\sigma}^{\infty}} (x_{mm}+ \bar{\sigma}^{\infty} \epsilon) -x_{mm} \|_2^2}{n}. \nonumber \\
\end{eqnarray} 

This inequality implies that $\bar{\theta}^{\infty}(p_x, \delta, \sigma_w^2)$ is below the fixed point of the deterministic SE using $D^{MM}$ at $x_{mm}$. Therefore, because $\sup_{x_o \in C} \theta^{\infty}(x_o, \delta, \sigma_w^2)$ will be equal to or above the fixed point of $D^{MM}$ at $x_{mm}$, it will satisfy $\sup_{p_x \in \mathcal{P}}  \bar{\theta}^{\infty}(p_x, \delta, \sigma_w^2) \leq \sup_{x_o \in C} \theta^{\infty}(x_o, \delta, \sigma_w^2)$. 
\end{proof}

\noindent Under some general conditions it is possible to prove that
\begin{eqnarray}\label{eq:bayesminimaxoptimality}
\lefteqn{\sup_{\pi \in \mathcal{P}} \inf_{D_{\sigma}} \mathbb{E}_{x_o, \epsilon}  \|D_{\sigma} (x_o + \sigma \epsilon) -x_o\|_2^2 \nonumber } &\\
&= \inf_{D_\sigma} \sup_{x_o \in C} \mathbb{E}_\epsilon \| D_{\sigma} (x_o+ \sigma \epsilon) -x_o \|_2^2. 
\end{eqnarray}

\noindent For instance, if we have
\begin{eqnarray}
\lefteqn{\sup_{\pi \in \mathcal{P}} \inf_{D_{\sigma}} \mathbb{E}_{x_o, \epsilon}  \|D_{\sigma} (x_o + \sigma \epsilon) -x_o\|_2^2 \nonumber} & \\
&= \inf_{D_{\sigma}}   \sup_{\pi \in \mathcal{P}}  \mathbb{E}_{x_o, \epsilon}  \|D_{\sigma} (x_o + \sigma \epsilon) -x_o\|_2^2, \nonumber
\end{eqnarray}
then \eqref{eq:bayesminimaxoptimality} holds as well. Since we work with square loss in the SE, swapping the infimum and supremum is permitted under quite general conditions on $\mathcal{P}$. For more information, see Appendix A of \cite{johnstone2002function}. If \eqref{eq:bayesminimaxoptimality} holds, then we can follow similar steps as in the proof of Theorem $\ref{thm:minimaxbayes}$ to prove that under the same set of conditions we can have
\[
\sup_{p_x \in \mathcal{P}} \inf_{D_{\sigma}}  \bar{\theta}^{\infty}(p_x, \delta, \sigma_w^2) =
\inf_{D_{\sigma}} \sup_{x_o \in C} \theta^{\infty}(x_o, \delta, \sigma_w^2).
\]
In words, the supremum of the fixed point of the Bayesian SE with the Bayes-optimal denoiser is equivalent to the supremum of the fixed point of the deterministic SE with the minimax denoiser.

\subsection{Why bother?}
Considering that the deterministic and Bayesian SEs look so similar, and under certain conditions have the same suprememums, it is natural to ask why we developed the deterministic SE at all.  That is, what is gained by using SE \eqref{eq:stateevo1} rather than \eqref{eq:stateevo2}?

The deterministic SE is useful because it enables us to deal with signals with poorly understood distributions.  Take, for instance, natural images.  To use the Bayesian SE on imaging problems, we would first need to characterize all images according to some generalized, almost assuredly inaccurate, pdf. In contrast, the deterministic SE deals with specific signals, not distributions.  Thus, even without knowledge of the underlying distribution, so long as we can come up with representative test signals, we can use the deterministic SE.  Because the SE shows up in the parameter tuning, noise sensitivity, and performance guarantees of AMP algorithms, being able to deal with arbitrary signals is invaluable.

\section{Calculation of the Onsager correction term}\label{sec:onsager}
So far, we have emphasized that the key to the success of approximate message passing algorithms is the Onsager correction term, $z^{t-1}{\rm div} D_{\hat{\sigma}^{t-1}}(x^{t-1}+\mathbf{A}^*z^{t-1}) /m$, but we have not yet addressed how one can calculate it for an arbitrary denoiser. 
In this section we provide some guidelines on the calculation of this term. If the input-output relation of the denoiser is known explicitly, then calculating the divergence, ${\rm div} D(x)$, and thus the Onsager correction term, is usually straightforward.\footnote{In the context of this work the divergence ${\rm div} D(x)$ is simply the sum of the partial derivatives with respect to each element of $x$, i.e., ${\rm div} D(x)=\sum\limits_{i=1}^{n} \frac{\partial D(x_i)}{\partial x_i}$.} We will review some popular denoisers and calculate the corresponding Onsager correction terms in the next section. However, most state-of-the-art denoisers are complicated algorithms for which the input-output relation is not explicitly known.  In Section \ref{ssec:onsagerapprox} we show that, even without an explicit input-output relationship, we can calculate a good approximation for the Onsager correction term.

\subsection{Soft-thresholding for sparse and low-rank signals}\label{ssec:onsagerexact}
Three of the most popular signal classes in the literature are sparse, group-sparse, and low-rank signals (when the signal has a matrix form). The most popular denoisers for these signals are soft-thresholding, block soft-thresholding, and singular value soft-thresholding, respectively. The goal of this section is to derive the Onsager correction term for each of these denoisers. Most of the results mentioned in this section have been derived elsewhere. We summarize these results to help the reader understand the steps involved in explicitly computing the Onsager correction term. 

\subsubsection{Soft-thresholding} Let $\eta_{\tau}$ denote the soft-threshold function. $\eta_{\tau}(x)$ for $x \in \mathbb{R}^n$ denotes the component-wise application of the soft-threshold function to the elements of $x$. In this case we have ${\rm div}\eta_{\tau}(x) = \sum_{i=1}^n \mathbb{I} (|x_i|> \tau)$, where $\mathbb{I}$ denotes the indicator function. 

\subsubsection{Block soft-thresholding} For a vector $x_B \in \mathbb{R}^B$ block soft-thresholding is defined as $\eta^B_\tau(x_B) = (\|x_B\|_2 - \tau) \frac{x_B}{\|x_B\|_2}$.  In other words, the threshold function retains the phase of the vector $x_B$ and shrinks its magnitude toward zero. Let $n =MB$ and $x = [(x_B^1)^T, (x_B^2)^T, \ldots, (x_B^M)^T]^T$. The notation $\eta^B_\tau(x)$ is defined as the block soft-thresholding function that is applied to each individual block. The divergence of block soft-thresholding can then be calculated according to
 \[
{\rm div} \eta^B_{\tau}(x)=\sum_{\ell=1}^M \left(B- \frac{(B-2)^2}{\|x_B^\ell\|_2^2} \right)\mathbb{I} (\|x_B^\ell)\|_2 \geq \tau). 
 \]
This result was derived in \cite{ AnMaBa12,MaAnYaBa11, DonohoJM13}. 

\subsubsection{ Singular value thresholding} Let $\mathbf{X}_o \in \mathbb{R}^{n \times n}$ denote our signal of interest. If $\mathbf{X}_o$ is low-rank then it can be estimated accurately from its noisy version $\mathbf{\Phi}= \mathbf{X}_o+ \sigma \mathbf{W}$ where $\mathbf{W}_{ij}$ denote i.i.d., $N(0,1)$ random variables. If the singular value decomposition of $\mathbf{\Phi}$ is given by $\mathbf{USV}^T$, with $\mathbf{S} = {\rm diag}(\sigma_1, \ldots, \sigma_n)$, where $\sigma_i$s denote the singular values of $\mathbf{\Phi}$, then the estimate of $\mathbf{X}_o$ has the form
\begin{eqnarray}
\hat{\mathbf{X}} &=& \mathbf{SVT}_{\lambda} (\mathbf{\Phi}) \nonumber \\
&=& \mathbf{U} {\rm diag} ((\sigma_1 - \lambda)_+, (\sigma_2- \lambda)_+, \ldots, (\sigma_n - \lambda)_+)\mathbf{V}^T,\nonumber
\end{eqnarray}
in which $\lambda$ is a regularization parameter that can be optimized for the best performance. Again this denoiser can be employed in the D-AMP framework to recover low-rank matrices from their underdetermined set of linear equations. To calculate the Onsager correction term  we should compute ${\rm div} \mathbf{SVT}_{\lambda} (\mathbf{\Phi}) $. According to \cite{candes2012unbiased} the divergence of singular value thresholding is given by
\[
{\rm div} \mathbf{SVT}_{\lambda} (\mathbf{\Phi}) = \sum_{i=1}^n \mathbb{I} (\sigma_i > \lambda) +2 \sum_{i, j=1, i \neq j}^n \frac{\sigma_i (\sigma_i - \lambda)_+}{\sigma_i^2 - \sigma_j^2}. 
\]

\subsection{Monte Carlo method}\label{ssec:onsagerapprox}
While simple denoisers often yield a closed form for their divergence, high-performance denoisers are often data dependent; making it very difficult to characterize their input-output relation explicitly. Here we explain how a good approximation of the divergence can be obtained in such cases. This method relies on a Monte Carlo technique first developed in \cite{RamaniMCSure}.   The authors of that work showed that given a denoiser $D_{\sigma,\tau}(x)$, using an i.i.d. random vector $b \sim N(0,I)$, we can estimate the divergence with
\begin{eqnarray}
\label{eqn:div_est}
{\rm div} D_{{\sigma},\tau}&=&\lim\limits_{\epsilon\rightarrow 0}\mathbb{E}_{b} \left\{b^{*}\left(\frac{D_{{\sigma},\tau}(x+\epsilon b)-D_{{\sigma},\tau}(x)}{\epsilon}\right)\right\}\nonumber \\ &\approx& \mathbb{E}_b \left( \frac{1}{\epsilon}b^{*}(D_{{\sigma},\tau}(x+\epsilon b)-D_{{\sigma},\tau}(x))  \right),\nonumber \\
&&\text{ for very small $\epsilon$.} \nonumber
\end{eqnarray}

The only challenge in using this formula is calculating the expected value. This can be done efficiently using Monte Carlo simulation. We generate $M$ i.i.d., $N(0,I)$ vectors $b^1, b^2, \ldots, b^M$. For each vector $b^i$ we obtain an estimate of the divergence $\widehat{\rm div}^i$. We then obtain a good estimate of the divergence by averaging
\[
{\rm div} \hat{D}_{{\sigma},\tau} = \frac{1}{M} \sum_{i=1}^M \widehat{\rm div}^i. 
\]
According to the weak law of large numbers, as $M \rightarrow \infty$ this estimate converges to $\mathbb{E}_b \left( \frac{1}{\epsilon}b^{*}(D_{{\sigma},\tau}(x+\epsilon b)-D_{{\sigma},\tau}(x))  \right)$. When dealing with images, due to the high dimensionality  of the signal, we can accurately approximate the expected value using only a single random sample.  That is, we can let $M=1$.\footnote{When dealing with short signals ($n<1000$), rather than images, we found that using additional Monte Carlo samples produced more consistent results.} Note that in this case the calculation of the Onsager correction term is quite efficient and requires only one additional application of the denoising algorithm. In all of the simulations in this paper we have used either the explicit calculation of the Onsager correction term or the Monte Carlo method with $M=1$. 

\section{Smoothing a denoiser}\label{sec:Smooth}
\begin{figure}[t]
\begin{center}
  \includegraphics[width=.5\textwidth]{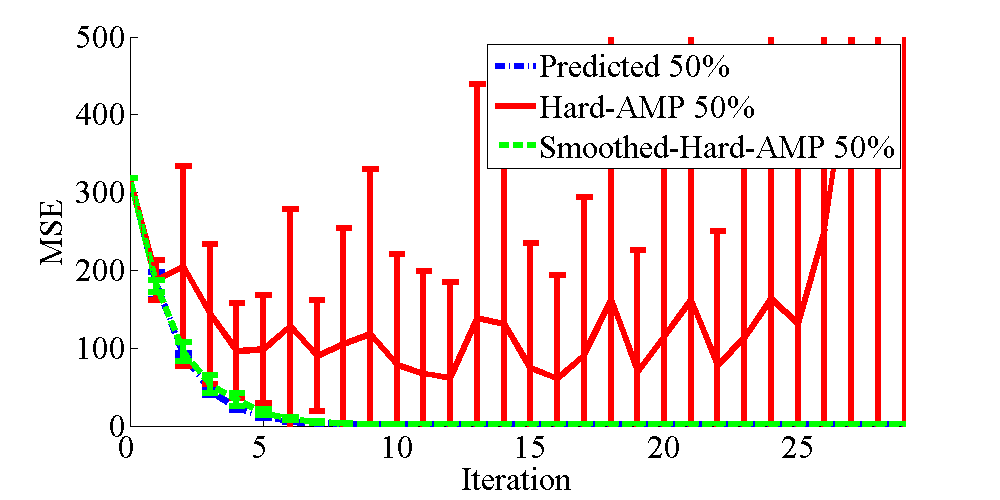} 
  \caption{Predicted and observed intermediate MSEs of hard-thresholding-based AMP, with and without smoothing.  Notice that the smoothed version is well predicted by the state evolution whereas hard-thresholding-based AMP without smoothing is not.  This discrepancy is due the fact that the hard-thresholding denoiser is not continuous.  
  }
   \label{fig:HTStateEvolution} 
  \end{center}
\end{figure}

\begin{figure*}[t]
	\begin{center}
		\includegraphics[width=\textwidth]{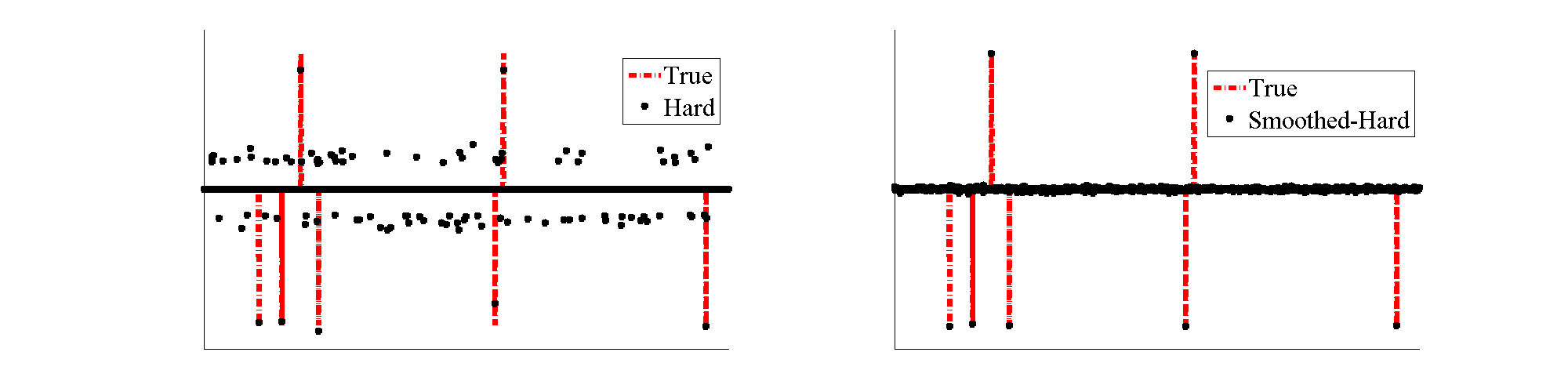} 
		\caption{Reconstructions of a sparse signal that was sampled at a rate of $\delta=1/3$. Notice that D-AMP based on smoothed-hard-thresholding successfully reconstructs the signal whereas D-AMP based on hard-thresholding, does not. The failure of hard-thresholding-based D-AMP is due to the discontinuity of the hard-thresholding denoiser.  
		}
		\label{fig:HTPerformance} 
	\end{center}
\end{figure*}

The denoiser used within D-AMP can take on almost any form.  However, the state evolution predictions are not necessarily accurate if the denoiser is not Lipschitz continuous.  This requirement seems to disallow some popular denoisers with discontinuities, such as hard-thresholding.  Figure \ref{fig:HTStateEvolution} compares the state evolution predictions for the hard thresholding denoiser alongside the actual performance of D-AMP based on hard thresholding; the state evolution predictions fail entirely. One simple idea to resolve this issue is to ``smooth'' the denoisers. The smoothed version should behave nearly the same as the original denoiser but, because it has no discontinuities, should satisfy the state evolution equations. 
The concept of smoothing simple denoisers and this process's effects on the performance of simple denoisers has been analyzed in \cite{deledalle2013stein, zheng2015does}. Here we explain how smoothing can be applied in practice. 

Let $\eta(x)$ be a discontinuous denoiser. Now define a new denoiser $\tilde{\eta}(x)$ as follows
\begin{equation}\label{eqn:GaussSmoothDenoiser}
\tilde{\eta}(x)=\int\limits_{\zeta \in \mathbb{R}^n} \eta(x-\zeta) \frac{1}{(2 \pi)^{n/2} r^{n}} e^{-\|\zeta\|_2^2/2r^2}d\zeta,
\end{equation}
where $d\zeta = d\zeta_1 d\zeta_2\ldots d \zeta_n$. 

The denoiser $\tilde{\eta}(x)$ is simply the convolution of $\eta(x)$ with a Gaussian kernel with standard deviation $r$. Note that the width $r$ dictates the amount of smoothness we apply to $\eta$. Larger values of $r$ lead to a smoother $\tilde{\eta}$. Below we present a simple lemma that proves $\tilde{\eta} (x)$ is in fact smooth.

Suppose that ${\eta}(x_1,x_2, \ldots, x_n)$ satisfies the following condition:
\begin{eqnarray}\label{eq:growthcondition}
\int |\eta_i (\tilde{\zeta}) \tilde{\zeta_i}| {\rm e}^{-\frac{\|\tilde{\zeta}\|_2^2}{4r^2}} d\tilde{\zeta}< \infty.
\end{eqnarray} 
Note that this condition implies that $\eta_i$ is not growing very fast as $\zeta_1, \zeta_2, \ldots, \zeta_n \rightarrow \infty$. 

\begin{lemma}
	If $\eta$ satisfies \eqref{eq:growthcondition}, then $\tilde{\eta}(x)$ defined in \eqref{eqn:GaussSmoothDenoiser} is continuously differentiable with a bounded derivative and is thus Lipschitz continuous.
\end{lemma}
\begin{proof} To prove $\tilde{\eta}$ is continuously differentiable with bounded derivative, we prove that all the partial derivatives exist, are bounded, and are continuous. 
Let $\eta = (\eta_1, \eta_2, \ldots, \eta_n)$ and  $\tilde{\eta} = (\tilde{\eta}_1, \tilde{\eta}_2, \ldots, \tilde{\eta}_n)$. By a simple change of integration variables we obtain
\begin{equation}\label{eqn:GaussSmoothDenoiser2}
\tilde{\eta}(x)=\int\limits_{\tilde{\zeta} \in \mathbb{R}^n} \eta(\tilde{\zeta}) \frac{1}{(2 \pi)^{n/2} r^{n}} e^{-\|x-\tilde{\zeta}\|_2^2/2r^2}d\tilde{\zeta}.
\end{equation}
Now we calculate the $j^{th}$ partial derivative of $\tilde{\eta}_i$. Let $b^j \in \mathbb{R}^n$ denote a vector whose elements are all zero except for the $j^{\rm th}$ element, which is equal to one. Then,
\begin{eqnarray}\label{eq:firststepder}
\lefteqn{\frac{\partial \tilde{\eta}_i(x)}{\partial x_j}} \nonumber \\
&=&
\!\!\!\! \lim_{\gamma \rightarrow 0} \frac{\tilde{\eta}_i(x+ \gamma b^j) -\tilde{\eta}_i(x) }{\gamma} \nonumber \\
&=&\!\!\!\! \lim_{\gamma \rightarrow 0} \int \limits_{\tilde{\zeta} \in \mathbb{R}^n}  \frac{\eta_i(\tilde{\zeta})}{(2 \pi)^{n/2} r^{n}} \left( \frac{ e^{\frac{-\|x+\gamma b^j- \tilde{\zeta}\|_2^2}{2r^2}}-e^{-\frac{\|x- \tilde{\zeta}\|_2^2}{2r^2}} }{\gamma}\right)d\tilde{\zeta}.\nonumber \\
\end{eqnarray}
From the mean value theorem we conclude that there exists $\tilde{\gamma}$ between $0$ and $\gamma$ such that
\begin{eqnarray}
\lefteqn{\left( \frac{ e^{-\|x+\gamma b^j- \tilde{\zeta}\|_2^2/2r^2}-e^{-\|x- \tilde{\zeta}\|_2^2/2r^2} }{\gamma}\right)} \nonumber \\
& =& \frac{(\tilde{\zeta}_j - \tilde{\gamma} - x_j)}{r^2} e^{-\|x+\tilde{\gamma} b^j- \tilde{\zeta}\|_2^2/2r^2},
\end{eqnarray}
where the last equality is due to the fact that the $j^{\rm th}$ element of $b^j$ is equal to one.  Also, note that for $|\gamma|<1$ we have
\begin{eqnarray}
\lefteqn{\left| \frac{\tilde{\zeta}_j - \tilde{\gamma}  - x_j}{r^2} e^{-\|x+\tilde{\gamma} b^j- \tilde{\zeta}\|_2^2/2r^2} \right|} \nonumber\\
&\leq& \frac{|\tilde{\zeta}_j| + 1 + |x_j|}{r^2} e^{-\sum_{k \neq j} (x_k- \tilde{\zeta}_k)^2/2r^2}e^{\frac{-\tilde{\zeta}_j^2+ 2(|x_j|+1)|\tilde{\zeta}_j| }{2 r^2}},\nonumber
\end{eqnarray}
where to obtain the last equality we used the fact that all the element of $b_j$ except the $j^{\rm th}$  one are zero.  
Define
\[
h(\tilde{\zeta}) = \frac{|\tilde{\zeta}_j| + 1 + |x_j|}{r^2} e^{-\sum_{k \neq j} (x_k- \tilde{\zeta}_k)^2/2r^2}e^{\frac{-\tilde{\zeta}_j^2+ 2(|x_j|+1)|\tilde{\zeta}_j| }{2 r^2}}.
\]
It is straightforward to use \eqref{eq:growthcondition} and check that $$\int \left| \eta_i(\tilde{\zeta}) \frac{1}{(2 \pi)^{n/2} r^{n}}  h(\tilde{\zeta})  \right| d \tilde{z} < \infty.$$ 

So far we have proved that the absolute value of the integrand in \eqref{eq:firststepder} is less than or equal to $ \frac{C}{(2 \pi)^{n/2} r^{n}}  h(\tilde{\zeta})$, which is an integrable function. Hence, we can employ the dominated convergence theorem to show that

\begin{eqnarray}
\lefteqn{ \frac{\partial \tilde{\eta}_i(x)}{\partial x_j}
	=\lim_{\gamma \rightarrow 0} \frac{\tilde{\eta}_i(x+ \gamma b^j)-\tilde{\eta}_i(x) }{\gamma}} \nonumber \\
 &=& \int\limits_{\tilde{\zeta} \in \mathbb{R}^n}  \frac{ \eta_i(\tilde{\zeta})}{(2 \pi)^{n/2} r^{n}} \lim_{\gamma \rightarrow 0}  \left( \frac{ e^{\frac{-\|x+\gamma b^j- \tilde{\zeta}\|_2^2}{2r^2}}-e^{-\frac{\|x- \tilde{\zeta}\|_2^2}{2r^2}} }{\gamma}\right)d\tilde{\zeta} \nonumber \\
&=&\int\limits_{\tilde{\zeta} \in \mathbb{R}^n}   \frac{\eta_i(\tilde{\zeta})}{(2 \pi)^{n/2} r^{n}} \frac{\tilde{\zeta}_j - x_j}{r^2} e^{-\|x- \tilde{\zeta}\|_2^2/2r^2} d\tilde{\zeta}. 
\end{eqnarray}
It is straightforward to conclude that this derivative is bounded. Proving the continuity of the derivative employs the same lines of reasoning and hence we skip it.
\end{proof}

Calculating $\tilde{\eta}(x)$ from $\eta (x)$ is not straightforward for the following two reasons: (i) Equation \eqref{eqn:GaussSmoothDenoiser} dictates that we integrate over all of $\mathbb{R}^n$, and (ii) We usually do not have access to the explicit form of $\eta$.  To get around this problem we again turn to Monte Carlo sampling.

To approximately calculate \eqref{eqn:GaussSmoothDenoiser} using Monte Carlo sampling first generate a series of $M$ random vectors $h^1, h^2, ... h^M$, each with i.i.d.\ Gaussian elements with standard deviation $r$.  Next, for each $h^i$, pass $h^i+x$ through the discontinuous denoiser $\eta(x)$ and then average the results to get a smooth denoiser $\hat{\tilde{\eta}}(x)$.  That is approximate $\tilde{\eta}(x)$ with

\begin{equation}\label{eq:smoothdenoiserapproxGauss}
\hat{\tilde{\eta}}(x)= \frac{1}{M}\sum_{i=1}^{M}\eta(x+h^i),
\end{equation}

where $h^i \sim N(0,r^2\mathbf{I})$ for all $i$.

Figure \ref{fig:HTandSmoothedHTplot} compares the input-output relationship of the hard-thresholding denoiser before and after it has been smoothed using this method.  Notice the smoothing process completely removes the discontinuities but otherwise leaves the function intact.


The above discussion does not provide any suggestion on how we should pick the smoothing parameter $r$. In fact, rigorous study of the effect of $r$ in AMP requires the evaluation of the difference  $|\frac{\|x^t-x_o\|_2^2}{N}- \theta^t(x_o, \delta, \sigma_w^2)|$, in terms of the dimension. We leave it as an open problem for future research. Nevertheless, from a practical perspective $r$ can be considered as just another denoiser parameter. The problem of optimizing denoiser parameters has been extensively studied in the field of image processing \cite{zhu2010automatic, GalatsanosParams, RamaniMCSure}.

Figures \ref{fig:HTStateEvolution} and \ref{fig:HTPerformance} demonstrate the benefits of smoothing a denoiser using this approach.  Unlike D-AMP using the original hard-thresholding denoiser, D-AMP using the smoothed denoiser closely follows the state evolution.  This change is significant because it allows us to take advantage of the theory and tuning strategies developed in Section \ref{sec:Theory}.  More importantly, Figure \ref{fig:HTPerformance} illustrates how D-AMP based on smoothed-hard-thresholding dramatically outperforms its discontinuous counterpart.  

Before proceeding, we would like to emphasize that the above process is {\em not} needed for any of the advanced denoisers that we explored in this paper.  We found that advanced denoisers satisfy the state evolution and perform exceptionally in D-AMP without any smoothing.  We believe this finding implies they are sufficiently smooth to begin with.







\begin{figure}[t]
	\begin{center}
		\includegraphics[width=.5\textwidth]{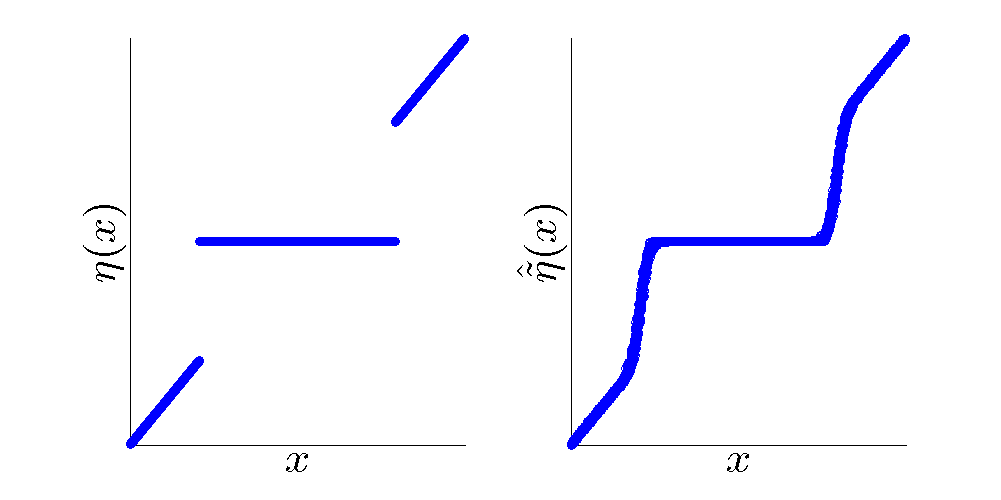} 
		\caption{Hard-thresholding and smoothed-hard-thresholding denoisers.  Note how the smoothing process has removed the discontinuities.  
		}
		\label{fig:HTandSmoothedHTplot} 
	\end{center}
\end{figure}

\section{Simulation results for imaging applications}\label{sec:DAMPinPractice}
To demonstrate the efficacy of the D-AMP framework, we evaluate its performance on imaging applications. 

\subsection{A menagerie of image denoising algorithms}\label{sec:reviewdenoise}
As we have discussed so far, D-AMP employs a denoising algorithm for signal recovery problems. In this section,  we briefly review some well-known image denoising algorithms that we would like to use in D-AMP. We later demonstrate that any of these denoisers, as well as many others, can be used within our D-AMP algorithm in order to reconstruct various compressively sampled signals. As we discussed in Section \ref{sec:genericprop}, theory says that if denoising algorithm $\cal{M}$ outperforms denoising algorithm $\cal{N}$, then D-AMP based on $\cal{M}$ will outperform D-AMP based on $\cal{N}$. We will see this behavior in our simulations as well. 

Below we represent a noisy image with the vector $f$; $f=x+ \sigma z$ where $x$ is the noise-free version of the image, $\sigma$ is the standard deviation of the noise, and the elements of $z$ follow an i.i.d.\ Gaussian distribution. 

\begin{enumerate}

\item Gaussian kernel regression: One of the simplest and oldest denoisers is Gaussian kernel regression, which is implemented via a Gaussian filter.  As the name suggests, it simply applies a filter whose coefficients follow a Gaussian distribution to the noisy image. It takes the form:
\begin{equation}
\label{eqn:GaussianFilter}
\hat{x}=f \star G ,
\end{equation}
where $G$ and $\star$ denote the Gaussian kernel and the convolution operator, respectively. The Gaussian filter operates under the model that a signal is smooth.  That is, neighboring pixels should have similar values. Note that this assumption is violated on image edges and hence this denoiser tends to over-smooth them. Compared to other approaches Gaussian kernel regression has a very low implementation cost. However, it does not remove noise as well as other denoisers.

\item Bilateral filter: Similar to kernel regression, the bilateral filter \cite{tomasi1998bilateral} sets each pixel value according to a weighted average of neighboring pixels.  However, whereas the Gaussian filter computes weights based on how close to one another two pixels are, the bilateral filter computes weights based on the similarity of the pixel values  (in addition to their spatial proximity).  The estimate produced by the bilateral filter can be written as

\begin{equation}
\label{eqn:NonLocalFilter}
\hat{x}(i)=\frac{\sum_{j \in \Omega_i}{w(i,j) f(j)}}{\sum_{j \in \Omega_i}{w(i,j)}} 
\end{equation}

\begin{equation}
\label{eqn:BilateralFilterWeights}
w(i,j)={\rm e}^{\frac{-(f(i)-f(j))^2}{h^2}},
\end{equation}
	
where $f(j)$ is the value of the $j^{th}$ pixel, $\Omega_i$ is a search window around pixel $i$, and $h$ is a smoothing parameter set according to the amount of noise in the signal.
Note that the bilateral filter tries to avoid averaging together light and dark pixels on opposite sides of an edge.  The bilateral filter has generally proven much more effective than the Gaussian filter. However, it fails entirely when a very large amount of noise is present and the denoiser cannot determine which pixels should be alike. 

\item Non-local means (NLM): Non-local means \cite{Buades05areview} extends the bilateral filter concept of averaging pixels with similar values to pixels with similar neighborhoods.  NLM's original implementation takes the same form as the bilateral filter \eqref{eqn:NonLocalFilter} but with the following weights:
\begin{equation}
\label{eqn:NLMFilterWeigths}
w(i,j)={\rm e}^{\frac{-\|N(i)-N(j)\|_2^2}{h^2}},
\end{equation}
where $N(i)$ represents a patch of pixels neighboring pixel $i$ and $h$ is a smoothing parameter set according to the variance of the noise.
Because the true value of a pixel is better reflected by the noisy value of its neighborhood than by just its noisy pixel value, NLM better recognizes which pixels should be alike and thus outperforms the bilateral filter.  However,  because two pixels on opposite sides of an edge usually have very similar neighborhoods, NLM still produces artifacts around edges\cite{maleki2012suboptimality}.

\item Wavelet thresholding: Wavelet thresholding \cite{Donoho94idealspatial} denoises natural images by assuming they are sparse in the wavelet domain. It transforms signals into a wavelet basis, thresholds the coefficients, and then inverses the transform. Hence if $\Psi'$ and $\Psi$ denote the wavelet transform and its inverse, respectively, then the denoised image is given by
\begin{equation}
\label{eqn:wavelet_thresholding}
\hat{x}=\Psi(\eta_\tau(\Psi' f)),
\end{equation}
where $\eta$ is some sort of thresholding function. The two most popular thresholding techniques are soft-thresholding $\eta_\tau^s(x) = (|x|- \tau)_+ {\rm sign}(x)$ and hard-thresholding $\eta_\tau^h(x) = (x)\mathbb{I} (|x| \geq \tau)$. Wavelet thresholding has superb performance if the signal is sparse in the wavelet domain.  Unfortunately, images do not have an exactly sparse wavelet representation.  As a result, the performance of wavelet thresholding denoising is generally worse than NLM.

\item BLS-GSM: Bayes least squares Gaussian scale mixtures \cite{PortillaBLSGSM} extends simple wavelet thresholding by using an overcomplete wavelet basis and computing denoised coefficient values not with a thresholding function, but via a Bayesian least squares estimate.  This estimate is computed by considering a neighborhood around every coefficient and then modeling the distribution of the coefficients within that neighborhood as the product of a Gaussian random vector and a random scalar, each with a carefully defined prior. The algorithm uses these priors to compute the expected value of the noiseless coefficient value. Because the distributions of the wavelet coefficients of natural images are highly dependent on one another, a Bayesian least squares estimate can remove noise while retaining far more structure than coefficient thresholding alone. Accordingly, BLS-GSM significantly outperforms wavelet thresholding. Its performance relative to NLM depends on the statistics of the image being denoised.

\item BM3D: Block matching 3D collaborative filtering \cite{DabovBM3D} can be considered a combination of NLM and wavelet thresholding.  The algorithm begins by comparing patches around the pixels in an image and then grouping similar patches into stacks.  It then performs 2D and 1D transforms on the group.  These transforms are a 2D DCT and a 1D Haar transform or a 2D bi-orthogonal spline wavelet (Bior) and a 1D Haar transform.  Which pair is used depends on the amount of noise in the image.  Next the algorithm shrinks the coefficients of these groups and performs an inverse transform to estimate each pixel.  It performs this process twice; once by hard-thresholding the coefficients and a second time using Wiener filtering based on the spectra of the initial estimate. In practice BM3D significantly outperforms NLM and wavelet thresholding techniques.  It does a great job at removing noise and produces fewer artifacts than competing methods.  Additionally, the authors of BM3D have provided well optimized code that makes this complicated algorithm quite efficient.
\item BM3D-SAPCA: BM3D with shape adaptive principal component analysis \cite{DabovBM3DSAPCA} combines two extensions to the original BM3D algorithm; block matching using shape adaptive patches and thresholding/filtering in a PCA derived basis. Using adaptive patches helps ensure that the algorithm groups only similar patches.  The use of an adaptive basis means that features not well captured by the DCT/Bior and Haar bases of BM3D will be retained. The performance of BM3D-SAPCA tends to be incrementally better than BM3D.  Unfortunately, this small increase in performance comes at a huge increase in computational cost.
\end{enumerate}
Table \ref{tab:denoisers} provides a comparison among the above denoising algorithms. The parameters of the Gaussian filter, the bilateral filter, non-local means, and wavelet thresholding were all experimentally tuned so as to maximize PSNR.\footnote{PSNR stands for peak signal-to-noise ratio and is defined as $10 \log_{10}(\frac{255^2}{{\rm mean}((\hat{x}-x_o)^2)})$ when the pixel range is 0 to 255.  It is a measure of how closely a signal estimate $\hat{x}$ is to the true signal $x_o$.  In this paper we use PSNR to measure both the denoising algorithms' and CS recovery algorithms' rescaled MSE.}  The parameters for the other 3 algorithms were set automatically using their respective packages.  The BM3D, BM3D-SAPCA, and BLS-GSM packages are available online.\footnote{\href{http://www.cs.tut.fi/~foi/GCF-BM3D/}{http://www.cs.tut.fi/\texttildelow{}foi/GCF-BM3D/}}
\footnote{\href{http://www.io.csic.es/PagsPers/JPortilla/software}{http://www.io.csic.es/PagsPers/JPortilla/software}}

\begin{table*}[t] 
	\caption{Performance (PSNR in dB) and computation time (Seconds) comparison of several denoisers.  Results are for $256 \times 256$ images with additive white Gaussian noise with standard deviation 15.} 
	\centering 
	\begin{tabular}{r|rrrrrr|r}
		\toprule
		\textbf{Denoiser} & \textbf{Lena} & \multicolumn{1}{c}{\textbf{Barbara}} & \multicolumn{1}{c}{\textbf{Boat}} & \multicolumn{1}{c}{\textbf{Fingerprint}} & \multicolumn{1}{c}{\textbf{House}} & \multicolumn{1}{c}{\textbf{Peppers}} & \textbf{Average Time} \\
		\midrule
		Gaussian filter & 26.5  & 25.9  & 24.4  & 18.0  & 28.1  & 24.2  & 0.005 \\
		Bilateral filter & 27.9  & 27.2  & 27.5  & 25.6  & 29.3  & 27.6  & 1.430 \\
		Non-local means & 31.3  & 30.8  & 30.0  & 27.6  & 32.8  & 30.8  & 7.507 \\
		Wavelet thresholding & 28.9  & 28.3  & 28.2  & 25.7  & 29.5  & 28.8  & 0.063 \\
		BLS-GSM & 32.4  & 30.7  & 30.9  & 27.8  & 33.8  & 31.9  & 14.548 \\
		BM3D  & 33.2  & 32.4  & 31.2  & 28.3  & 35.1  & 32.6  & 1.128 \\
		BM3D-SAPCA & 33.5  & 32.8  & 31.5  & 28.6  & 35.3  & 32.9  & 1251.633 \\
		\toprule
	\end{tabular}%
	\label{tab:denoisers}
\end{table*}

\subsection {Implementation details of D-AMP and D-IT}

\subsubsection{Terminology}
Our goal is to plug each of the denoising algorithms that we reviewed in Section \ref{sec:reviewdenoise} into our D-AMP algorithm.  In the rest of the paper we use the following terminology: If denoising method $\cal{M}$ is employed in D-AMP, then we call the reconstruction algorithm $\cal{M}$-AMP. For instance, if we use NLM, the resulting algorithm will be called NLM-AMP and if we use BM3D, the resulting algorithm will be called BM3D-AMP. 

We use the same terminology for D-IT: If we use the BM3D denoiser then we call the resulting algorithm BM3D-IT. 
  
\subsubsection{Denoising parameters}\label{ssec:tunepractice}
One of the main challenges in comparing different recovery algorithms is the tuning of each algorithm's free parameters. As discussed in Section \ref{sec:DAMP_Params}, the parameters of D-AMP can be  tuned efficiently with a greedy strategy. In other words, at every iteration we optimize the parameters to obtain the minimum MSE at that iteration. 
Toward this goal, we can employ different strategies that have been proposed in the literature for setting the parameters of denoising algorithms \cite{zhu2010automatic, GalatsanosParams, RamaniMCSure}.

A variety of techniques exist to estimate the standard deviation of the noise in an image; however, we tackled this problem by using a convenient feature of AMP algorithms: $||z^t||_2^2/m \approx (\sigma^t)^2$ \cite{MalekiThesis}. Additionally, the packages provided with many of the state-of-the-art denoising algorithms \cite{DabovBM3D,DabovBM3DSAPCA,PortillaBLSGSM}, work with just two inputs; the noisy signal and an estimate of the standard deviation of the Gaussian noise.  The packages then tune all other parameters internally so as to minimize the MSE. Thus, for the BM3D, BM3D-SAPCA, and BLS-GSM variants of D-AMP we use $(\hat{\sigma}^t)^2=||z^t||_2^2/m$ along with the packages and skip the parameter tuning problem. 
	
For denoisers without self-tuning packages, such as NLM, the tuning problem is challenging because at early iterations the effective noise has a large standard deviation but at later iterations the effective noise has a small standard deviation. This means the best parameters for early iterations are very different than the best parameters for later iterations. To get around this problem we use look-up-tables to set the parameters according to $\hat{\sigma}^t$. We naively generated these tables by first constructing artificial denoising problems with varying amounts of additive white Gaussian noise and then sweeping through the tuning parameters at each noise level. Figure \ref{fig:PSNRvsh} presents how we chose the parameter $h$ used in NLM. At each noise level we simply chose the parameter values that maximized the PSNR of the denoising problem.  For example, for NLM our look-up-table set $h$ to .9 for $\hat{\sigma}^t$ between 15 and 30. The same parameters were applied to all images; we did not optimize our code for individual images.


Recall that the state evolution comparison (Figure \ref{fig:StateEvolution}) showed that the MSE of BM3D-IT rose as the number of iterations increased. We attribute this to non-Gaussian effective noise and correct for this behavior by over-smoothing BM3D-IT at each iteration. The over-smoothing was set by using parameters optimized for $2\hat{\sigma}$ rather than $\hat{\sigma}$.  The scalar 2 was chosen as it provided the best MSE among the scalar values we tested.\\

	\begin{figure}[t]
	\begin{center}
	  \includegraphics[width=.5\textwidth]{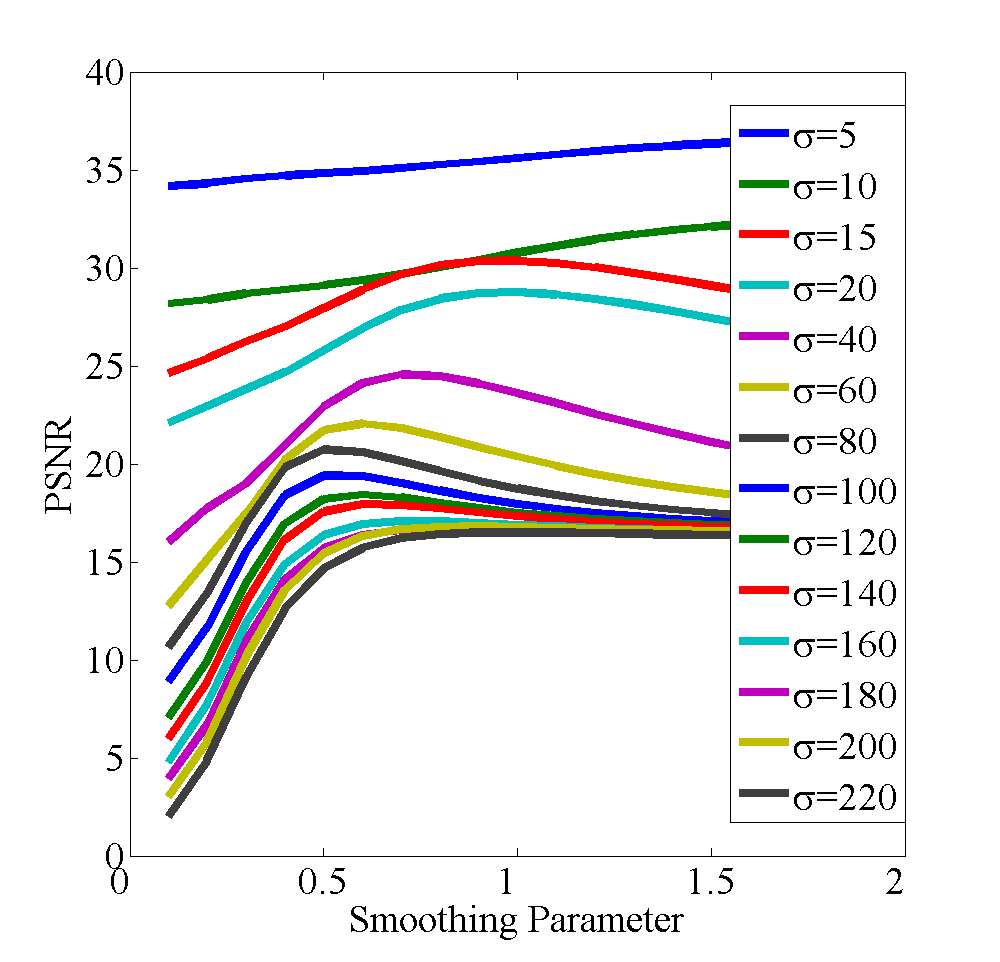} 
	  \caption{PSNR of a NLM denoised image as a function of the smoothing parameter, $h$ in \eqref{eqn:NLMFilterWeigths} divided by the standard deviation of the noise, $\sigma$.  The noisy images had been contaminated with AWGN with various standard deviations.  Notice that different noise levels required different smoothing parameters. We used this data to create a look-up table used for parameter control within the NLM-AMP algorithm.}
	   \label{fig:PSNRvsh} 
	 \end{center}
	\end{figure}

\subsubsection{Stopping criterion}
	AMP is typically designed to stop after some number of iterations or when $\frac{\|x^t- x^{t-1}\|_2}{\|x^t\|_2}$ is less than a threshold. Figure \ref{fig:IterVPSNR} demonstrates the PSNR evolution of BM3D-AMP (as a function of iterations) for different sampling rates of the $128\times128$ Barbara test image. As the figure suggests, after about 10 iterations the PSNR has generally approached its maximum, but the variance of the estimates remains very high. After 30 iterations the variance is quite low. Therefore to reduce variation in our results, we decided to run BM3D-AMP for $30$ iterations. The other D-AMP algorithms, as well as D-IT, IST, and AMP, exhibited similar behavior and were also run for 30 iterations.\\
		
\begin{figure}[t]
\begin{center}
  \includegraphics[width=.5\textwidth]{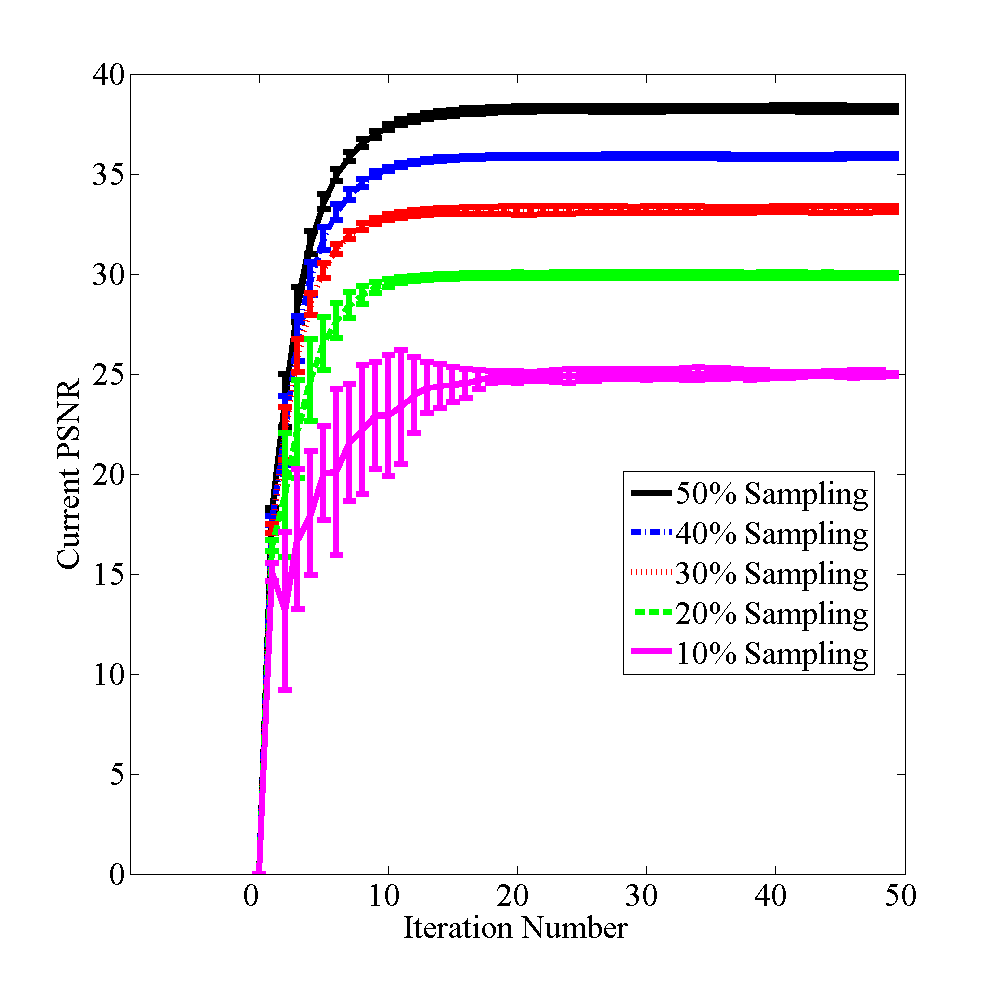} 
  \caption{The progression of the intermediate estimates' PSNRs on a $128 \times 128$ Barbara test image over several iterations at different sampling rates. Notice that the estimates have high variance at first but generally stabilize by iteration 30.}
   \label{fig:IterVPSNR} 
 \end{center}
\end{figure}

\subsubsection{Onsager correction}
In all our implementations of D-AMP (except for the original AMP for which we used the closed form solution) we have used the Monte Carlo method for calculating the Onsager correction term, as reviewed in Section \ref{ssec:onsagerapprox}. While the algorithm seems to be insensitive to the exact value of $\epsilon$ and works for a wide range of values of $\epsilon$, we used $\epsilon=\frac{\|x\|_\infty}{1000}$.  We found this value was small enough for the approximation to be effective while not so small as to result in rounding errors. In the case of the original AMP, we have used the calculations we described in Section \ref{ssec:onsagerexact}.

\subsection{State evolution of D-AMP}\label{sec:seandgaussianitycheck}
\begin{figure*}[t]
\begin{center}
  \includegraphics[width=\textwidth]{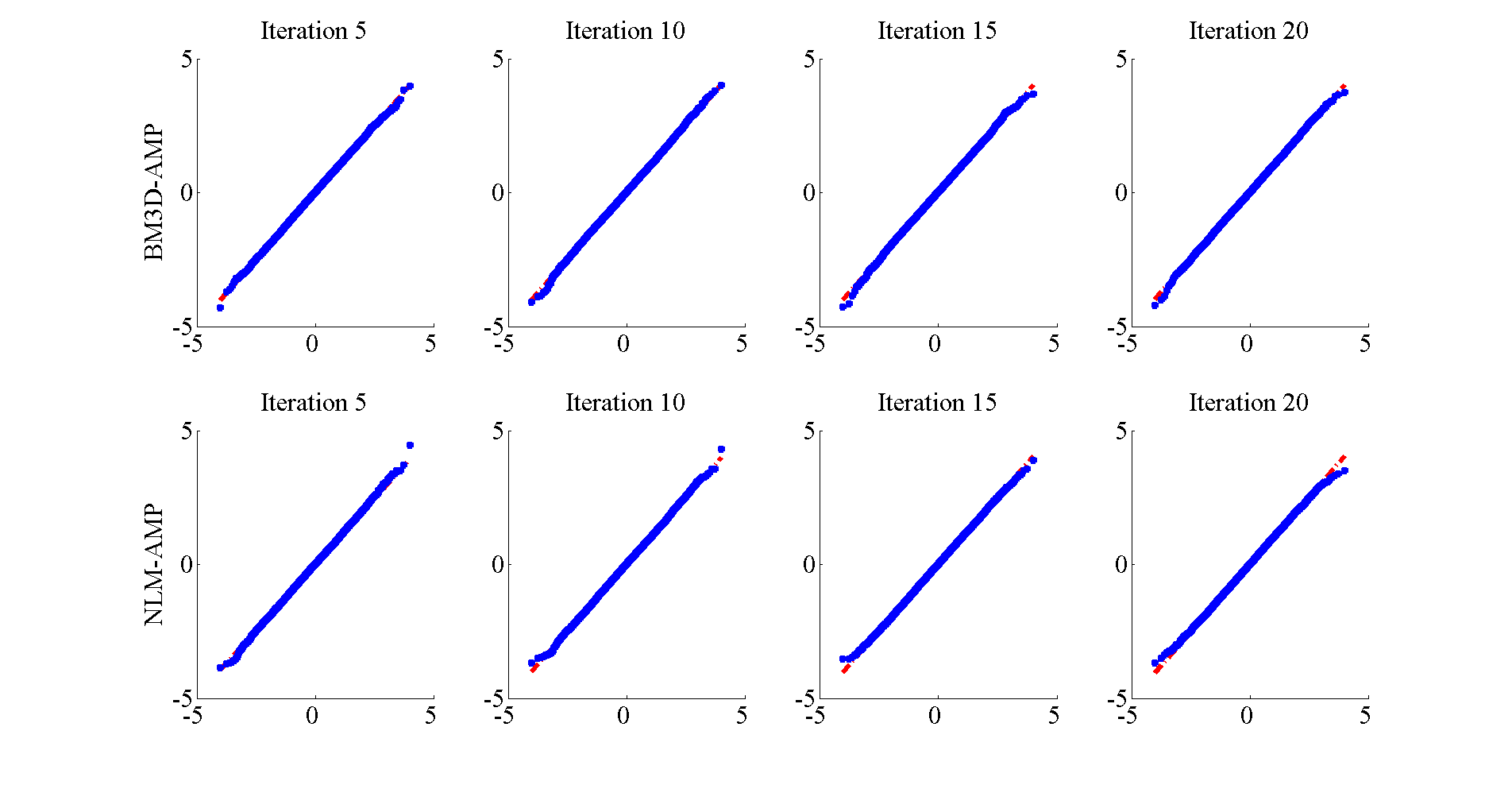} 
  \caption{QQplots of the effective noise at various iterations of BM3D-AMP and NLM-AMP.  Notice that the effective noise remains Gaussian.}
   \label{fig:MultipleQQplots} 
 \end{center}
\end{figure*}

\begin{figure*}[t]
	\begin{center}
		\includegraphics[width=\textwidth]{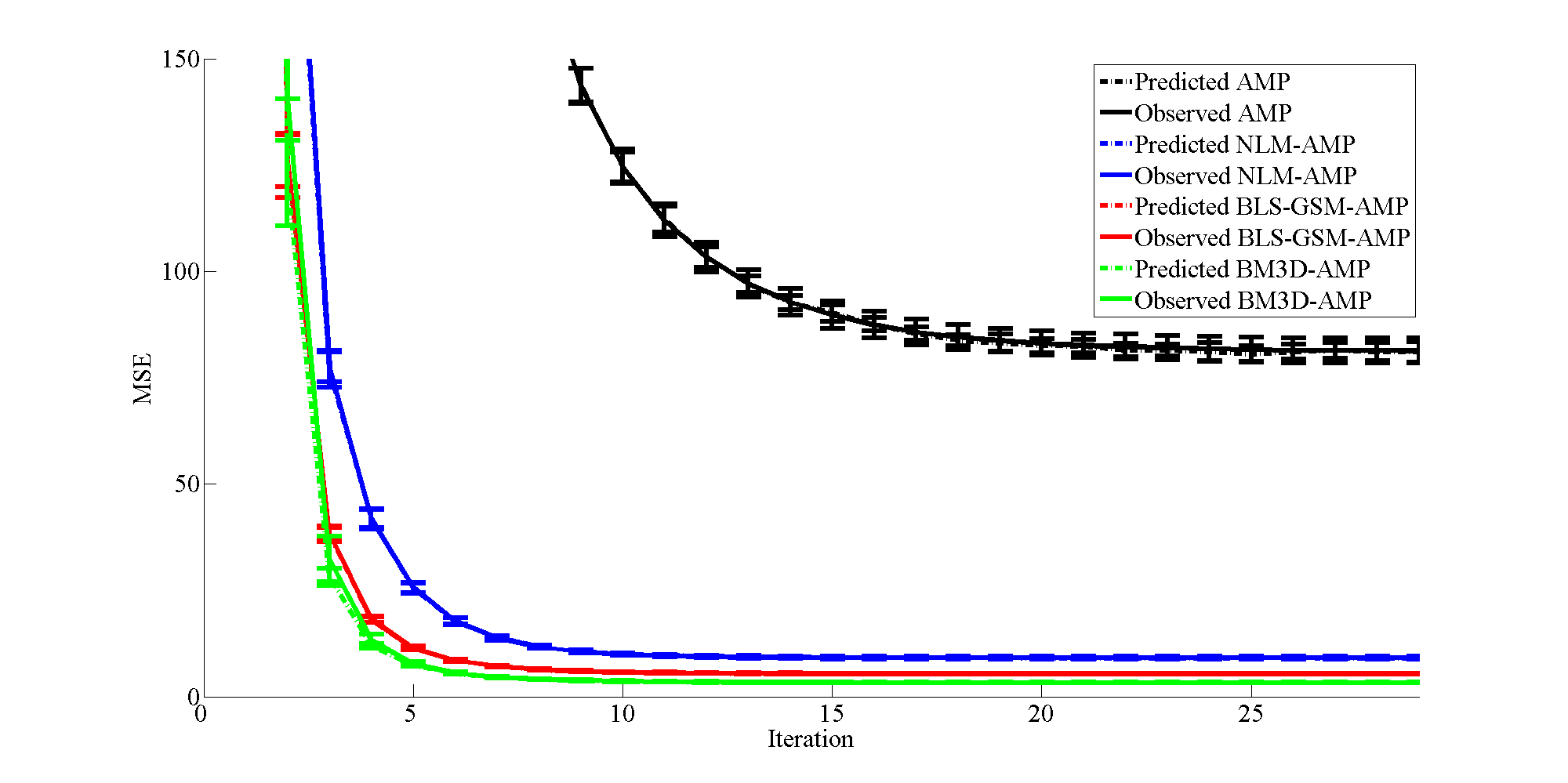} 
		\caption{State evolutions of multiple D-AMP algorithms when applied to a 40\% sampled $128 \times 128$ House test image with no measurement noise.  There is near perfect correspondence between the predicted and true MSE.}
		\label{fig:MultipleStateEvolutions} 
	\end{center}
\end{figure*}

Because the effective noise within D-AMP iterations is Gaussian, as further illustrated in Figure \ref{fig:MultipleQQplots}, state evolution serves as an effective predictor of D-AMP's performance. As the first step in our simulations, we would like to provide evidence of this prediction accuracy. To do so we compare the predicted and observed performance of D-AMP with NLM, wavelet thresholding, BLS-GSM, and BM3D. 

Recall that the state evolution of D-AMP is defined by
\begin{equation} 
\theta^{t+1}(x_o, \delta, \sigma_w^2) = \frac{1}{n}  \mathbb{E} \| D_{\sigma^t} (x_o + \sigma^t \epsilon) -x_o \|_2^2, \nonumber
\end{equation}
where $(\sigma^t)^2 = \frac{\theta^t}{\delta}(x_o, \delta, \sigma_w^2) + \sigma_w^2$. To compute this value in practice, at every iteration we added white Gaussian noise with standard deviation $\sigma^t$ to $x_o$, denoised the signal with denoiser $D_{\sigma^t}$ (using the true, rather than estimated, $\sigma^t$), and then computed the MSE.

Figure \ref{fig:MultipleStateEvolutions} displays the state evolutions alongside the true intermediate MSEs of the four aforementioned denoising-based algorithms when applied to a $\delta=0.4$ sampled $128 \times 128$ House test image with no measurement noise. The average true MSEs at iteration 29 of AMP, NLM-AMP, BLS-GSM-AMP, and BM3D-AMP are all within 1.2\% of the MSEs predicted by their respective state evolutions. We have posted our code online\footnote{\href{http://dsp.rice.edu/software/DAMP-toolbox}{http://dsp.rice.edu/software/DAMP-toolbox}} to enable other researchers to verify and explore the validity of our state evolution predictions for these and other D-AMP algorithms.

\subsection{One-dimensional synthetic test}\label{sec:one dim}
As a simple demonstration of the improvements that can be achieved by employing better denoising algorithms in AMP, we compare the performance of the original AMP (that employs sparsity in the wavelet domain) with the performance of NLM-AMP on a piecewise constant signal. Within the test AMP used a Haar basis for wavelet thresholding and used the max-min optimal threshold as determined by \cite{MalekiAmpAnalysis}. The Haar basis was chosen because it well captures signal discontinuities. NLM-AMP used a length 11 patch, $|N(i)|=11$, a length 21 search window, $|\Omega_i|$=21, and a smoothing parameter of 1.5, $h=1.5$.  These settings were chosen because they allow NLM to effectively denoise piecewise constant signals at a variety of noise levels. The results of our simulation are shown in Figure \ref{fig:1D_Constant}. As is clear from the figure, NLM-AMP significantly outperforms the original AMP. Even though the signal is relatively sparse in the wavelet domain, NLM captures its structure far more effectively. Hence NLM-AMP outperforms the standard AMP that employs sparsity in the wavelet domain.

%
%

\subsection{Imaging tests}
\subsubsection{State-of-the-art recovery algorithms}
In this section we compare the performance of D-AMP, using a variety of denoisers, with other CS reconstruction algorithms. In particular, we compare the performance of our D-AMP algorithm with turbo-AMP \cite{som2012compressive}\footnote{\href{http://www2.ece.ohio-state.edu/~schniter/turboAMPimaging/}{http://www2.ece.ohio-state.edu/\texttildelow{}schniter/turboAMPimaging/}}, which is a hidden Markov tree model-based AMP algorithm, and ALSB \cite{JianZhangALSB}\footnote{\href{http://idm.pku.edu.cn/staff/zhangjian/ALSB/}{http://idm.pku.edu.cn/staff/zhangjian/ALSB/}} and NLR-CS \cite{dongcompressive}\footnote{\href{http://see.xidian.edu.cn/faculty/wsdong/NLR_Exps.htm}{http://see.xidian.edu.cn/faculty/wsdong/NLR\_Exps.htm}}, which both utilize non-local group-sparsity. 
NLR-CS represents the current state-of-the-art in CS image reconstruction algorithms.
We compare these 3 algorithms to D-AMP based on the NLM, BLS-GSM, BM3D, and BM3D-SAPCA denoisers.  The performance of D-AMP using the Gaussian filter and the bilateral filter was not competitive and has been omitted from the results. We include comparisons with AMP, using a wavelet basis. We also include comparisons with the BM3D-IT algorithm to illustrate the importance of the Onsager correction term in the performance of D-AMP. Other D-IT algorithms demonstrated considerably worse performance and are therefore omitted from the results.

\subsubsection{Test Settings}
ALSB uses rows drawn from a $32^2 \times 32^2$ orthonormalized Gaussian measurement matrix to perform block-based compressed sensing, as described in \cite{gan2007block}. All other tests used an $m \times n$ measurement matrix that was generated by first using Matlab's \texttt{randn(m,n)} command and then normalizing the columns.  All simulations were conducted on a 3.16 GHz Xeon 
quad-core processor with 32GB of memory. 

For the AMP algorithm we used Daubechies 4 wavelets as the sparsifying basis and set its threshold optimally according to \cite{MalekiAmpAnalysis}. The parameters of D-AMP and D-IT were set following the methods described in section \ref{ssec:tunepractice}. All D-IT and D-AMP algorithms were run for 30 iterations. AMP was run for 30 iterations as well. Turbo-AMP was run for 10 iterations.  We experimented with running turbo-AMP for 30 iterations but found that this yielded no improvement in performance while nearly tripling the computation time. Because the DCT-sparsity-based iterative soft-thresholding method used to generate an initial estimate in NLR-CS's provided source code failed for Gaussian measurement matrices, we generated the initial estimates used by NLR-CS by running BM3D-AMP for 8 iterations for noiseless tests and 4 iterations for noisy tests. Only 4 iterations of BM3D-AMP were used during noisy tests because if run for 8 iterations the initial estimates from BM3D-AMP were often better than the final estimates from NLR-CS. Turbo-AMP, ALSB, and NLR-CS were otherwise tested under their default settings.\\

\subsubsection{Image database}
The data was generated using six standard image processing images drawn from Javier Portilla's dataset:\footnote{\href{http://www.io.csic.es/PagsPers/JPortilla/software}{http://www.io.csic.es/PagsPers/JPortilla/software}}  Lena, Barbara, Boat, Fingerprint, House, and Peppers. The images each have a pixel range of roughly $0-255$. Each of these images, except the examples presented in Figures \ref{fig:VisualComparo} and \ref{fig:VisualComparoNoisy}, were rescaled to $128 \times 128$ for testing.  Restricting the tests to $128 \times 128$ enabled the entire measurement matrix $\mathbf{A}$ to be stored in memory.  We also created a version of D-AMP that does not store $\mathbf{A}$ but instead generates sections of $\mathbf{A}$ as required.  This version can handle images of arbitrarily large size but is extremely slow.\\

\subsubsection{Noiseless image recovery}

While matching the denoiser to the signal produces impressive results in one-dimensional settings (as summarized in Section \ref{sec:one dim}), the results in 2D are even more pronounced.  We begin this section with a visual comparison of three algorithms: Figure \ref{fig:VisualComparo} illustrates the image recovery performance of AMP, NLR-CS, and our BM3D-SAPCA-AMP algorithm. BM3D-SAPCA-AMP outperformed NLR-CS slightly; 29.96 dB vs 29.31 dB.  Both of these algorithms dramatically outperformed the wavelet sparsity-based AMP algorithm; 20.07 dB. 

We also present a more complete comparison of D-AMP with other algorithms in Table \ref{tab:Comparison}.\footnote{Model-CoSaMP and other model-based techniques have not been included in our simulation results.  First and foremost these methods were too slow for us to gather data before finishing the report.  Additionally, we found they were not competitive: In the original Model-CoSaMP paper \cite{RichModelBasedCS} the authors reported a RMSE of 11.1 (PSNR of 27.22 dB) from a reconstruction of a $128 \times 128$ pepper test image using 5000 Gaussian measurements.  By comparison, BM3D-AMP returns a RMSE of 5.1 (PSNR of 33.98 dB) on the same test.} As is clear from this table, BM3D-AMP or BM3D-SAPCA-AMP outperform all the other algorithms in a large majority of the tests. In the next section we demonstrate that the denoising based-AMP algorithms perform far better than competing methods when in the presence of measurement noise.\\

\begin{figure*}[bhpt]
\centering
\subfigure[Original]{\includegraphics[width=.35\textwidth]{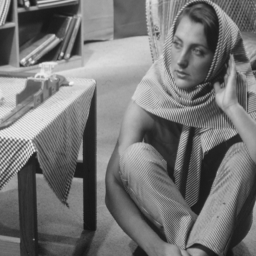}} 
\subfigure[AMP recovery]{\includegraphics[width=.35\textwidth]{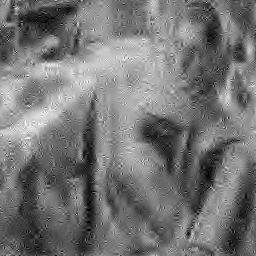}} 

\subfigure[NLR-CS recovery]{\includegraphics[width=.35\textwidth]{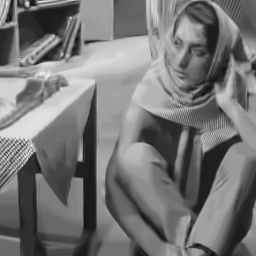}} 
\subfigure[BM3D-SAPCA-AMP recovery]{\includegraphics[width=.35\textwidth]{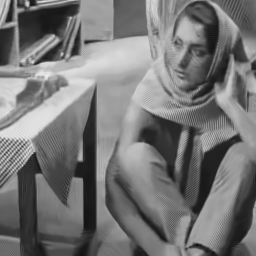}} 
\caption{Reconstructions of 10\% sampled $256 \times 256$ Barbara test image. The performance of BM3D-SAPCA-AMP is slighlty better than the state-of-the-art NLR-CS algorithm and dramatically better than AMP.}
\label{fig:VisualComparo}
\end{figure*}
\begin{table*}[t]
  \centering
  \caption{PSNR of $128 \times 128$ reconstructions with no measurement noise.}
    \begin{tabular}{r|rrrrrr}
	\toprule
    \textbf{10\% Sampling} & \textbf{Lena} & \multicolumn{1}{c}{\textbf{Barbara}} & \multicolumn{1}{c}{\textbf{Boat}} & \multicolumn{1}{c}{\textbf{Fingerprint}} & \multicolumn{1}{c}{\textbf{House}} & \multicolumn{1}{c}{\textbf{Peppers}} \\
    \midrule
    AMP   & 18.47 & 17.67 & 18.96 & 15.87 & 19.98 & 17.50 \\
    Turbo-AMP & 18.35 & 17.46 & 18.62 & 16.30 & 21.77 & 17.01 \\
    ALSB  & 25.30 & 24.01 & 22.44 & 16.25 & 31.09 & 24.01 \\
    NLR-CS & \textbf{26.74} & \textbf{24.95} & 23.97 & 18.11 & \textbf{34.46} & \textbf{25.21} \\
    BM3D-IT & 5.68  & 5.97  & 5.43  & 4.70  & 4.93  & 5.72 \\
    NLM-AMP & 21.81 & 20.17 & 21.43 & 17.69 & 24.81 & 20.42 \\
    BLS-GSM-AMP & 24.92 & 23.35 & 23.98 & 17.53 & 30.52 & 24.09 \\
    BM3D-AMP & 26.01 & 24.24 & \textbf{24.07} & \textbf{18.24} & 34.12 & 24.41 \\
    BM3D-SAPCA-AMP & 15.04 & 24.28 & 22.62 & 18.17 & 32.74 & 23.99 \\
    \toprule
    \textbf{20\% Sampling} & \textbf{Lena} & \multicolumn{1}{c}{\textbf{Barbara}} & \multicolumn{1}{c}{\textbf{Boat}} & \multicolumn{1}{c}{\textbf{Fingerprint}} & \multicolumn{1}{c}{\textbf{House}} & \multicolumn{1}{c}{\textbf{Peppers}} \\
    \midrule
    AMP   & 21.26 & 20.08 & 21.62 & 16.86 & 22.97 & 20.27 \\
    Turbo-AMP & 23.48 & 21.45 & 23.36 & 16.31 & 28.20 & 21.78 \\
    ALSB  & 28.66 & 27.98 & 26.09 & 17.42 & 36.28 & 28.12 \\
    NLR-CS & 31.88 & 30.31 & 26.96 & 21.10 & 38.70 & 30.42 \\
    BM3D-IT & 25.64 & 24.38 & 23.79 & 6.63  & 32.92 & 23.87 \\
    NLM-AMP & 27.73 & 24.27 & 23.97 & 19.72 & 31.75 & 23.70 \\
    BLS-GSM-AMP & 29.77 & 28.00 & 27.06 & 18.45 & 35.76 & 29.14 \\
    BM3D-AMP & 31.12 & 29.83 & \textbf{27.58} & 21.14 & 38.30 & 30.00 \\
    BM3D-SAPCA-AMP & \textbf{32.15} & \textbf{30.41} & 27.35 & \textbf{22.02} & \textbf{38.94} & \textbf{31.09} \\
    \toprule
    \textbf{30\% Sampling} & \textbf{Lena} & \multicolumn{1}{c}{\textbf{Barbara}} & \multicolumn{1}{c}{\textbf{Boat}} & \multicolumn{1}{c}{\textbf{Fingerprint}} & \multicolumn{1}{c}{\textbf{House}} & \multicolumn{1}{c}{\textbf{Peppers}} \\
    \midrule
    AMP   & 23.90 & 22.70 & 23.67 & 17.57 & 26.15 & 23.12 \\
    Turbo-AMP & 25.88 & 24.35 & 24.80 & 16.33 & 32.18 & 24.58 \\
    ALSB  & 31.91 & 30.69 & 28.69 & 22.76 & 38.51 & 31.85 \\
    NLR-CS & 35.86 & 33.78 & 30.27 & 23.01 & 41.15 & 34.80 \\
    BM3D-IT & 28.16 & 27.21 & 24.43 & 18.44 & 35.48 & 25.49 \\
    NLM-AMP & 29.94 & 29.39 & 27.67 & 20.81 & 36.49 & 29.92 \\
    BLS-GSM-AMP & 33.29 & 31.06 & 29.95 & 19.20 & 39.17 & 32.73 \\
    BM3D-AMP & 34.87 & 33.14 & 30.60 & 22.95 & 40.92 & 33.83 \\
    BM3D-SAPCA-AMP & \textbf{36.21} & \textbf{34.18} & \textbf{31.22} & \textbf{23.71} & \textbf{41.55} & \textbf{34.91} \\
    \toprule
    \textbf{40\% Sampling} & \textbf{Lena} & \multicolumn{1}{c}{\textbf{Barbara}} & \multicolumn{1}{c}{\textbf{Boat}} & \multicolumn{1}{c}{\textbf{Fingerprint}} & \multicolumn{1}{c}{\textbf{House}} & \multicolumn{1}{c}{\textbf{Peppers}} \\
    \midrule
    AMP   & 26.35 & 24.77 & 25.41 & 18.65 & 29.20 & 25.36 \\
    Turbo-AMP & 27.91 & 26.11 & 26.98 & 16.65 & 35.37 & 26.83 \\
    ALSB  & 34.17 & 34.19 & 30.92 & 24.14 & 41.13 & 35.15 \\
    NLR-CS & 39.07 & 36.99 & 32.75 & 24.78 & 43.45 & 37.63 \\
    BM3D-IT & 29.50 & 28.22 & 25.13 & 19.47 & 36.94 & 28.86 \\
    NLM-AMP & 32.58 & 32.15 & 28.94 & 21.53 & 38.62 & 31.47 \\
    BLS-GSM-AMP & 36.50 & 34.33 & 32.41 & 20.32 & 40.84 & 35.86 \\
    BM3D-AMP & 38.05 & 35.94 & 32.77 & 24.59 & 42.97 & 36.77 \\
    BM3D-SAPCA-AMP & \textbf{39.33} & \textbf{37.05} & \textbf{33.56} & \textbf{25.01} & \textbf{43.86} & \textbf{38.06} \\
    \toprule
    \textbf{50\% Sampling} & \textbf{Lena} & \multicolumn{1}{c}{\textbf{Barbara}} & \multicolumn{1}{c}{\textbf{Boat}} & \multicolumn{1}{c}{\textbf{Fingerprint}} & \multicolumn{1}{c}{\textbf{House}} & \multicolumn{1}{c}{\textbf{Peppers}} \\
    \midrule
    AMP   & 28.12 & 27.19 & 27.44 & 19.84 & 31.86 & 27.99 \\
    Turbo-AMP & 30.64 & 27.69 & 28.80 & 19.24 & 37.54 & 29.17 \\
    ALSB  & 36.95 & 37.10 & 32.96 & 25.80 & 42.76 & 38.11 \\
    NLR-CS & 42.05 & \textbf{39.86} & 35.31 & 26.26 & 45.65 & 40.51 \\
    BM3D-IT & 30.95 & 29.18 & 27.14 & 20.24 & 38.19 & 29.56 \\
    NLM-AMP & 35.09 & 34.72 & 31.45 & 25.34 & 39.71 & 34.10 \\
    BLS-GSM-AMP & 38.92 & 36.42 & 34.72 & 21.61 & 42.34 & 38.72 \\
    BM3D-AMP & 40.89 & 38.21 & 35.07 & 25.99 & 44.91 & 39.38 \\
    BM3D-SAPCA-AMP & \textbf{42.12} & 39.49 & \textbf{36.05} & \textbf{26.76} & \textbf{45.70} & \textbf{40.61} \\
    \toprule
    \end{tabular}%

  \label{tab:Comparison}%
\end{table*}%

\subsubsection{Imaging in the presence of measurement noise}
In realistic settings compressive samples are subject to measurement noise. Noisy sampling can be modeled by $y=\mathbf{A}x+w$ where $w$ represents additive white Gaussian noise (AWGN). In Figure \ref{fig:VisualComparoNoisy} we provide a visual comparison between the reconstructions of BM3D-SAPCA-AMP (26.86 dB) and NLR-CS (25.30 dB) in the presence of measurement noise. In Table \ref{tab:NoisyPerformance} we compare the performance of the BM3D variant of D-AMP to NLR-CS and ALSB when varying amounts of measurement noise are present. As one might expect from a denoising-based algorithm, D-AMP was found to be exceptionally robust to noise. It outperformed the other methods in almost all tests and in some tests by as much as 7.4 dB.\\

\begin{figure*}[t]
\centering
\subfigure[NLR-CS recovery]{\includegraphics[width=.35\textwidth]{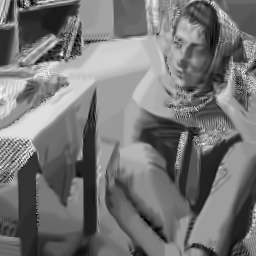}} 
\subfigure[BM3D-SAPCA-AMP recovery]{\includegraphics[width=.35\textwidth]{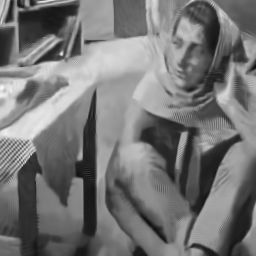}} 
\caption{Reconstructions of 10\% sampled $256 \times 256$ Barbara test image with additive white Gaussian measurement noise with standard deviation 30. Note that BM3D-SAPCA-AMP exhibits far fewer artifacts than NLR-CS.}
\label{fig:VisualComparoNoisy}
\end{figure*}

\begin{table*}[t]
  \centering
  \caption{PSNR of reconstruction of $128 \times 128$ Barbara test image with additive white Gaussian measurement noise with various standard deviations  (s.d.).}
    \begin{tabular}{r|rrrrr}
    \toprule
    \multicolumn{6}{c}{\textbf{AWGN with s.d. 10}} \\
    \midrule
    \textbf{Sampling rate (\%)} & \textbf{10} & \textbf{20} & \textbf{30} & \textbf{40} & \textbf{50} \\
    \midrule
    ALSB  & 21.82 & 24.20 & 25.44 & 26.52 & 27.30 \\
    NLR-CS & 24.29 & 27.84 & 28.85 & 29.24 & 28.48 \\
    BM3D-AMP & 24.25 & 28.44 & 29.88 & 31.06 & 31.34 \\
    \toprule
    \multicolumn{6}{c}{\textbf{AWGN with s.d. 20}} \\
    \midrule
    \textbf{Sampling Rate (\%)} & \textbf{10} & \textbf{20} & \textbf{30} & \textbf{40} & \textbf{50} \\
    \midrule
    ALSB  & 19.32 & 20.83 & 21.40 & 22.15 & 22.72 \\
    NLR-CS & 22.30 & 25.43 & 25.74 & 25.43 & 23.84 \\
    BM3D-AMP & 23.79 & 26.65 & 27.54 & 28.18 & 28.24 \\
    \toprule
    \multicolumn{6}{c}{\textbf{AWGN with s.d. 30}} \\
    \midrule
    \textbf{Sampling Rate (\%)} & \textbf{10} & \textbf{20} & \textbf{30} & \textbf{40} & \textbf{50} \\
    \midrule
    ALSB  & 15.89 & 16.92 & 17.69 & 18.16 & 17.96 \\
    NLR-CS & 21.90 & 22.49 & 22.05 & 20.46 & 18.38 \\
    BM3D-AMP & 22.61 & 24.21 & 24.38 & 24.75 & 24.89 \\
    \toprule
    \multicolumn{6}{c}{\textbf{AWGN with s.d. 40}} \\
    \midrule
    \textbf{Sampling Rate (\%)} & \textbf{10} & \textbf{20} & \textbf{30} & \textbf{40} & \textbf{50} \\
    \midrule
    ALSB  & 15.89 & 16.92 & 17.69 & 18.16 & 17.96 \\
    NLR-CS & 21.92 & 22.48 & 21.99 & 20.44 & 18.38 \\
    BM3D-AMP & 22.65 & 24.22 & 24.61 & 24.88 & 25.06 \\
    \toprule
    \multicolumn{6}{c}{\textbf{AWGN with s.d. 50}} \\
    \midrule
    \textbf{Sampling Rate (\%)} & \textbf{10} & \textbf{20} & \textbf{30} & \textbf{40} & \textbf{50} \\
    \midrule
    ALSB  & 14.54 & 15.72 & 16.42 & 16.62 & 16.20 \\
    NLR-CS & 21.02 & 21.49 & 20.66 & 18.77 & 16.56 \\
    BM3D-AMP & 22.04 & 23.36 & 23.47 & 23.82 & 23.95 \\
    \toprule
    \end{tabular}%
  \label{tab:NoisyPerformance}%
\end{table*}%

\subsubsection{Computational complexity}

\begin{table*}[t]
  \centering
  \caption{Average computation times, in minutes, of $128 \times 128$ reconstructions at various sampling rates.  
  }
    \begin{tabular}{r|rrrrr}
    \toprule
    \textbf{Sampling Rate (\%)} & \textbf{10} & \textbf{20} & \textbf{30} & \textbf{40} & \textbf{50} \\
    \midrule
    AMP   & 0.3   & 0.6   & 1.0   & 1.3   & 1.6 \\
    Turbo-AMP & 1.4   & 2.4   & 3.4   & 4.5   & 5.5 \\
    ALSB  & 52.4  & 60.1  & 66.2  & 70.9  & 71.7 \\
    NLR-CS & 31.6  & 60.6  & 88.1  & 122.8 & 152.2 \\
    BM3D-IT & 0.8   & 1.2   & 1.4   & 1.7   & 2.0 \\
    NLM-AMP & 11.3  & 6.7   & 4.4   & 4.2   & 3.8 \\
    BLS-GSM-AMP & 5.1   & 5.0   & 5.3   & 5.6   & 5.9 \\
    BM3D-AMP & 1.0   & 1.3   & 1.5   & 1.7   & 2.2 \\
    BM3D-SAPCA-AMP & 318.3 & 328.7 & 345.1 & 362.1 & 378.0 \\
    \toprule
    \end{tabular}%
  \label{tab:ComputationTime}%
\end{table*}%

Table \ref{tab:ComputationTime} demonstrates that, depending on the denoiser in use, D-AMP can be quite efficient: The BM3D variant of D-AMP is dramatically faster than NLR-CS and ALSB.  The table also illustrates how using different denoisers within D-AMP presents not only a means of capturing different signal models, but also a way to balance performance and run times.

\section{Conclusions}\label{sec:Conclusion}
Through extensive testing we have demonstrated that the approximate message passing (AMP) compressed sensing recovery algorithm can be extended to use arbitrary denoisers to great effect.  Variations of this denoising-based AMP algorithm (D-AMP) deliver state-of-the-art compressively sampled image recovery performance while maintaining a low computational footprint. 
Our theoretical results and simulations show that the performance of D-AMP can be predicted accurately by state evolution. We have also proven that the problem of tuning the parameters of D-AMP is no more difficult than the tuning of the denoiser that is used in the algorithm. Finally, we have shown that D-AMP is extremely robust to measurement noise.  D-AMP represents a plug and play method to recover compressively sampled signals of arbitrary class; simply choose a denoiser well matched to the signal model and plug it in the AMP framework. Since designing denoising algorithms that employ complicated structures is usually much easier than designing recovery algorithms, D-AMP can benefit many different application areas.  

A significant amount of work remains to be done.  First and foremost, all of the theory we developed for D-AMP relies upon the assumption that residual signals follow Gaussian distributions.  In this paper we supported this assumption with state evolution and QQplot experiments. Theoretical validation of this assumption is left for future research.  Likewise, all theory and results have been for i.i.d. Gaussian (or subGaussian) measurement matrices.  Extension to other measurement matrices such as Fourier samples is another open direction that is left for future research.  

\appendix
\subsection{Proof of Proposition \ref{prop:minimaxconnec}}\label{app:proofprop1}
Let $D^M_{\sigma}$ and $D^*_{\sigma}$  denote the minimax denoiser, and minimax optimal family of denoisers for D-AMP, respectively. For notational simplicity we assume that $\sup_{x_o} \mathbb{E} \|D_{\sigma}^M(x_o+ \sigma \epsilon)-x_o  \|_2^2$ is achieved at certain point $x_o^{m,\sigma}$, and that $\sup_{x_o} \mathbb{E} \|D_{\sigma}^*(x_o+ \sigma \epsilon)-x_o  \|_2^2$ is achieved at certain point $x_o^{*,\sigma}$. Note that according to the state evolution for every $x_o$ the fixed point of state evolution is given by
\begin{eqnarray}
\theta^{\infty}_{D^M} (x_o, \delta, \sigma_w^2)= \frac{1}{n} \mathbb{E} \|D_{\sigma}^M(x_o+ \sigma \epsilon)-x_o  \|_2^2,\nonumber
\end{eqnarray}
where $\sigma^2= \frac{\theta^{\infty}_{D^M} (x_o, \delta, \sigma_w^2)}{\delta}+ \sigma_w^2$. Define
\begin{eqnarray}
\tilde{\theta}^{\infty}_{D^M} (\delta, \sigma_w^2) = \sup_{x_o \in \{x_o^{m, \sigma} \: \ \sigma>0 \}} \theta^{\infty}_{D^M} (x_o, \delta, \sigma_w^2).\nonumber
\end{eqnarray}
Again for notational simplicity assume that the supremum is achieved at $x_o^{m*}$. The following lemma will be useful in our proof. It also has a nice interpretation that we describe after proving it. 

\begin{lemma}\label{lem:leastfavorablesig}
If $\theta^{\infty}_D(x_o, \delta, \sigma_W^2)$ denotes the fixed point of the state evolution with denoiser $D$ at signal $x_o$, then
\[
\theta^{\infty}_{D^M} (x_o, \delta, \sigma_w^2) \leq \tilde{\theta}^{\infty}_{D^M} (\delta, \sigma_w^2) = {\theta}^{\infty}_{D^M} (x^{m*}, \delta, \sigma_w^2). 
\]
\end{lemma}
\begin{proof}
We first claim that for every $\theta> \tilde{\theta}^{\infty}_{D^M} (\delta, \sigma_w^2)$ and for every $x_o$ we have
\begin{equation}\label{eq:proof1prop2}
\theta> \frac{1}{n} \mathbb{E} \| D_\sigma^M(x_o+\sigma \epsilon) -x_o \|_2^2,
\end{equation}
where $\sigma^2 = \frac{\theta}{\delta}+ \sigma^2_W$. Suppose that this is not true, i.e., there exists $x_o$ and $\theta> \tilde{\theta}^{\infty}_{D^M} (\delta, \sigma_w^2)$ such that
\begin{equation*}
\theta \leq  \frac{1}{n} \mathbb{E} \| D_\sigma^M(x_o+\sigma \epsilon) -x_o \|_2^2.
\end{equation*}
Then
\begin{eqnarray}
\theta &\leq& \frac{1}{n} \mathbb{E} \| D_\sigma^M(x_o+\sigma \epsilon) -x_o \|_2^2 \nonumber \\
 &\leq& \frac{1}{n} \mathbb{E} \| D_\sigma^M(x_o^{\sigma, M}+\sigma \epsilon) -x_o^{\sigma, M} \|_2^2,
\end{eqnarray}
where the last inequality is due to the definition of $x_o^{\sigma, M}$. This implies that the fixed point of $D_{\sigma}^M$ for vector $x_o^{\sigma, M}$ will be larger than $\theta$ and hence will be larger than $ \tilde{\theta}^{\infty}_{D^M} (\delta, \sigma_w^2)$. This is in contradiction with the definition of $ \tilde{\theta}^{\infty}_{D^M} (\delta, \sigma_w^2)$. Therefore, for any $x_o$ \eqref{eq:proof1prop2} holds. Furthermore, \eqref{eq:proof1prop2} implies that for every $x_o$ the fixed point of the state evolution of $D_\sigma^M$ can only happen for  $\theta < \tilde{\theta}^{\infty}_{D^M} (\delta, \sigma_w^2)$. Hence establishes the result.
\end{proof}

This result has an interesting interpretation. The least favorable signal for D-AMP, i.e., the signal that leads to the highest fixed point, is one of the least favorable signals for the denoiser $D_{\sigma}$. While we proved this result for a specific denoiser $D_{\sigma}^M$, the proof can be easily extended to any denoiser $D_{\sigma}$.

We may now return to the proof of Proposition \ref{prop:minimaxconnec}. Similar to Lemma \ref{lem:leastfavorablesig} define
\[
\tilde{\theta}^{\infty}_{D^*} (\delta, \sigma_w^2) = \sup_{x_o \in \{x_o^{*,\sigma} \ : \ \sigma>0 \}} {\theta}^{\infty}_{D^*} (x_o, \delta, \sigma_w^2). 
\]
Also, suppose that the supremum is achieved at $x_o^{**}$. Clearly, for any $\theta> \tilde{\theta}^{\infty}_{D^*} (\delta, \sigma_w^2)$ we have
\begin{eqnarray}
\theta &>& \frac{1}{n} \mathbb{E} \|D_\sigma^*(x_o^{**}+\sigma \epsilon) -x_o^{**} \|_2^2\nonumber\\
& =& \sup_{x_o}  \frac{1}{n} \mathbb{E} \|D^*_\sigma(x_o+\sigma \epsilon) -x_o \|_2^2  \nonumber \\
&\geq& \inf_{D_\sigma} \sup_{x_o}  \frac{1}{n} \mathbb{E} \|D_\sigma(x_o+\sigma \epsilon) -x_o \|_2^2 \nonumber \\
&=& \frac{1}{n} \mathbb{E} \|D^M_\sigma(x_o^{m,*}+\sigma \epsilon) -x_o^{m,*} \|_2^2.
\end{eqnarray}
Hence the fixed point of $D_{\sigma}^M$ is less than or equal to the fixed point of $D^*_{\sigma}$. Hence the proof is complete. \hfill
\subsection{Proof of Proposition \ref{prop:onemeasurementrecovery}}\label{app:proof theorem1}

Let $B_k$ denote the class of $k$-sparse signals with zero-one elements. Suppose that we have observed $y=Ax_o$ ($x_o \in B_k$) and the goal is to recover $x_o$ from $y$.  Consider the following recovery algorithm that is a special form of compressible signal pursuit proposed in \cite{jalali2013compression, jalali2012minimum}:
\[
\hat{x}_o = \arg \min _{x \in B_k } \|y-Ax\|_2^2. 
\]

\begin{lemma} For $m\geq1$, 
\[
\mathbb{E} \|\hat{x}_o-x_o\|_2^2 =0. 
\]
\end{lemma}
\begin{proof}
First note that $\mathbb{P} (\hat{x}_o \neq x_o ) = \mathbb{P} (A(x_o -\hat{x}_o)=0) =0$. We can use the union bound and the fact that there are only ${n \choose k}$ vectors in this space, to show that
\[
\mathbb{P} (\exists x_o \in B_k \ : \ \hat{x}_o \neq x_o) =0. 
\]
When the algorithm incorrectly estimates $x_o$, $\|\hat{x}_o- x_o\|_2^2$ is at most $2k$. Since the error is bounded our result is established.
\end{proof}

This is essentially the proof of the second part of the theorem. We now prove the first part of the theorem. 

Consider the distribution $\underline{\pi}_i^{*}= (1- \frac{k}{n}+\gamma)\delta_0+(\frac{k}{n}- \gamma)\delta_1$, where $\delta_a$ denotes a point mass at $a$. Construct a distribution on $\mathbb{R}^n$ in the following way: 
\[
\underline{\pi}^* = \underline{\pi}_1^{*} \times \underline{\pi}_2^{*}\times \ldots \times \underline{\pi}_n^{*}. 
\]
Here are the main steps of the proof:
\begin{enumerate}[(i)]
\item We first prove that the samples we draw from $\underline{\pi}^* $ belong to $B_k$ with high probability.
\item We employ the result of step one to derive a lower bound for the minimax risk. 
\end{enumerate}
 Step (i) is a simple application of Hoeffding inequality.  Let $\underline{x}_o$ be a sample from this distribution. By using Hoeffding inequality we obtain:
\[
\mathbb{P} \left(\Big |  \frac{1}{n}\|\underline{x}_o \|_0 - \frac{k}{n} - \gamma \Big|< \gamma \right) \leq 2 {\rm e}^{- n \gamma^2/2}.
\]
Therefore, 
\[
\mathbb{P} \left(\frac{1}{n}\|\underline{x}_o \|_0 < \frac{k}{n} \right) \leq 2 {\rm e}^{- n \gamma^2/2}.
\]
In other words, with very high probability the samples that are generated from $\underline{\pi}^*$ belong to $B_k$. Set $\gamma = \frac{\sqrt{2}}{n^{1/4}}$, define the event $\mathcal{A}$ as $\|x_o\|_0 \leq k$  and let $\pi^{**}$ denote the distribution of $x_o$ conditioned on event $\mathcal{A}$. Note that the support of $\pi^{**}$ is a subset of $B_k$. Now we can discuss step (ii), i.e., deriving a lower bound for minimax risk. Since the support of $\pi^{**}$ is a subset of $B_k$, for every denoiser $D_{\sigma}$ we have
\begin{eqnarray}
\lefteqn{\mathbb{E}_{x_o \sim \pi^{**}} \mathbb{E} ( \|D_{\sigma} (x_o+ \sigma z)-x_o \|_2^2 \ | \ x_o)} \nonumber \\  &\leq& \sup_{x_o \in {B}_k} \mathbb{E} ( \|D_{\sigma} (x_o+ \sigma z)-x_o \|_2^2 \ | \ x_o). 
\end{eqnarray}
By taking the infimum over $D_{\sigma}$ from both sides, since the optimal denoiser on the left is the Bayes denoiser, we obtain
\begin{eqnarray}
\lefteqn{\mathbb{E} \|\mathbb{E}_{x_o\sim \pi^{**}} (x_o \ | \ x_o + \sigma z) -x_o\|} \nonumber \\
& \leq& \inf_{D_{\sigma}} \sup_{x_o \in B_k} \mathbb{E} \|D_\sigma(x_o+ \sigma z) - x_o\|_2^2. 
\end{eqnarray}
In other words we have derived a lower bound for the minimax risk based on $\pi^{**}$. Our next step is to calculate the lower bound we have on the left hand side. Note that $ \mathbb{P} (\mathcal{A}^c) = O({\rm e}^{- \sqrt{n}})$, and 
\begin{eqnarray}
\mathbb{E}_{\pi^{**}} (x_o \ | \ y= x_o+\sigma z)=\mathbb{E}_{\pi^*} (x_o \ | \ y= x_o+ \sigma z, \mathcal{A}). 
\end{eqnarray}
Hence we have 
\begin{eqnarray}\label{eq:riskcalc1}
&&\mathbb{E}_{\pi^{**}} (x_o \ | \ y= x_o+\sigma z)\nonumber\\
= \!\!\!\!\!\!\!&&\frac{\mathbb{E}_{\pi^*} (x_o \ | \ x_o+ \sigma z) - \mathbb{E}_{\pi^*} (x_o \ | \ x_o+ \sigma z, \mathcal{A}^c)\mathbb{P}(\mathcal{A}^c)}{\mathbb{P} (A)}. 
\end{eqnarray}
Define $\phi(z_i)= \frac{1}{\sqrt{2\pi}} {\rm e}^{- z_i^2/2}$ and $\underline{\phi}(z) = \phi(z_1) \times \phi(z_2)\times \ldots \times \phi(z_n)$.
\begin{eqnarray}
\lefteqn{\mathbb{E}_{x_o\sim \pi^{**}, z \sim \underline{\phi}} \| x_o - \mathbb{E}_{\pi^{**}} (x_o \ | \ y= x_o+\sigma z)\|_2^2} \nonumber \\
&=& \mathbb{E}_{x_o \sim \pi^{**},  z \sim \underline{\phi}} \left\| x_o - \frac{\mathbb{E}_{\pi^{*}} (x_o \ | \ y= x_o+\sigma z)}{1- \mathbb{P} (A^c)}\right\|_2^2 \nonumber \\
&&+ O(n {\rm e}^{- \sqrt{n}}) \nonumber \\
&=&\mathbb{E}_{x_o \sim \pi^{**},  z \sim \underline{\phi}} \left\|  \frac{x_o - \mathbb{E}_{\pi^{*}} (x_o \ | \ y= x_o+\sigma z)}{1- \mathbb{P} (A^c)} \right\|_2^2 \nonumber \\
&& + O(n {\rm e}^{- \sqrt{n}}) \nonumber \\
&=&\mathbb{E}_{x_o \sim \pi^{**},  z \sim \underline{\phi}} \| {x_o - \mathbb{E}_{\pi^{*}} (x_o \ | \ y= x_o+\sigma z)}\|_2^2  \nonumber \\
&&+ O(n {\rm e}^{- \sqrt{n}}) \nonumber \\
&=&\mathbb{E}_{x_o \sim \pi^{*},  z \sim \underline{\phi}} \| {x_o - \mathbb{E}_{\pi^{*}} (x_o \ | \ y= x_o+\sigma z)}\|_2^2 \nonumber \\
&&+ O(n {\rm e}^{- \sqrt{n}}).
\end{eqnarray}
Define $\phi_\sigma(z) = \phi(z/\sigma)$, $\tilde{\gamma} = k/n-\gamma$ and $\bar{\tilde{\gamma}} = 1-k/n+\gamma$. Note that since the prior we defined on $x_o$, i.e., $\pi^*$ is a product of similar measure on the individual $x_{o,i}$ we conclude that
\begin{eqnarray}
\lefteqn{\mathbb{E}_{x_o \sim \pi^{*},  z \sim \underline{\phi}} \| {x_o - \mathbb{E}_{\pi^{*}} (x_o \ | \ y= x_o+\sigma z)}\|_2^2} \nonumber\\
 &=& n \mathbb{E}_{x_{o,1} \sim \pi^{*}_1,  z_i \sim \phi} ( {x_{o1} - \mathbb{E}_{\pi^{*}_1} (x_{o1} \ | \ y_1= x_{o1}+\sigma z_1)})^2\nonumber\\
 &=& \mathbb{E}_{z_1 \sim \phi} \left( \frac{(k/n - \gamma) \phi_{\sigma}(z_1)}{(\tilde{\gamma}) \phi_\sigma(z_1) +(\bar{\tilde{\gamma}}) \phi_{\sigma}(z_1+1) }-1 \right)^2(\tilde{\gamma}) \nonumber \\ 
 &&+ \mathbb{E}_{z_1 \sim \phi}  \left( \frac{(k/n - \gamma) \phi_{\sigma}(z_1-1)}{(\tilde{\gamma}) \phi_\sigma(z_1-1) +(\bar{\tilde{\gamma}}) \phi_{\sigma}(z_1) } \right)^2(\bar{\tilde{\gamma}}). \nonumber 
\end{eqnarray}
Finally, by the dominated convergence theorem we prove that
\begin{eqnarray}
\lefteqn{\lim_{n \rightarrow \infty} \lefteqn{\mathbb{E}_{x_o \sim \pi^{*},  z \sim \underline{\phi}} \| {x_o - \mathbb{E}_{\pi^{*}} (x_o \ | \ y= x_o+\sigma z)}\|_2^2}}\nonumber \\
 &=& \mathbb{E}_{z_1 \sim \phi} \left( \frac{\rho \phi_{\sigma}(z_1)}{\rho \phi_\sigma(z_1) +(1-\rho) \phi_{\sigma}(z_1+1) }-1 \right)^2\rho \nonumber \\ 
 &&+ \mathbb{E}_{z_1 \sim \phi}  \left( \frac{(\rho) \phi_{\sigma}(z_1-1)}{\rho \phi_\sigma(z_1-1) +(1-\rho) \phi_{\sigma}(z_1) } \right)^2(1-\rho). \nonumber
\end{eqnarray}

\FloatBarrier

\bibliographystyle{ieeetr}
\bibliography{myrefs}

\begin{IEEEbiographynophoto}
	{Christopher A.\ Metzler} 
is a Ph.D. student in the Department of Electrical and Computer Engineering at Rice University.  He received his M.S.E.E. and B.S.E.E. degrees from Rice in 2014 and 2013, respectively.  He is currently supported by an NSF Graduate Research Fellowship and was previously supported by a DoD NDSEG fellowship. His current research focuses on the application of signal processing to imaging and communication systems.
\end{IEEEbiographynophoto}
\begin{IEEEbiographynophoto}
	{Arian Maleki} 
is an assistant professor in the Department of Statistics at Columbia University. He received Ph.D. from Stanford University in 2010. Before joining Columbia University, he was a postdoctoral scholar in the department of Electrical and Computer Engineering at Rice University.

\end{IEEEbiographynophoto}
\begin{IEEEbiographynophoto}
	{Richard G. Baraniuk} 
is the Victor E.\ Cameron Professor of Electrical and Computer Engineering at Rice University and the Founder and Director of OpenStax.   His research interests lie in new theory, algorithms, and hardware for sensing, signal processing, and machine learning. He is a Fellow of the IEEE and AAAS and has received national young investigator awards from the US NSF and ONR, the Rosenbaum Fellowship from the Isaac Newton Institute of Cambridge University, the ECE Young Alumni Achievement Award from the University of Illinois, the Wavelet Pioneer and Compressive Sampling Pioneer Awards from SPIE, the IEEE Signal Processing Society Best Paper Award, and the IEEE Signal Processing Society Technical Achievement Award.  His work on the Rice single-pixel compressive camera has been widely reported in the popular press and was selected by MIT Technology Review as a TR10 Top 10 Emerging Technology.  For his teaching and education projects, including Connexions (cnx.org) and OpenStax College (openstaxcollege.org), he has received the C. Holmes MacDonald National Outstanding Teaching Award from Eta Kappa Nu, the Tech Museum of Innovation Laureate Award, the Internet Pioneer Award from the Berkman Center for Internet and Society at Harvard Law School, the World Technology Award for Education, the IEEE-SPS Education Award, the WISE Education Award, and the IEEE James H. Mulligan, Jr. Medal for Education.  He is a Capricorn.
	
\end{IEEEbiographynophoto}

\end{document}